\documentclass[10pt]{article}

\usepackage{natbib}
\usepackage[pdftex, final]{graphicx}
\DeclareGraphicsExtensions{.png,.pdf,.jpg}
\usepackage{enumerate}
\usepackage{url} 
\usepackage{soul}
\usepackage[scriptsize]{subfigure}
\usepackage{amsmath}
\usepackage{bm}
\usepackage{bbm}
\usepackage{verbatim}
\usepackage{amssymb,amsthm,amsmath}
\usepackage{dsfont}
\usepackage{mathrsfs}

\usepackage{authblk}
\usepackage{setspace}

\def\E{\mathds{E}}
\def\G{\mathcal{G}}

\def\P{\mathds{P}}
\def\R{\mathds{R}}

\def\de{\mathrm{d}}

\def\CRM{\mathrm{CRM}}

\def\ind{\mathds{1}}

\def\Bern{\mathrm{Bernoulli}}
\def\Beta{\mathrm{Beta}}
\def\NegBin{\mathrm{NegativeBinonial}}
\def\Gammad{\mathrm{Gamma}}

\newcommand{\largestjump}[2]{\Delta_{#1,#2}}

\newcommand{\stablenews}[2]{\gamma_{#1}^{(#2)}}
\newcommand{\stablenewsfreq}[3]{\rho_{#1}^{(#2,#3)}}

\newcommand{\predaccuracy}[2]{v_{#1}^{(#2)}}
\newcommand{\predaccuracyfreq}[3]{v_{#1}^{(#2,#3)}}
\newcommand{\unseen}[2]{U_{#1}^{(#2)}}
\newcommand{\unseenfreq}[3]{U_{#1}^{(#2,#3)}}
\newcommand{\estpred}[2]{\hat{U}_{#1}^{(#2)}} 
\newcommand{\estpredfreq}[3]{\hat{U}_{#1}^{(#2, #3)}} 

\newtheorem{lemma}{Lemma}
\newtheorem{theorem}{Theorem}
\newtheorem{proposition}{Proposition}
\newtheorem{definition}{Definition}

\newtheorem{remark}{Remark}
\def\simind{\stackrel{\mbox{\scriptsize{ind}}}{\sim}}
\def\simiid{\stackrel{\mbox{\scriptsize{iid}}}{\sim}}

\usepackage{mathtools}
\usepackage[colorlinks=true,linkcolor=blue,citecolor=blue]{hyperref}
\usepackage{cleveref}

\usepackage[top=1.2in,bottom=1.2in,left=1.2in,right=1.2in]{geometry}

\begin{document}

\title{\bf {\Large{ Scaled process priors for Bayesian nonparametric estimation of the unseen genetic variation}}}

\author[,1]{Federico Camerlenghi \thanks{Also affiliated to Collegio Carlo Alberto, Piazza V. Arbarello 8, Torino, and BIDSA, Bocconi University, Milano, Italy; federico.camerlenghi@unimib.it}}
\author[,2]{Stefano Favaro \thanks{Also affiliated to Collegio Carlo Alberto, Piazza V. Arbarello 8, Torino, and IMATI-CNR ``Enrico  Magenes", Milan, Italy; stefano.favaro@unito.it}}
\author[,3]{Lorenzo Masoero \thanks{lom@mit.edu}}
\author[,3]{Tamara Broderick \thanks{tbroderick@csail.mit.edu}}

\affil[1]{Department of Economics, Management and Statistics, University of Milano - Bicocca, Piazza dell'Ateneo Nuovo 1, Milano}
\affil[2]{Department of Economics and Statistics, University of Torino, Corso Unione Sovietica 218/bis, Torino}
\affil[3]{Department of Electrical Engineering and Computer Science, CSAIL, Massachusetts Institute of Technology, Cambridge, Massachusetts, USA}

\date{}
\maketitle
\thispagestyle{empty}

\setcounter{page}{1}

\begin{abstract}
There is a growing interest in the estimation of the number of unseen features, mostly driven by biological applications. A recent work brought out a peculiar property of the popular completely random measures (CRMs) as prior models in Bayesian nonparametric (BNP) inference for the unseen-features problem: for fixed prior's parameters, they all lead to a Poisson posterior distribution for the number of unseen features, which depends on the sampling information only through the sample size. CRMs are thus not a flexible prior model for the unseen-features problem and, while the Poisson posterior distribution may be appealing for analytical tractability and ease of  interpretability, its independence from the sampling information makes the BNP approach a questionable oversimplification, with posterior inferences being completely determined by the estimation of unknown prior's parameters. In this paper, we introduce the stable-Beta scaled process (SB-SP) prior, and we show that it allows to enrich the posterior distribution of the number of unseen features arising under CRM priors, while maintaining its analytical tractability and interpretability. That is, the SB-SP prior leads to a negative Binomial posterior distribution, which depends on the sampling information through the sample size and the number of distinct features, with corresponding estimates being simple, linear in the sampling information and computationally efficient. We apply our BNP approach to synthetic data and to real cancer genomic data, showing that: i) it outperforms the most popular parametric and nonparametric competitors in terms of estimation accuracy; ii) it provides improved coverage for the estimation with respect to a BNP approach under CRM priors.
\end{abstract}



\section{Introduction} \label{sec:intro}

The problem of estimating the number of unseen features generalizes the popular unseen-species problem \citep{orlitsky2016optimal}, and its importance has grown dramatically in recent years, driven by applications in biological sciences \citep{ionita2009estimating,gravel2014predicting,zou2016quantifying,chakraborty2019somatic}. Consider a generic population in which each individual is endowed with a finite collection of $\mathds{W}$-valued features, with $\mathds{W}$ possibly being an infinite space, and denote by $p_{i}$ the probability that an individual has feature $w_{i}\in\mathds{W}$ for $i\geq1$. The unseen-features problem assumes $N\geq1$ observable random samples $Z_{1:N}=(Z_{1},\ldots,Z_{N})$ from the population, such that $Z_{n}=(A_{n,i})_{i\geq1}$ are independent Bernoulli random variables with unknown parameters $(p_{i})_{i\geq1}$. Then, the goal is to estimate the number of hitherto unseen features that would be observed if $M\geq1$ additional samples were collected, i.e.
\begin{displaymath}
U=\sum_{i \geq 1} \ind\left(\sum_{n=1}^N A_{n,i} = 0\right)\ind\left(\sum_{m=1}^M A_{N+m,i} > 0\right),
\end{displaymath}
with $\ind$ being the indicator function. The unseen-species problem arises under the assumption that each individual is endowed with only one feature, i.e. a species. A wide range of approaches have been developed to estimate $U$, including Bayesian methods \citep{ionita2009estimating,masoero2019more}, jackknife \citep{gravel2014predicting}, linear programming \citep{zou2016quantifying}, and variations of Good-Toulmin estimators \citep{orlitsky2016optimal,chakraborty2019somatic}. 

In biological sciences, we may think of individuals as organisms and of features as groups to which organisms belong to, with each group being defined by any difference in the genome relative to a reference genome, i.e.\ a (genetic) variant. In human biology, the estimation of $U$ arises in the context of optimal allocation of resources between quantity and quality in genetic experiments: spending resources to sequence a greater number of genomes (quantity), which reveals more about variation across the population, or spending resources to sequence genomes with increased accuracy (quality), which reveals more about individual organisms' genomes. Accurate estimates of $U$ are critical in the experimental pipeline towards the goal of maximizing the usefulness of experiments under the trade-off between quantity and quality \citep{ionita2010optimal,zou2016quantifying}. While in human-biology the cost of  sequencing has decreased in recent years \citep{schwarze2020complete}, the expense remains non-trivial, and it is still critical in fields where scientists work with relatively budgets, e.g.\ non-human and non-model organisms \citep{souza2017efficiency}. Other applications arise in precision medicine \citep{momozawa2020unique}, microbiome analysis \citep{sanders2019optimizing}, single-cell sequencing \citep{zhang2020determining} and wildlife monitoring \citep{johansson2020identification}.

\subsection{Our contributions}

We introduce a Bayesian nonparametric (BNP) approach to the unseen-features problem, which relies on a novel prior distribution for the unknown $(p_{i})_{i\geq1}$. Completely random measures (CRMs) \citep{kingman1992poisson} provide a broad class of nonparametric priors for feature sampling problems, the most popular being the stable-Beta process prior \citep{james2017bayesian,broderick2018posteriors}. In a recent work, \citet{masoero2019more} brought out a peculiar feature of CRM priors in the unseen-features problem: they all lead to a Poisson posterior distribution of $U$, given $Z_{1:N}$ and fixed prior's parameters, which depends on $Z_{1:N}$ only through the sample size $N$. Despite the broadness of the class of CRM priors, such a common Poisson posterior structure makes CRMs not a flexible prior model for the unseen-features problem. While the Poisson posterior distribution may be appealing in principle, making posterior inferences analytically tractable and easy to interpret, its independence from $Z_{1:N}$ makes the BNP approach a questionable oversimplification, with posterior inferences being completely determined by the estimation of the unknown prior's parameters.  A somehow similar scenario occurs in BNP inference for the unseen-species problem under a Dirichlet process (DP) prior \citep{ferguson1973bayesian}, and led to the use of the Pitman-Yor process (PYP) prior \citep{pityor_97} for enriching the posterior distribution of the number of unseen species, while maintaining analytical tractability and interpretability of the DP prior \citep{Lijoi_07}.

We show that scaled process (SP) priors, first introduced in \citet{james2015scaled}, allow to enrich the posterior distribution of $U$ arising under CRM priors. Under SP priors, we characterize the posterior distribution of $U$ as a mixture of Poisson distributions that may include, through the mixing distribution, the whole sampling information in terms of the number of distinct features and their frequencies. While this is appealing in principle, it may be at stake with analytical tractability and interpretability, which are critical for a concrete use of SP priors. Then, we introduce the stable-Beta SP (SB-SP) prior, which provides a sensible trade-off between the amount of sampling information introduced in the posterior distribution of $U$, and analytical tractability and interpretability of the posterior inferences. In particular, we characterize the SB-SP prior as the sole SP prior for which the posterior distribution of $U$, given $Z_{1:N}$ and fixed prior's parameters, depends on $Z_{1:N}$ through the sample size $N$ and the number $K_{N}$ of distinct features; the SB-SP may thus be considered as the natural counterpart of the PYP for the unseen-feature problem. Under the SB-SP prior, the posterior distribution of $U$, as well as of a refinement of $U$ that deals with the number of unseen rare features, is a negative Binomial posterior distributions, whose parameters depend on $N$, $K_{N}$ and the prior's parameters. Corresponding Bayesian estimates of $U$, with respect to a squared loss function, are simple, linear in $K_{N}$ and computationally efficient.

We present an empirical validation of the effectiveness of our BNP methodology, both on synthetic and real data. As for real data, we consider cancer genomic data, where the goal is to estimate the number of new (genomic) variants to be discovered in future unobservable samples. In cancer genomics, accurate estimates of the number of new variants is of particular importance, as it might help practitioners understand the site of origin of cancers, as well as the clonal origin of metastasis, and in turn be a useful tool to develop effective clinical strategies \citep{chakraborty2019somatic,huyghe2019discovery}. We make use of data from the cancer genome atlas (TCGA), and focus on the challenging scenario in which the sample size $N$ is particularly small, and also small with respect to the extrapolation size $M$. Such a scenario is of interest in genomic applications, where only few samples of rare cancer might be available. We show that our BNP methodology outperforms the most popular parametric and nonparametric competitors, both classical (frequentist) and Bayesian, in terms of estimation accuracy of $U$ and a refinement of $U$ for rare features. In addition, with respect to the BNP approach under the stable-Beta process prior \citep{masoero2019more}, our approach provides improved coverage for the estimation. This is an empirical evidence of the effectiveness of replacing the Poisson posterior distribution with the negative Binomial posterior distribution, which allows to better exploit the sampling information.

\subsection{Organization of the paper}

In Section \ref{sec:theory} we show how SP priors allow to enrich the posterior distribution of $U$ arising under CRM priors. In Section \ref{sec:theory1} we introduce and investigate the SB-SP prior in the context of the unseen-features problem: i) we characterize the SB-SP prior in the class of SP priors, providing its predictive distribution; ii) we apply the SB-SP prior to the unseen-features problem, providing the posterior distribution of $U$ and a BNP estimator. Section \ref{sec:exp} contains illustrations of our  method. In Section \ref{sec:discussion} we discuss our approach, a multivariate extension of it, and  future research directions. Proofs and additional experiments are in the Appendix.


\section{Scaled process priors for feature sampling problems} \label{sec:theory}

For a measurable space of features $\mathds{W}$, we assume $N\geq1$ observable individuals to be modeled as a random sample $Z_{1:N}$ from the $\{0,1\}$-valued stochastic process $Z(w)  = \sum_{i \ge 1} A_{i}\delta_{w_i}(w)$, $w\in\mathds{W}$, where $(w_{i})_{i\geq1}$ are features in $\mathds{W}$ and $(A_{i})_{i\geq1}$ are independent Bernoulli random variables with unknown parameters $(p_{i})_{i\geq1}$, $p_{i}$ being the probability that an individual has feature $w_{i}$, for $i\geq1$. That is, $Z$ is a Bernoulli process with parameter $\zeta=\sum_{i\geq1}p_{i}\delta_{w_{i}}$, denoted as $\text{\rm BeP}(\zeta)$. BNP inference for feature sampling problems relies on the specification of a prior distribution on the discrete measure $\zeta$, leading to the BNP-Bernoulli model,
\begin{align}\label{exch_mod}
Z_n\,|\,\zeta & \quad\simiid\quad \text{\rm BeP}(\zeta) \qquad n=1,\ldots,N,\\
\notag \zeta& \quad\sim\quad\mathscr{Z},
\end{align}
namely $\zeta$ is a discrete random measure on $\mathds{W}$ whose law $\mathscr{Z}$ takes on the interpretation of a prior distribution for the unknown feature's composition of the population. By de Finetti's theorem, the random variables $Z_{n}$'s in \eqref{exch_mod} are exchangeable with directing measure $\mathscr{Z}$ \citep{Aldous_83}.
In this section, we show how SP priors for $\zeta$ \citep{james2015scaled} allow to enrich the posterior distribution of the number of unseen features arising under CRM priors.

\subsection{CRM priors for Bernoulli processes}

CRMs provide a standard tool to define nonparametric prior distributions on the parameter $\zeta$ of the Bernoulli process $Z$. Consider a homogeneous CRM $\mu_{0}$ on $\mathds{W}$, i.e.\ $\mu_{0}~=~\sum_{i\geq1}\rho_{i}\delta_{W_{i}}$, where the $\rho_{i}$'s are $(0,1)$-valued random atoms such that $\sum_{i\geq1}\rho_{i}<+\infty$, while the $W_{i}$'s are i.i.d.\ $\mathds{W}$-valued random locations independent of the $\rho_{i}$'s. The law of $\mu_{0}$ is characterized, through L\'evy-Khintchine formula, by the L\'evy intensity measure $\nu_{0} (\de s , \de w)= \lambda_{0} (s) \de s P (\de w)$ on $(0,1) \times \mathds{W}$, where: i) $\lambda_{0}$ is a measure on $(0,1)$,  which controls the distribution of the $\rho_{i}$'s, and such that $\int_{(0,1)} \min \{ s,1 \} \lambda_{0} (s) \de s < +\infty$; ii) $P$ is a non-atomic measure on $\mathds{W}$, which controls the distribution of the $W_{i}$'s. For short, $\mu_0 \sim \CRM(\nu_{0})$. See Appendix \ref{appendix_crm} for an account on CRMs \citep[Chapter 8]{kingman1992poisson}. Note that, since $P$ is non-atomic, the random atoms $W_{i}$'s are almost surely distinct, that is to say the different features cannot coincide almost surely. The law of $\mu_{0}$ provides a natural prior distribution for the parameter $\zeta$ of the Bernoulli process. The Beta and the stable-Beta processes are popular examples of $\mu_0 \sim \CRM(\nu_{0})$, for suitable specifications of $\nu_{0}$. A comprehensive posterior analysis of CRM priors is presented in  \cite{james2017bayesian}. In the next proposition, we recall the predictive distribution of CRM priors \citep[Proposition 3.2]{james2017bayesian}.

\begin{proposition} \label{prop:predictive_IBP}
Let $Z_{1:N}$ be a random sample from \eqref{exch_mod} with $\zeta \sim \CRM(\nu_{0})$. If $Z_{1:N}$ displays $K_{N}=k$ distinct features $\{W_{1}^{\ast},\ldots,W^{\ast}_{K_{N}}\}$, each feature $W_i^*$ appearing exactly $M_{N,i}=m_i$ times, then the conditional distribution of $Z_{N+1}$, given $Z_{1:N}$, coincides with the distribution of
\begin{equation}\label{eq:predictive_IBP}
Z_{N+1}\,|\, Z_{1:N} \stackrel{d}{=} Z_{N+1}^{\prime}+\sum_{i=1}^{K_N} A_{N+1,i} \delta_{W_{i}^{\ast}},
\end{equation}
where: i) $Z_{N+1}^{\prime} \,|\, \mu_{0}^{\prime}= \sum_{i \geq 1} A_{N+1,i}^{\prime} \delta_{W_{i}^{\prime}}\sim\text{\rm BeP}(\mu_{0}^{\prime})$ and $\mu_{0}^{\prime}\sim\CRM(\nu_{0}^{\prime})$, with $\nu^{\prime}_{0} (\de s , \de w) =(1-s)^N \lambda_0 (s) \de s P(\de w)$; ii) the $A_{N+1,i}$'s are independent Bernoulli random variables with parameters $J_{i}$'s, such that $J_{i}$ is distributed according to the density function $f_{J_i} (s)\propto (1-s)^{N-m_{i}} s^{m_{i}} \lambda_0 (s)$ for $i\geq1$.
\end{proposition}

According to \eqref{eq:predictive_IBP}, $Z_{N+1}$ displays ``new" features $W_{i}^{\prime}$'s, i.e.\ features not appearing in the initial sample $Z_{1:N}$, and ``old" features $W_{i}^{\ast}$'s, i.e.\ features appeared in the initial sample $Z_{1:N}$. The posterior distribution of statistics of ``new" features is determined by the law of $Z_{N+1}^{\prime}$, which depends on $Z_{1:N}$ only through the sample size $N$; the posterior distribution of statistics of ``old" features is determined by the law of $\sum_{1\leq i\leq K_{N}}A_{N+1, i} \delta_{W_{i}^{\ast}}$, which depends on  $Z_{1:N}$ through the sample size $N$, the number $K_{N}$ of distinct features and their frequencies $(M_{N,1},\ldots,M_{N,K_{N}})$. As a corollary of Proposition \ref{prop:predictive_IBP}, the posterior distribution of the number of ``new" features in $(Z_{N+1},\ldots,Z_{N+M})$, given $Z_{1:N}$ and fixed prior's parameters, is a Poisson distribution that depends on $Z_{1:N}$ only through $N$ \citep{masoero2019more}. Such a posterior structure is peculiar to CRM priors, being inherited by the Poisson process formulation of CRMs \citep{kingman1992poisson}. That is, despite the broadness of the class of CRM priors, all CRM priors lead to the same Poisson posterior structure for the number of unseen features, which thus makes them not a flexible prior model for the unseen-features problem. While the Poisson posterior distribution may be appealing in principle, making the posterior inferences analytically tractable and of easy interpretability, its independence from $Z_{1:N}$ makes the BNP approach under CRM priors a questionable oversimplification, with posterior inferences being completely determined by the estimation of unknown prior's parameters.

\begin{remark}
For the sake of mathematical convenience, and in agreement with the work of \cite{james2017bayesian}, in the sequel we maintain the random measure formulation for both the prior model $\mu_{0}$ and the Bernoulli processes $Z_n$. However, we point out that each $Z_n$ is equivalently characterized by means of the Bernoulli variables $(A_{n,i})_{i \geq 1}$ and the random features $(W_i)_{i \geq 1}$. In other terms, there exits a one-to-one correspondence between $Z_n$ and the sequence of points $\{(A_{n,i}, W_i)\}_{i \geq 1}$. Finally, note that, although the values of features' labels $W_i$ are immaterial, the features $W_i$'s are assumed to be random. This is in line with the BNP literature on species sampling models, where the species' labels are assumed to be random \citep{Pitman1996}.
\end{remark}

\subsection{SP priors for Bernoulli processes}  \label{sec:SP_prior}

Consider a homogeneous CRM $\mu=\sum_{i\geq1}\tau_{i}\delta_{W_{i}}$ on $\mathds{W}$, where the $\tau_{i}$'s are non-negative and such that $\sum_{i\geq1}\tau_{i}<+\infty$, and the $W_{i}$'s are i.i.d. and independent of the $\tau_{i}$'s. We denote by $ \nu (\de s , \de w)= \lambda (s) \de s P (\de w)$ on $\R_{+} \times \mathds{W}$, with $\int_{\R_+} \min \{ s,1 \} \lambda (s) \de s < +\infty$, the L\'evy intensity measure of $\mu$. Let $\Delta_{1}>\Delta_{2}>\ldots$ be the decreasingly ordered $\tau_{i}$'s, and consider the discrete random measure
\begin{align*}
\mu_{\Delta_{1}} = \sum_{i\geq 1}\frac{\Delta_{i+1}}{\Delta_{1}}\delta_{W_{i+1}},
\end{align*}
such that $\Delta_{i+1}/\Delta_{1}\in(0,1)$, for $i\geq1$, and $\sum_{i\geq1}\Delta_{i+1}/\Delta_{1}<+\infty$. A SP on $\mathds{W}$ is defined from $\mu_{\Delta_{1}}$ as follows. Let $F_{\Delta_1}(\de a) = \exp\left\{-\int_{a}^{\infty} \lambda(s)\de s\right\} \lambda(a)\de a$ be the distribution of $\Delta_{1}$ \citep[pg. 1636]{FergusonKlass_1972}, and let $G_{a}$ be the conditional distribution of $(\Delta_{i+1}/\Delta_{1})_{i\geq1}$ given $\Delta_{1}=a$. Moreover, let $\Delta_{1,h}$ denote a  random variable whose distribution has a density function $f_{\Delta_{1,h}}(a)=h(a)f_{\Delta_{1}}(a)$, where $h$ is a non-negative function and $f_{\Delta_{1}}$ is the density function of $F_{\Delta_1}$. If $(\rho_{i})_{i\geq1}$ are $(0,1)$-valued random variables with distribution $G_{\Delta_{1,h}}$ then
\begin{equation} \label{eq:SP_random}
\mu_{\Delta_{1,h}} = \sum_{i\geq 1}\rho_{i}\delta_{W_{i+1}}.
\end{equation}
is a SP. For short, $\mu_{\Delta_{1,h}}\sim\text{ \rm SP}(\nu,h)$. The law of $\mu_{\Delta_{1,h}}$ is a prior distribution for the parameter $\zeta$ of the Bernoulli process. The next proposition characterizes the predictive distribution of SP priors. See also \citet[Proposition 2.2]{james2015scaled} for a posterior analysis of SP priors.

\begin{proposition} \label{prop:general_pred}
Let $Z_{1:N}$ be a random sample from \eqref{exch_mod} with $\zeta\sim\text{\rm SP}(\nu,h)$. If $Z_{1:N}$ displays $K_{N}=k$ distinct features $\{W_{1}^{\ast},\ldots,W^{\ast}_{K_{N}}\}$, each feature $W_i^*$ appearing exactly $M_{N,i}=m_i$ times, then the conditional distribution of $\Delta_{1,h}$, given $Z_{1:N}$, has a density function of the form
\begin{equation}\label{eq_mixing1}
  g_{\Delta_{1,h}\,|\, Z_{1:N}} (a) \propto   \frac{\prod_{i=1}^{k}\int_0^1  s^{m_{i}} (1-s)^{N-m_{i}} a\lambda (a s ) \de s}{\exp\left\{\sum_{n=1}^N \int_{0}^{1}s(1-s)^{n-1}a\lambda(as)\de s\right\}}f_{\Delta_{1,h}}(a).
\end{equation}
Moreover, the conditional distribution of $Z_{N+1}$, given $(\Delta_{1,h}, Z_{1:N})$, coincides with the distribution of
\begin{equation} \label{eq:general_predictive1}
      Z_{N+1}\,|\, (\Delta_{1,h}, Z_{1:N})  \overset{d}{=} Z_{N+1}^{\prime} + \sum_{i=1}^{K_N} A_{N+1, i} \delta_{W_i^{\ast}},
\end{equation}
where: i) $Z_{N+1}^{\prime}\,|\, \mu^{\prime}_{\Delta_{1,h}}= \sum_{i \geq 1} A_{N+1, i}^{\prime} \delta_{W_i^{\prime}}\sim\text{\rm BeP}(\mu^{\prime}_{\Delta_{1,h}})$ and $\mu^{\prime}_{\Delta_{1,h}}\,|\,\Delta_{1,h}\sim\CRM(\nu^{\prime}_{\Delta_{1,h}})$, with $\nu^{\prime}_{\Delta_{1,h}}(\de s,\,\de w)=(1-s)^N \Delta_{1,h} \lambda(s \Delta_{1,h}) \ind_{(0,1)} (s) \de sP(\de w )$; ii)  the $A_{N+1, i}$'s are independent Bernoulli random variables with parameters $J_{i}$'s, respectively, such that $J_{i}\,|\,\Delta_{1,h}$ is distributed according to the density function $f_{J_{i}\,|\,\Delta_{1,h}} (s)  \propto  (1-s)^{N-m_{i}} s^{m_{i}} \Delta_{1,h}\lambda (\Delta_{1,h} s) \ind_{(0,1)} (s) \de s$ for $i\geq1$.
\end{proposition}

See Appendix \ref{appendix_scaledpriors} for the proof of Proposition \ref{prop:general_pred}. The marginalization of \eqref{eq:general_predictive1} with respect to \eqref{eq_mixing1} leads to the predictive distribution of SP priors: i) $Z_{N+1}$ displays ``new" features $W_{i}^{\prime}$'s, and the posterior distribution of statistics of ``new" features, given $Z_{1:N}$, is determined by the law of $(\Delta_{1,h},Z^{\prime}_{N+1})$; ii) $Z_{N+1}$ displays ``old" features $W_{i}^{\ast}$'s, and the posterior distribution of statistics of ``old" features, given $Z_{N+1}$, is determined by the law of $(\Delta_{1,h},\sum_{1\leq i\leq K_{N}}A_{N+1, i} \delta_{W_{i}^{\ast}})$. Because of \eqref{eq_mixing1} and \eqref{eq:general_predictive1}, the law of $(\Delta_{1,h},Z^{\prime}_{N+1})$ may include the whole sampling information, depending on the specification of $\nu$ and $h$, and hence the posterior distribution of statistics of ``new" features, given $Z_{1:N}$, also includes such an information. As a corollary of Proposition \ref{prop:general_pred}, the posterior distribution of the number of unseen features, given $Z_{1:N}$ and fixed prior's parameters, is a mixture of Poisson distributions that may include the whole sampling information; in particular, the amount of sampling information in the posterior distribution is uniquely determined by the mixing distribution, namely by the conditional distribution of $\Delta_{1,h}$, given $Z_{1:N}$. SP priors thus allow to enrich the Poisson posterior structure arising from CRM priors, in terms of both a more flexible distribution and the inclusion of more sampling information than the sole sample size $N$, though they may lead to unwieldy posterior inferences due to the marginalization with respect to \eqref{eq_mixing1}. 

The use of the sampling information in the predictive structure of SPs somehow resembles that of Poisson-Kingman (PK) models \citep{Pitman_06}. PK models form a broad class of nonparametric priors for species sampling problems. The DP prior is a PK model whose predictive distribution is such that: i) the conditional probability that the $(N+1)$-th draw is a ``new" species, given $N$ observable samples, depends only on the sample size; ii) the conditional probability that the $(N+1)$-th draw is an ``old" species, given $N$ observable samples, depends on the sample size, the number of distinct species and their frequencies.  Such a behaviour resembles that of CRM priors, i.e. Proposition \ref{prop:predictive_IBP}. PK models allow to include more sampling information in the probability of discovering a ``new species" arising under the DP prior, which typically determines a loss of the analytical tractability of posterior inferences for the number of unseen species \citep{marcoB}. Such a behaviour resembles that of SP priors, i.e. Proposition \ref{prop:general_pred}. The PYP prior is arguably the most popular PK model. It stands out for enriching the probability of discovering a ``new" species arising under the DP prior, by including the sampling information on the number of distinct species, while maintaining the analytical tractability and interpretability of the DP prior. 


\section{Stable-Beta Scaled Process (SB-SP) priors for the unseen-features problem}\label{sec:theory1}

In Section \ref{sec:theory} we showed how SP priors allow to enrich the Poisson posterior structure of the number of unseen features arising under CRM priors, e.g.\ the Beta and the stable-Beta process priors. While this is an appealing property, it may lead to a lack of analytical tractability and interpretability of posterior inferences, thus making SP priors not of practical interest in applications. In this section, we introduce and investigate a peculiar SP prior, which is referred to as the SB-SP prior, and we show that: i) it leads to a negative Binomial posterior distribution for the number of unseen features, which generalizes the Poisson distribution while maintaining its analytical tractability and interpretability; ii) it leads to a posterior distribution for the number of unseen features, which depends on the sampling through the sample size and the number of distinct features. The SB-SP prior thus provides a sensible trade-off between the enrichment of the Poisson posterior structure of the number of unseen features arising under CRM priors and the analytical tractability and interpretability of posterior inferences. In particular, we characterize the SB-SP prior as the sole SP prior for which the posterior distribution of the number of unseen features depends on the observable sample only through the sample size and the number of distinct features. The SB-SP may thus be considered as a natural counterpart of the PYP for the unseen-feature problem.

\subsection{SB-SP priors for Bernoulli processes}\label{sec:stable}

Stable scaled processes (S-SP) \citep{james2015scaled} form a subclass of SPs, and hence their definition follows from Section \ref{sec:theory}. In particular, for any $\sigma\in(0,1)$, let $\mu_{\sigma}$ be the $\sigma$-stable CRM on $\mathds{W}$ \citep{kingman1975random}, which is characterized by the L\'evy intensity measure $ \nu_{\sigma} (\de s , \de w)= \lambda_{\sigma} (s) \de s P (\de w)$ on $\R_{+} \times \mathds{W}$, with $\int_{\R_+} \min \{ s,1 \} \lambda_{\sigma} (s) \de s < +\infty$, where $\lambda_{\sigma}(s)=\sigma s^{-1-\sigma}$. We recall that the largest atom $\Delta_{1}$ of $\mu_{\sigma}$ is distributed according to the density function
\begin{equation}\label{eq:stable_largest_jump}
f_{\Delta_1} (a)= \sigma a^{-1-\sigma} \exp\left\{-a^{-\sigma}\right\}.
\end{equation}
That is, $\Delta_{1}=E^{-1/\sigma}$, where $E$ denotes a negative exponential random variable with parameter $1$. For any non-negative function $h$, a S-SP on $\mathds{W}$ is defined as the SP with law $\text{\rm SP}(\nu_{\sigma},h)$. S-SP priors generalizes the Beta process prior, which is recovered by setting $h$ to be the identity function, and then letting $\sigma\rightarrow0$ \citep{james2015scaled}. The predictive distribution of $\zeta\sim\text{\rm SP}(\nu_{\sigma},h)$ is obtained from Proposition \ref{prop:general_pred}. In the next theorem, we characterize the S-SP priors as the sole SP priors for which the conditional distribution of $\Delta_{1,h}$, given $Z_{1:N}$, depends on $Z_{1:N}$ only through the sample size $N$ and the number $K_{N}$ of distinct features in $Z_{1:N}$.

\begin{theorem}\label{thm:characterization}
Let $Z_{1:N}$ be a random sample from \eqref{exch_mod} with $\zeta\sim\text{\rm SP}(\nu,h)$, and let $Z_{1:N}$ displays $K_{N}$ distinct features with corresponding frequencies $(M_{N,1},\ldots,M_{N,K_{N}})$. Moreover, let $\nu(\de s , \de w)=\lambda(s)\de sP(\de w)$, and let $f_{\Delta_{1,h}}$ be the density function of $\Delta_{1,h}$. If $f_{\Delta_{1,h}}>0$ on $\R_+$ and the functions $\lambda$ and $f_{\Delta_{1,h}}$ are continuously differentiable, then the conditional distribution of $\Delta_{1,h}$, given $Z_{1:N}$, depends on $Z_{1:N}$ only through $N$ and $K_{N}$ if and only if $\nu=\nu_{\sigma}$.
\end{theorem}

See Appendix \ref{appendix_scaledbetapriors} for the proof of Theorem \ref{thm:characterization}. We recall from Section \ref{sec:theory} that the conditional distribution of $\Delta_{1,h}$, given $Z_{1,N}$, uniquely determines the amount of sampling information included in the posterior distribution of statistics of ``new" features.  Then, according to Theorem \ref{thm:characterization}, S-SP priors are the sole SP priors for which the posterior distribution of the number of unseen features, given  $Z_{1:N}$ and fixed prior's parameters, depends on $Z_{1:N}$ only through $N$ and $K_{N}$. As a corollary of Theorem \ref{thm:characterization}, the Beta process prior is the sole S-SP prior for which the posterior distribution of statistics of ``new" features depends on $Z_{1:N}$ only through $N$. Analogous predictive characterizations are well-known in species sampling problems, and they are typically referred to as ``sufficientness" postulates' \citep{marcoB}. In particular, the DP prior is characterized as the sole species sampling prior for which the conditional probability that the $(N+1)$-th draw is a ``new" species, given $N$ observable samples, depends only on the sample size \citep{Reg78}. Moreover, the PYP prior is characterized as the sole species sampling prior for which the conditional probability that the $(N+1)$-th draw is a ``new" species, given $N$ observable samples, depends only on the sample size and the number of distinct species in the sample \citep{zabell05}. Theorem \ref{thm:characterization} provides a ``sufficientness" postulates' in the context of feature sampling problems.

As a noteworthy example of S-SPs, we introduce the SB-SP. The SB-SP is a S-SP obtained by a suitable specification of the non-negative function $h$. In particular, for any $c,\beta>0$ let
\begin{equation}\label{tilting_h}
h_{c,\beta}(a)=\frac{\beta^{c+1}}{\Gamma(c+1)} a^{-c\sigma} \exp\left\{-(\beta-1) a^{-\sigma} \right\},
\end{equation}
where $\Gamma(\cdot)$ denotes the Gamma function. Then a SB-SP on $\mathds{W}$ is defined as the SP with law $\text{\rm SP}(\nu_{\sigma},h_{c,\beta})$. For short, we denote the law of a SB-SP by $\text{\rm SB-SP}(\sigma,c,\beta)$. The SB-SP prior generalizes the Beta process prior, which is recovered by setting $c=0$ and $\beta=1$, and then letting $\sigma\rightarrow0$. According to the construction of SPs, the distribution of $\Delta_{1,h_{c,\beta}}$ has a density function obtained by combining \eqref{eq:stable_largest_jump} and \eqref{tilting_h}; this is a polynomial-exponential tilting of the density function \eqref{eq:stable_largest_jump}. In particular, $\Delta_{1,h_{c,\beta}}^{-\sigma}$ is distributed as a Gamma distribution with shape $(c+1)$ and rate $\beta$. Such a straightforward distribution for $\Delta_{1,h_{c,\beta}}$ is at the core of the analytical tractability of posterior inferences under the SB-SP prior; this fact will be clear in the application of the SB-SP prior to the problem of estimating the number of unseen features. The next proposition characterizes the predictive distribution of the SB-SP prior.

\begin{proposition} \label{prop:stable_predictive}
Let $Z_{1:N}$ be a random sample from \eqref{exch_mod} with $\zeta\sim\text{\rm SB-SP}(\sigma,c,\beta)$. If $Z_{1:N}$ displays $K_{N}=k$ distinct features $\{W_{1}^{\ast},\ldots,W^{\ast}_{K_{N}}\}$, each feature $W_i^*$ appearing exactly $M_{N,i}=m_i$ times, then the conditional distribution of $\Delta_{1,h_{c,\beta}}$, given $Z_{1:N}$, has a density function of the form
\begin{equation}\label{eq_mixing}
g_{\Delta_{1,h_{c,\beta}}\,|\,Z_{1:N}}(a)=\sigma\frac{(\beta+\gamma_{0}^{(N)})^{k+c+1}}{\Gamma(k+c+1)}a^{-k\sigma-(c+1)\sigma-1}\text{e}^{-a^{-\sigma}(\beta+\gamma_{0}^{(N)})},
\end{equation}
where $\gamma_{0}^{(N)}=\sigma\sum_{1\leq i\leq N}B(1-\sigma,i)$, with $B (\cdot,\cdot)$ being the (Euler) Beta function. Moreover, the conditional distribution of $Z_{N+1}$, given $(\Delta_{1,h_{c,\beta}}, Z_{1:N})$, coincides with the distribution of
\begin{equation} \label{eq:general_predictive}
      Z_{N+1}\,|\, (\Delta_{1,h_{c,\beta}}, Z_{1:N})  \overset{d}{=} Z_{N+1}^{\prime} + \sum_{i=1}^{K_N} A_{N+1, i} \delta_{W_i^{\ast}},
\end{equation}
where:
\begin{itemize}
    \item[i)] $Z_{N+1}^{\prime}\,|\, \mu^{\prime}_{\Delta_{1,h_{c,\beta}}}= \sum_{i \geq 1} A_{N+1, i}^{\prime} \delta_{W_i^{\prime}}\sim\text{\rm BeP}(\mu^{\prime}_{\Delta_{1,h_{c,\beta}}})$ such that $\mu^{\prime}_{\Delta_{1,h_{c,\beta}}}\,$ $|\,\Delta_{1,h_{c,\beta}}\sim\CRM( \nu^{\prime}_{\Delta_{1,h_{c,\beta}}})$, with
\begin{displaymath}
\nu^{\prime}_{\Delta_{1,h_{c,\beta}}}(\de s,\,\de w)=      \Delta_{1,\Delta_{1,h_{c,\beta}}}^{-\sigma} (1-s)^N \sigma s^{-1-\sigma}  \ind_{(0,1)} (s) \de s P(\de w);
\end{displaymath}
        \item[ii)] the $A_{N+1,i}$'s are independent Bernoulli random variables with parameters $J_{i}$'s, respectively, such that each $J_{i}\,|\,\Delta_{1,h_{c,\beta}}$ is distributed according to a density function of the form
    \begin{equation*}
    f_{J_{i}\,|\,\Delta_{1,h_{c,\beta}}} (s)=\frac{1}{B(m_{i}-\sigma,N-m_{i}+1)}  s^{m_{i}-\sigma}(1-s)^{N-m_{i}+1} \ind_{(0,1)} (s) . 
\end{equation*}
\end{itemize}
\end{proposition}

See Appendix \ref{appendix_scaledbetapriors} for the proof of Proposition \ref{prop:stable_predictive}. According to Equation \eqref{eq_mixing}, the conditional distribution of $\Delta_{1,h_{c,\beta}}$, given $Z_{1:N}$, depends on $Z_{1:N}$ only through the sample size $N$ and the number $K_{N}$ of distinct features in $Z_{1:N}$. This agrees with Theorem \ref{thm:characterization}, implying that the posterior distribution of the number of unseen features, given $Z_{1:N}$ and fixed prior's parameters, depends on $Z_{1:N}$ only through $N$ and $K_{N}$. Because of \eqref{eq_mixing} and \eqref{eq:general_predictive}, the posterior distribution of statistics of ``new" features stands out for analytical tractability, thus being competitive with that arising from CRMs, e.g. the Beta and the stable-Beta processes. In particular, from Equation \eqref{eq:general_predictive}, the conditional distribution of $Z^{\prime}_{N+1}$, given $(\Delta_{1,h_{c,\beta}},Z_{1:N})$ is a Poisson distribution that depends on $Z_{1:N}$ only through $N$. Then, from \eqref{eq_mixing}, its marginalization with respect to the conditional distribution of $\Delta_{1,h_{c,\beta}}$, given $Z_{1:N}$, leads to a negative Binomial posterior distribution. Such an appealing property arises from the peculiar form $h_{c,\beta}$ that, combined with $\nu_{\sigma}$, leads to a conjugacy property for the conditional distribution of $\Delta_{1,h_{c,\beta}}$, given $Z_{1:N}$. That is, the conditional distribution of $\Delta_{1,h_{c,\beta}}^{-\sigma}$, given $Z_{1:N}$, is a Gamma distribution with shape $(K_{N}+c+1)$ and rate $\beta+\gamma_{0}^{(N)}$, which is the distribution $\Delta_{1,h_{c,\beta}}^{-\sigma}$ with shape and rate being updated through $Z_{1:N}$. The next proposition establishes the distribution of a random sample $Z_{1:N}$ from a SB-SP prior. See Appendix \ref{appendix_scaledbetapriors} for details. 

\begin{proposition} \label{prop:stable_marginal} 
Let $Z_{1:N}$ be a random sample from \eqref{exch_mod} with $\zeta\sim\text{\rm SB-SP}(\sigma,c,\beta)$. The probability that $Z_{1:N}$ displays a particular feature allocation of $k$ distinct features with frequencies $(m_{1},\ldots,m_{k})$ is
\begin{equation}\label{eq:stable_marginal}
p^{(N)}_{k}(m_{1},\ldots,m_{k})= \frac{\frac{\sigma^{k}\beta^{c+1}}{(\beta+\stablenews{0}{N})^{k+c+1}}}{\frac{\Gamma(c+1)}{\Gamma (k+c+1)}}\prod_{i=1}^{k}\frac{\Gamma(m_{i} -\sigma)\Gamma(N-m_{i} +1)}{\Gamma(N-\sigma+1)}.
\end{equation}
\end{proposition}

\subsection{BNP inference for the unseen-features problem} \label{sec:extrapolation}

Now, we apply the SB-SP prior to the unseen-features problem. For any $N\geq1$ let $Z_{1:N}$ be an observable sample modeled as the BNP Bernoulli model \eqref{exch_mod}, with $\zeta\sim\text{\rm SB-SP}(\sigma,c,\beta)$. Moreover, under the same model of the $Z_{n}$'s, for any $M\geq1$ let $(Z_{N+1},\ldots,Z_{N+M})$ be additional unobservable sample. Then, the unseen-feature problem calls for the estimation of 
\begin{equation}\label{eq:number_new_features}
\unseen{N}{M} = \sum_{i \ge 1} \ind\left(\sum_{m=1}^M A_{N+m, i} > 0\right)
 \ind\left( \sum_{n=1}^N A_{n, i} = 0 \right),
\end{equation}
namely the number of hitherto unseen features that would be observed in $(Z_{N+1},\ldots,Z_{N+M})$. As generalization of the unseen-feature problem \eqref{eq:number_new_features}, for $r\geq1$ we consider the estimation of
\begin{equation} \label{eq:number_new_features_freq}
\unseenfreq{N}{M}{r}= \sum_{i \ge 1} \ind\left( \sum_{m=1}^M A_{N+m, i} = r \right) \ind\left( \sum_{n=1}^N A_{n, i} = 0 \right),
\end{equation}
namely the number of hitherto unseen features that would be observed with prevalence $r$ in $(Z_{N+1},\ldots,Z_{N+M})$. Of special interest is $r=1$, which concerns rare (unique) features. The next theorem characterizes the posterior distributions of  $\unseen{N}{M}$ and $\unseenfreq{N}{M}{r}$, given $Z_{1:N}$. We denote by $\NegBin(n,p)$ the negative Binomial distribution with parameter $n$ and $p\in(0,1)$.

\begin{theorem} \label{thm:stable_news}
Let $Z_{1:N}$ be a random sample from \eqref{exch_mod} with $\zeta\sim\text{\rm SB-SP}(\sigma,c,\beta)$, and let $Z_{1:N}$ displays $K_{N}=k$ distinct features with frequencies $(M_{N,1},\ldots,M_{N,K_{N}})=(m_{1},\ldots,m_{k})$. Then, the posterior distributions of $\unseen{N}{M}$ and of $\unseenfreq{N}{M}{r}$, given $Z_{1:N}$, coincide with the distributions of
\begin{equation} \label{eq:stable_news}
\unseen{N}{M}\,|\, Z_{1:N} \sim\NegBin\left(K_N+c+1, \frac{\stablenews{N}{M}}{\beta+\stablenews{0}{N+M}}\right),    
\end{equation}
and
\begin{equation} \label{eq:stable_news_freq}
\unseenfreq{N}{M}{r}\,|\, Z_{1:N} \sim\NegBin\left(K_N+c+1, \frac{\stablenewsfreq{N}{M}{r}}{\beta+\stablenews{0}{N}+\stablenewsfreq{N}{M}{r}} \right),
\end{equation}
for any index of prevalence $r\geq1$, respectively, where $\gamma_{N}^{(M)}=\sigma\sum_{1\leq i\leq M}B(1-\sigma,N+i)$ and where $\rho^{(M,r)}_{N}={M\choose r}\sigma B(r-\sigma,N+M-r+1)$, with $B(\cdot,\cdot)$ denoting the (Euler) Beta function.
\end{theorem}

See Appendix \ref{appendix_extrapolation} for the proof of Theorem \ref{thm:stable_news}. The posterior distributions \eqref{eq:stable_news} and \eqref{eq:stable_news_freq} depend on $Z_{1:N}$ through the sample size $N$ and the number $K_{N}$ of distinct features. This is in contrast with the corresponding posterior distributions obtained under the Beta and the stable-Beta process priors, which are Poisson distributions that depend on $Z_{1:N}$ only through $N$ \citep[Proposition 1]{masoero2019more}. BNP estimators of $\unseen{N}{M}$ and $\unseenfreq{N}{M}{r}$, with respect to a squared loss function, are obtained as the posterior expectations of \eqref{eq:stable_news} and \eqref{eq:stable_news_freq}, i.e. 
\begin{equation} \label{eq:stable_news_est}
\hat{U}_{N}^{(M)}=(K_N+c+1)\frac{\stablenews{N}{M}}{\beta+\stablenews{0}{N+M}-\stablenews{N}{M}}
\end{equation}
and
\begin{equation} \label{eq:stable_news_freq_est}
\hat{U}_{N}^{(M,r)}=(K_N+c+1)\frac{\stablenewsfreq{N}{M}{r}}{\beta+\stablenews{0}{N}}
\end{equation}
respectively. The estimators \eqref{eq:stable_news_est} and \eqref{eq:stable_news_freq_est} are simple, linear in the sampling information and computationally efficient. In the next theorem we establish the large $M$ asymptotic behaviour of the posterior distributions  \eqref{eq:stable_news} and \eqref{eq:stable_news_freq}, showing that the number of unseen features has a power-law growth in $M$. The same growth in $M$ holds under the stable-Beta process prior \citep[Proposition 2]{masoero2019more}, though the limiting distribution is degenerate.

\begin{theorem} \label{thm:stable_conv}
Let $Z_{1:N}$ be a random sample from \eqref{exch_mod} with $\zeta\sim\text{\rm SB-SP}(\sigma,c,\beta)$, and let $Z_{1:N}$ displays $K_{N}=k$ distinct features with frequencies $(M_{N,1},\ldots,M_{N,K_{N}})=(m_{1},\ldots,m_{k})$. As $M\rightarrow+\infty$
\begin{align}
\label{eq:conv_as}
    \frac{\unseen{N}{M}}{M^{\sigma}}\mid Z_{1:N} \stackrel{{\rm a.s.}}{\longrightarrow} W_N,
\end{align}
where $W_{N}$ is a Gamma random variable with shape $(K_N+c+1)$ and rate $(\beta+\stablenews{0}{N})/\Gamma (1-\sigma)$, and 
\begin{align}
\label{eq:conv_as_freq}
    \frac{\unseenfreq{N}{M}{r}}{M^{\sigma}}\mid Z_{1:N} \stackrel{{\rm a.s.}}{\longrightarrow} W_{N,r},
\end{align}
where $W_{N,r}$ is a Gamma random variable with shape $(K_N+c+1)$ and rate $\Gamma(r+1)(\beta+\stablenews{0}{N})/\sigma\Gamma(r-\sigma)$.
\end{theorem}


\section{Experiments} \label{sec:exp}

Over the last decade, genomics has witnessed an extraordinary improvement in the data availability due to the advent of next generation sequencing technologies. Thanks to larger and richer datasets, researchers have started uncovering the role and impact of rare genetic variants in heritability and human disease \citep{hernandez2019ultrarare,momozawa2020unique}. The development of methods for estimating the number of new genomic variants to be observed in future studies is an active research area, as it can aid the design of effective clinical procedures in precision medicine \citep{ionita2009estimating,zou2016quantifying}, enhance understanding of cancer biology \citep{chakraborty2019somatic}, and help to optimize sequencing procedures \citep{rashkin2017optimal,masoero2019more}. Here, we consider datasets of individual genomic sequences. Following common practice, we assume that an underlying fixed and idealized genomic sequence (the ``reference'') is given. Then, each coordinate of an individual sequence reports the presence ($1$) or absence ($0$) of variation at a given locus with respect to the reference. All variants are treated equally, namely, any expression differing from the underlying reference at a given locus counts as a variant. We make use of our methodology to estimate the number of genomic loci at which variation was not observed in the original sample, and is going to be observed in (at least one of) $M$ additional datapoints.

We find in our experiments that the estimates of the total number of new variants to be observed produced using the SB-SP-Bernoulli model, hereafter referred to as SSB, tend to be more accurate than other available methods in the literature. This phenomenon is particularly evident when the sample size $N$ of the training set is small, and when the extrapolation size $M$ is large with respect to $N$. Moreover, the SSB model is particularly effective in estimating the number of new rare variants, e.g.\ variants appearing only once in the additional unobservable samples. Accurate estimation of rare variants is particularly important, as these are believed to be largely responsible for heritability of human disease \citep{rashkin2017optimal,chakraborty2019somatic}. To benchmark the quality of the SSB, we consider a number of competing methodologies for the feature prediction problem available in the literature: i) Jackknife estimators (J) \citep{gravel2014predicting}; ii) a linear programming method (LP) \citep{zou2016quantifying} and variations of Good-Toulmin estimators (GT) \citep{chakraborty2019somatic}. We also compare our empirical findings to a BNP estimator obtained under the stable-Beta process prior (3BB), which has been introduced in \citet{masoero2019more}.  We complete our analysis with a thorough investigation on synthetic data in \Cref{sec:app_synthetic_model} and \Cref{sec:app_synthetic_zipf}, as well as on additional real data from the gnomAD database \citep{karczewski2019variation} in \Cref{sec:gnomAD}.

\subsection{Empirics and evaluation metrics} \label{sec:learning_EFPF}

For the SSB method to be useful, we need to estimate the underlying, unknown, parameters of the SB-SP prior. To learn these prior's parameters, we here adopt an empirical Bayes procedure, which consists in maximizing the marginal distribution \eqref{eq:stable_marginal}. In particular, we maximize numerically \Cref{eq:stable_marginal} with respect to the parameters $\beta>0$, $c>0$ and $\sigma\in(0,1)$ of the SB-SP prior, and use the resulting values to produce our estimators. That is, we let
\begin{displaymath}
    (\hat{\beta}, \hat{c}, \hat{\sigma})= \arg\max_{(\beta,c,\sigma)} \left\{ p_{k}^{(N)}(m_{1}, \cdots, m_{k}) \right\},
\end{displaymath}
and plug these values in the BNP estimator \eqref{eq:stable_news} and \eqref{eq:stable_news_freq}. The resulting values provide our BNP estimates of the number $\unseen{N}{M}$ of new variants and the number $\unseen{N}{M,r}$ of new variants with prevalence $r$. 

To assess the accuracy of our estimates, we consider the percent deviation of the estimate from the truth to be the achieved accuracy. That is, the accuracy of the estimator $\estpred{N}{M}$ is defined as
\begin{align}
    \predaccuracy{N}{M}:= 1- \min\left\{\frac{|\unseen{N}{M} - \estpred{N}{M}|}{\unseen{N}{M}}, 1 \right\}. \label{eq:percentage_accuracy}
\end{align}
In particular, the accuracy $\predaccuracy{N}{M}$ equals $1$ when the estimate is perfect (no error is incurred), and decreases to $0$ as the estimate deviates from the truth. The $\min$ operator in \eqref{eq:percentage_accuracy} ensures that $\predaccuracy{N}{M}$ lies in $[0,1]$: we let the accuracy to be equal to $0$ whenever there is a severe overestimation, and the percentage estimation error exceeds $100\%$, i.e.\ when $\estpred{N}{M} \ge 2\times \unseen{N}{M}$. The SSB,  3BB  and LP methods also offer an estimate for the number of new features observed with a given prevalence $r$. We let $\predaccuracyfreq{N}{M}{r}$ be the accuracy metric, where we replace in \eqref{eq:percentage_accuracy} $\unseen{N}{M}$ with $\unseenfreq{N}{M}{r}$, the number of new features observed with prevalence $r$, and $\estpred{N}{M}$ with  $\estpredfreq{N}{M}{r}$. 

\subsection{Estimating the number of new variants in cancer genomics}

Following the empirical study of \citet{chakraborty2019somatic}, we make use of data from the Cancer Genome Atlas (TCGA), the largest publicly available cancer genomics dataset, containing somatic mutations from $ 10{,}295$ patients and spanning $33$ different cancer types. We partition the samples into $33$ smaller datasets according to cancer-type annotation of each patient. See \citet{chakraborty2019somatic} and \citet[Appendix F]{masoero2019more} for details on the data and the experimental setup. For each cancer type, we retain a small fraction of the data for purposes of training, and consider the task of estimating the number of new variants that will be observed in a follow-up sample given a pilot sample. We validate our estimates by comparing the estimate $\estpred{N}{M}$ of the number of distinct variants to the true value, obtained by extrapolating to the remaining data.  To assess the variability and error in our estimates, we repeat for every cancer type the experiment on $S=1{,}000$ subsets of the data, each obtained by randomly subsampling without replacement from the full sample. 
 
We find that the SSB and 3BB methods perform particularly well when the training sample size $N$ is small compared to the extrapolation sample size $M$.  This setting is relevant in the context of cancer genomics, as scientists are interested in understanding the ``unexploited potential'' of the genetic information, especially for rare cancer subtypes \citep{chakraborty2019somatic, huyghe2019discovery}. To compare and quantify the performance of the available methodologies in this setting, we report in \Cref{fig:cancer_pred_small_N} the distribution of the estimation accuracy when retaining only $N=10$ samples for training and extrapolating to the largest possible sample size $M$ for which we can compute the accuracy metric (\Cref{eq:percentage_accuracy}). We report results for the $10$ cancer types with the largest number of samples in the original dataset. For each cancer type and  for each method, the distribution of the estimation accuracy is obtained by considering its performance across the $S=1{,}000$ replicates. Across all cancer types, the estimates obtained from the SSB method achieve higher accuracy. 

We show in \Cref{fig:cancer_pred_1}  the behavior of $\estpred{N}{i}$ for five different cancer types as $i=1, \ldots,M$. Again, we let $N=10$, and $M$ be the largest possible extrapolation value, as dictated by the dataset size. We report the estimates obtained from a fixed sample of size $N=10$, as well as the variability around such estimates obtained by re-fitting each model, iteratively leaving one datapoint out from the sample. In this setting, the SSB method outperforms competing methods in terms of estimation accuracy. Moreover, the variability in the estimates arising from re-fitting the model on subsets of the data provides a useful measure of uncertainty in such estimation. 
 \begin{figure}[h!]
      \centering \includegraphics[width=\textwidth,height=\textheight,keepaspectratio]{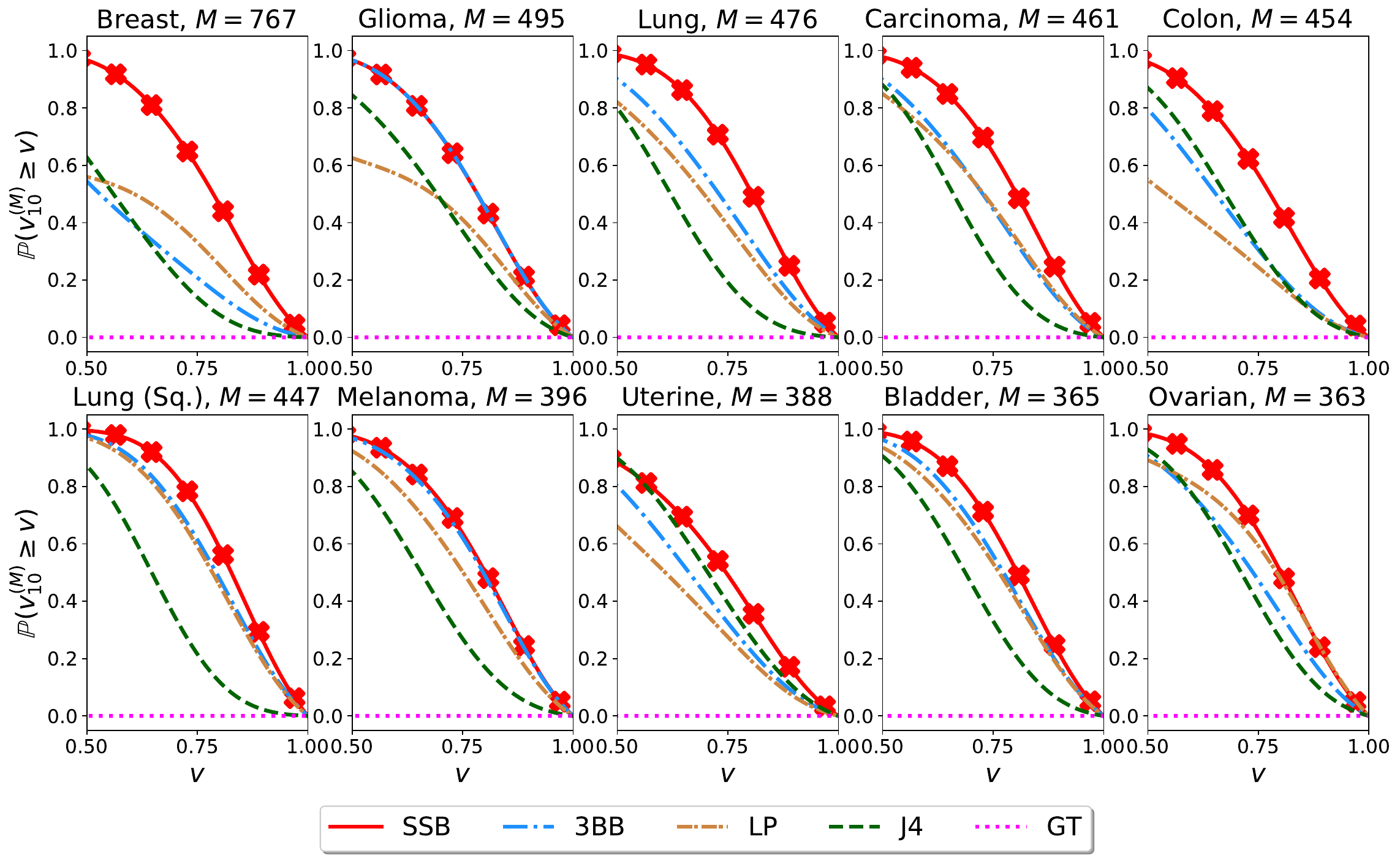}
    \caption{\footnotesize{Estimation accuracy $\predaccuracy{N}{M}$ for the number of new genomic variants $\estpred{10}{M}$. For each method and each cancer type, we retain $N=10$ random samples and use them to estimate up to $M$ total observations, where $N+M$ is the size of the original sample.}}
    \label{fig:cancer_pred_small_N}
\end{figure}

\begin{figure}
      \centering \includegraphics[width=\textwidth,height=\textheight,keepaspectratio]{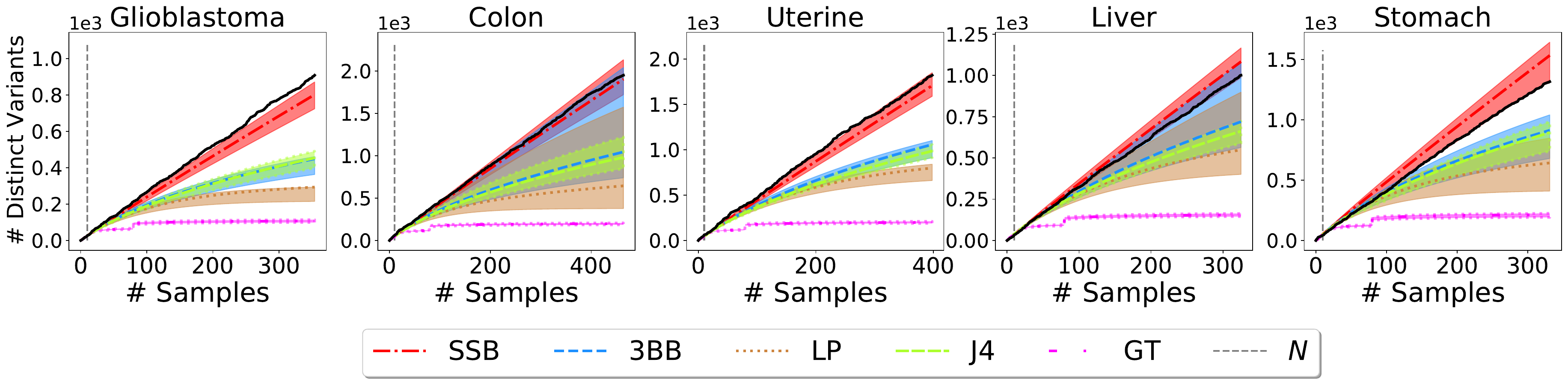}
    \caption{\footnotesize{Estimation of the number of new genomic variants $\estpred{N}{i}$, for $i=1,\ldots,M$. For each method and cancer type, we retain $N=10$ random samples and use them to estimate up to the largest possible size. We fit each model on the full sample, as well as $N = 10$ additional times by iteratively leaving one datapoint out from the training sample. The solid black line is the true number of features that would have been observed  (vertical axis) for any extrapolation size $N+M$ (horizontal axis), for a fixed ordering of the data. Shaded regions report the prediction range obtained from the estimates from the leave-one-out fits.}}
    \label{fig:cancer_pred_1}
\end{figure}

\subsection{Estimating the number of new rare variants in cancer genomics}

In recent years, the cancer genomics research community has become increasingly interested in studying and understanding the role of extremely rare variants, such as singletons, i.e.\ observed in only one patient. Evidence suggests that rare deleterious variants can have far stronger effect sizes than common variants \citep{rasnic2020expanding} and can play an important role in the development of cancer. For example, in breast cancer, it is well accepted that the risk of a variant is inversely proportional with respect to its prevalence: the rarer the variant, the higher the risk \citep{wendt2019identifying}. Therefore, effective identification and discovery of rare variants is an active, is an ongoing research area \citep{lawrenson2016functional, lee2019boadicea}. This phenomenon is not limited to breast cancer, but is progressively being studied across different cancer types. See, e.g. the recent works on ovarian \citep{phelan2017identification}, skin \citep{goldstein2017rare}, prostate  \citep{nguyen2020rare} and lung \citep{liu2021rare} cancers and references therein.  In downstream analysis, these estimates could be useful for planning and designing future experiments, e.g.\ informing scientists on the number of new samples to be collected in order to observe a target number of new variants, or for power analysis considerations in rare variants association tests \citep{rashkin2017optimal}. 

 \begin{figure}[h!]
      \centering \includegraphics[width=\textwidth,height=\textheight,keepaspectratio]{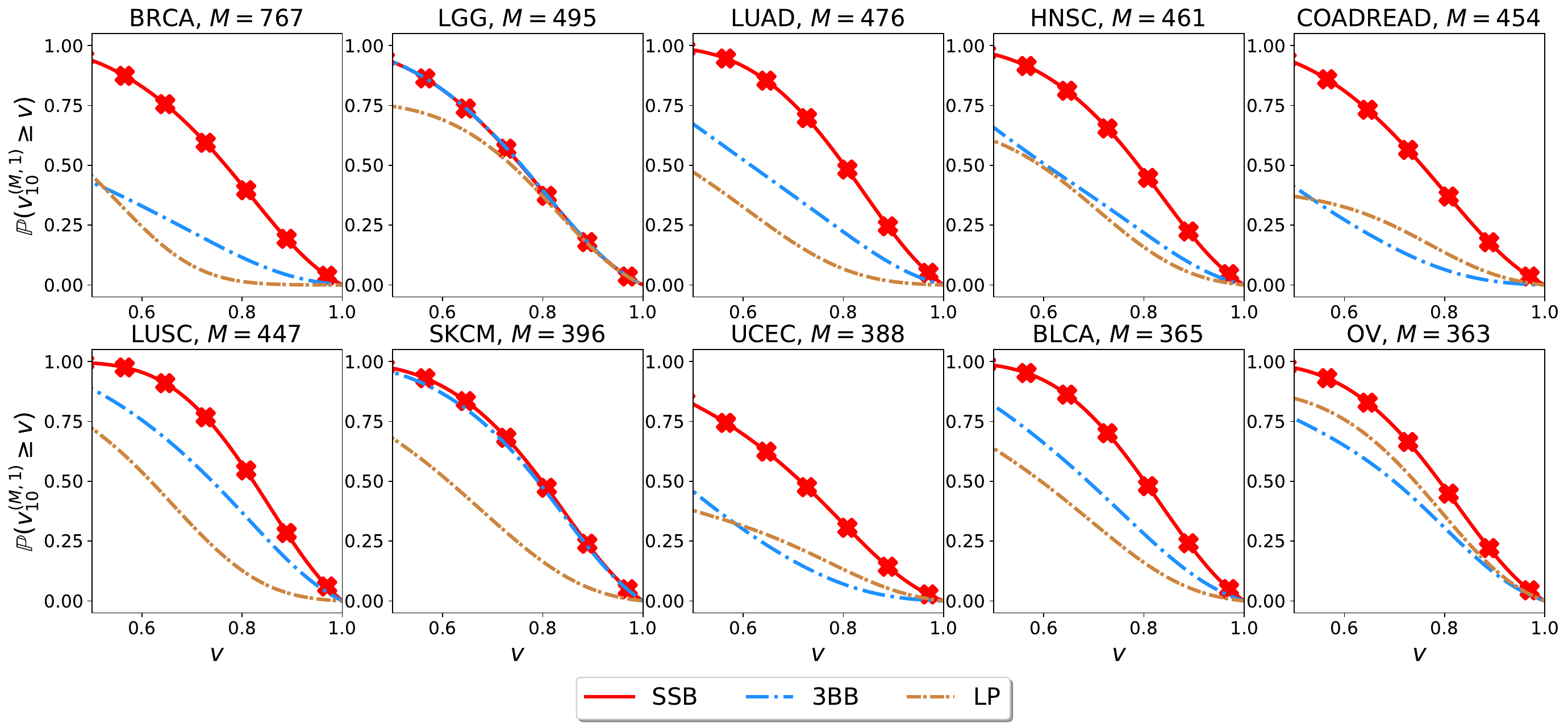}
    \caption{\footnotesize{Estimation accuracy $\predaccuracyfreq{N}{M}{1}$ for new variants appearing with prevalence one in future unobservable samples for different cancer types. For each method and each cancer, we retain $N=10$ random samples and use them to estimate up to the largest possible size.}}
    \label{fig:cancer_rare_small_N}
\end{figure}

The BNP framework considered here allows us to estimate the number of new rare variants to be discovered. While \citet{zou2016quantifying} did not consider the problem of estimating rare variants, it is straightforward to obtain an estimate for this quantity from their framework. Indeed, for every prevalence $x \in [0,1]$, the LP estimates the histogram $h(x)$, which counts the number of variants appearing with prevalence $x$ in the population, and the number of variants appearing with prevalence $r$ follows from the binomial sampling model assumption, namely $\estpredfreq{N}{M}{r}=\sum_x h(x)\left\{\binom{N+M}{r}x^r(1-x)^{N+M-r}-\binom{N}{r} x^r(1-x)^{N-r}\right\}$. We show in \Cref{fig:cancer_rare} that the SSB method provides better estimates than the 3BB and LP methods.
\begin{figure}[h!]
      \centering \includegraphics[width=\textwidth,height=\textheight,keepaspectratio]{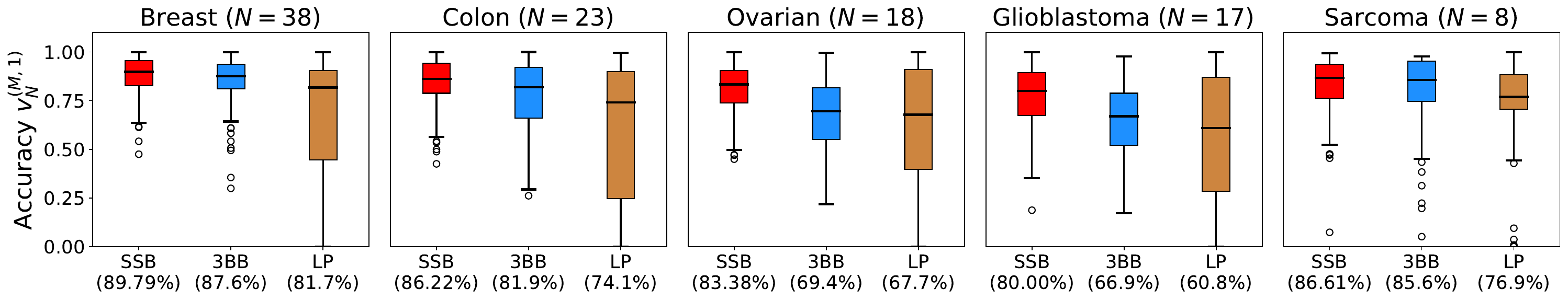}
    \caption{\footnotesize{Estimation accuracy $\predaccuracyfreq{N}{M}{1}$ for new variants appearing with prevalence one in future samples. For each method and different cancer types, we  retain a random sample of size $N=5\%$ of the available dataset, and use it to estimate up to the largest possible size.}}
    \label{fig:cancer_rare}
\end{figure}

\subsection{Coverage and calibrated uncertainties}

One of the benefits of the BNP approach is that it automatically yields a notion of variability of the estimate of $U$ via posterior credible intervals. We here check whether these intervals produce a useful notion of uncertainty, by investigating their calibration. For $\alpha \in (0,1)$, we say that a $100\times\alpha\%$ credible interval is calibrated if it contains the true value of interest, arising from hypothetical repeated draws, $100\times\alpha\%$ of the times. We here assess the calibration of a $100\times\alpha \%$ credible interval for $U_N^{(M)}$ conditionally given $Z_{1:N}$ as follows. Let $S$ be a large number ($S=1{,}000$ in our experiments). For each $s=1,\ldots,S$, we retain a random subset of the data of size $N$, and estimate the corresponding parameters $\hat{\beta},\hat{c},\hat{\sigma}$ as discussed in \Cref{sec:learning_EFPF}. Then,
 we let $\hat{W}_{N, s, low}^{(M)}(\alpha), \hat{W}_{N, s, hi}^{(M)}(\alpha) $ be the endpoints of a $100\times \alpha\%$ credible interval for the distribution of the number of new features, as given by \Cref{eq:stable_news}, centered around the posterior predictive mean. We compute coverage calibration via
\begin{align*}
    w_{N}^{(M)}(\alpha) = \frac{1}{S}\sum_{s=1}^S \ind\left\{ \hat{W}_{N, s, low}^{(M)} (\alpha) \le K_{N+M} \le \hat{W}_{N,s,hi}^{(M)}(\alpha) \right\}.
\end{align*}
This is the fraction of the $S$ experiments in which the true value was contained by an $100\times\alpha\%$ credible interval. The closer $w_{N}^{(M)}(\alpha)$ to $\alpha$, the better calibrated the credible intervals. We compute the same quantity for the 3BB method using the results in \citet{masoero2019more}. Although still not perfect, we find that the posterior predictive intervals obtained from the SSB method are better calibrated than the ones under the 3BB method (see \Cref{fig:cancer_coverage}). 
\begin{figure}
\centering \includegraphics[width=\textwidth,height=\textheight,keepaspectratio]{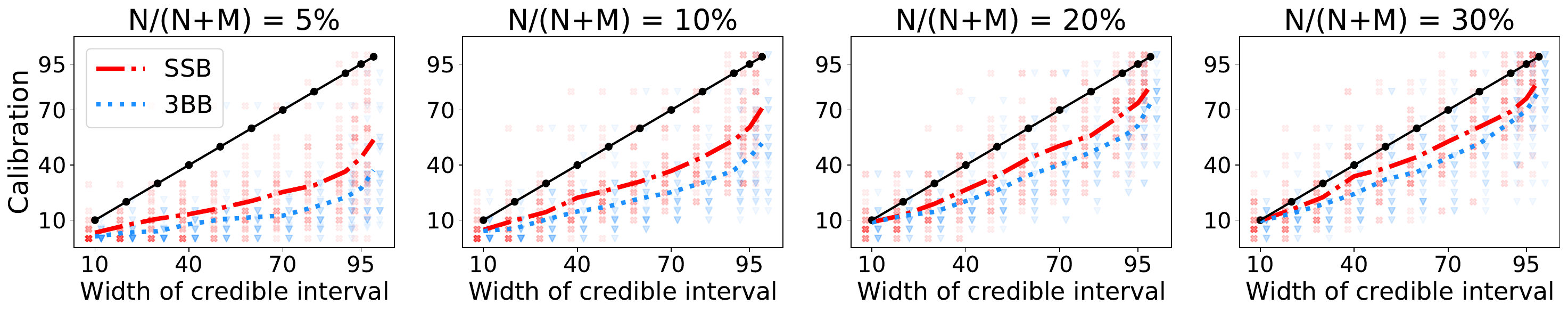}
\caption{\footnotesize{Coverage calibraiton of BNP estimators for number of new variants in future samples across all cancer types in TCGA. Different subplots refer to different ratios of the training $N$ with respect to the extrapolation $M$. For each cancer, we retain a training sample of size $N \in \{5\%, 10\%, 20\%, 30\%\}$ of the total available dataset, and extrapolate up to the largest available $M$. Colored lines report the average coverage $w_{N}^{(M)}(\alpha)$  across all cancer types ($y$-axis) as a function of $\alpha$ ($x$-axis). Faded dots refer to coverage for individual cancer types.}}
    \label{fig:cancer_coverage}
\end{figure}


\section{Discussion} \label{sec:discussion}

\citet{masoero2019more} first applied CRM priors to the unseen-features problem, showing that: i) despite the broadness of the class of CRM priors, all CRM priors lead to the same Poisson posterior structure for the number of unseen features, which thus makes them not a flexible prior model for the unseen-features problem; ii) while the Poisson posterior distribution may be appealing in principle, making the posterior inferences analytically tractable and of easy interpretability, its independence from $Z_{1:N}$ makes the BNP approach a questionable oversimplification, with posterior inferences being completely determined by the estimation of unknown prior's parameters. In this paper, we introduced the SB-SP prior, and showed that: i) it enriches the posterior distribution of the number of unseen features arising under CRM priors, which results in a negative Binomial distribution whose parameters depend on the sample size and the number of distinct features; ii) it maintains the same analytical tractability and interpretability as CRM priors, which results in BNP estimators that are simple, linear in the sampling information and computationally efficient. The effectiveness of the SB-SP prior is showcased through an empirical analysis on synthetic and real data. Under the SB-SP prior, we found that estimates of the unseen number of features are accurate, and they outperform the most popular competitors in the challenging scenario where the sample size $N$ is particularly small, and also small with respect to the extrapolation size $M$.

Our approach admits an extension to the multiple-feature setting, which takes into account of the many forms of variation, e.g. single nucleotide changes, tandem repeats, insertions and deletions, copy number variations  \citep{zou2016quantifying}. We briefly describe the multiple-feature setting, and defer to Appendix \ref{appendix_multivariate} for details. It is assumed that a feature $w_i$ comes with a characteristic, i.e.\ the form of variation, chosen among $q>1$ characteristics. For $N \geq 1$, the observable sample $\bm{Z}_{1:N}= (\bm{Z}_1, \ldots , \bm{Z}_N)$ is modeled as a $\{0,1\}^q$-valued stochastic process $\bm{Z}  = \sum_{i \geq 1} \bm{A}_{i} \delta_{w_i }$, where $\bm{A}_i :=(A_{i,1}, \ldots , A_{i,q})$ is a Multinomial random variable with parameter $\bm{p}_i = (p_{i,1}, \ldots , p_{i,q})$ such that $|\bm{p}_i| = \sum_{1\leq j\leq q} p_{i,j} <1$, and the $\bm{A}_i$'s are i.i.d. That is, for any $i\geq1$ all the $A_{i,j}$'s are equal to $0$ with probability $(1-|\bm{p}_i|)$, i.e.\ $w_{i}$ does not display variation, or only one $A_{i,j}$'s is equal to $1$ with probability $p_{i,j}$, i.e.\ $w_{i}$ displays variation with characteristic $j$. $\bm{Z}$ is a multivariate Bernoulli process with parameter $\bm{\zeta}= \sum_{i \geq 1} \bm{p}_i \delta_{w_i}$. The stable-Beta-Dirichlet process prior for $\bm{\zeta}$ is a multivariate generalization of the stable-Beta process prior \citep{james2017bayesian}, and it leads to a Poisson posterior distribution for the number of unseen features, given $\bm{Z}_{1:N}$, which depends on $\bm{Z}_{1:N}$ only through $N$. In Appendix \ref{appendix_multivariate} we introduce a scaled version of the stable-Beta-Dirichlet process, and show that it leads to a negative Binomial posterior distribution for the number of unseen features, which depends on $\bm{Z}_{1:N}$ through  $N$ and the number of distinct features in $\bm{Z}_{1:N}$.

SP priors have been introduced in \citet{james2015scaled} and, to the best of our knowledge, since then no other works have further investigated such a class of priors. To date, the peculiar predictive properties of SP priors appear to be unknown in the BNP literature. Our work on the unseen-features problem is the first to highlight the great potential of SP priors in BNPs, showing that they provide a critical tool for enriching the predictive structure of the popular CRM priors \citep{james2017bayesian,broderick2018posteriors}. We believe that SPs may be of interest beyond the unseen-features problem, and more generally beyond the broad class of feature sampling problems. CRM priors, and in particular the Beta and stable-Beta process priors, have been widely used in several contexts, with a broad range of applications in topic modeling, analysis of social networks, binary matrix factorization for dyadic data, analysis of choice behaviour arising from psychology and marketing surveys, graphical models, and analysis of similarity judgement matrices. See \cite{Gri_11} and references therein for details. In all these contexts, SP priors may be more effective than CRM priors, as they allow to better exploit the sampling information in posterior inferences.

Among applications of SP priors beyond features sampling problems, it is worth mentioning the use of SP priors as hierarchical (or latent) priors in models of unsupervised learning \citep[Section 5]{Gri_11}, the most popular being Gaussian latent feature modeling. Differently from features sampling problems, where the values of features' labels $W_i$s are immaterial, in Gaussian latent feature modeling the values the $W_{i}$'s become material. That is, under the Gaussian latent feature model with a SP prior, observations are assumed to modeled as a multivariate Gaussian distribution, whose mean depends on latent features that are modeled with a SP prior, thus making the values of features' labels $W_{i}$'s of critical importance for the analysis. Bayesian factor analysis \citep{Knowles2011} provides another context where SP priors may be usefully applied as hierarchical priors. Within the context of factor analysis, we also mention the work of \cite{Ayed_2021} with applications to network analysis. There, the authors exploit CRM priors to recover the latent community structure in a network between individuals, and the features' labels describe the level of affiliation of a certain individual to a latent community.  In such a context, we believe that SP priors may be used in place of CRM priors, with the advantage of introducing richer predictive structure. In this respect, our work paves the way to promising directions of future research, in terms of both methods and applications.


\section*{Acknowledgement}
The authors thank Joshua Schraiber for useful discussions. Federico Camerlenghi and Stefano Favaro received funding from the European Research Council (ERC) under the European Union's Horizon 2020 research and innovation programme under grant agreement No 817257. Federico Camerlenghi and Stefano Favaro gratefully acknowledge the financial support from the Italian Ministry of Education, University and Research (MIUR), ``Dipartimenti di Eccellenza" grant 2018-2022. Lorenzo Masoero and Tamara Broderick were supported in part by the DARPA I2O LwLL program, an NSF CAREER Award, and ONR award N00014-17-1-2072.


\appendix

\section{A brief account on completely random measures} \label{appendix_crm}

In this section we provide a short account on completely random measures (CRMs). For a more exhaustive treatment refer to \citet{daleyII,kingman1992poisson}. Let us denote by $\mathds{W}$ a Polish space equipped with its Borel $\sigma$-field $\mathcal{W}$, and we also indicate by $\mathcal{B}_{\R_+}$ the Borel $\sigma$-field of the positive real line $\R_+$.
Denote by  $\mathsf{M}_\mathds{W}$ the space of all bounded and finite measures on $(\mathds{W}, \mathcal{W})$, in other words $\mu \in \mathsf{M}_\mathds{W}$ iff $\mu (A)< +\infty$ for any bounded set $A\in \mathcal{W}$. The space $\mathsf{M}_\mathds{W}$ is usually assumed to be equipped with a proper Borel $\sigma$-algebra, which is induced by the so called \textit{weak-hash convergence} and denoted here as $\mathcal{M}_\mathds{W}$ (see \cite{daleyII} for details).  
\begin{definition}
A Completely Random Measure (CRM) $\mu$ on $(\mathds{W}, \mathcal{W})$ is a random element defined on a suitable probability space and taking values in $(\mathsf{M}_\mathds{W},\mathcal{M}_\mathds{W})$ such that the random variables $\mu (A_1), \ldots , \mu (A_n)$ are independent for any choice of bounded and disjoint sets 
$A_1, \ldots , A_n \in \mathcal{W}$ and for any $n \geq 1$.
\end{definition}
\cite{kingman1967completely} proved that a CRM may be decomposed as the sum of three main components: i) a deterministic drift $u$, namely a deterministic measure on $(\mathds{W} , \mathcal{W})$; ii)  a part with random jumps $(\tau_i)_{i \geq 1}$ at random locations $(W_i)_{i \geq 1}$, denoted here as $\mu_c= \sum_{i \geq 1} \tau_i \delta_{W_i}$; iii) a component with random jumps $(\eta_i)_{i \geq 1}$ at fixed locations 
$w_1, w_2, \ldots \in \mathds{W}$. That is to say
\begin{equation}  \label{eq:decomposition}
\mu (\,\cdot \,) = u(\,\cdot \,)+\mu_c (\,\cdot \,) +\sum_{i \geq 1} \eta_i \delta_{w_i} (\,\cdot \,) .
\end{equation}
See \citet{daleyII} for a proof.

Following standard practice in the nonparametric literature, in this paper we deal with CRMs without deterministic drift and without fixed atoms, namely we assume  that $\mu \equiv \mu_c$. In this case $\mu = \mu_c$ is characterized through the L\'evy-Khintchine representation of its Laplace functional:
\begin{equation}
\label{eq:LK}
\E \left[ e^{-\int_\mathds{W} f (w) \mu_c(\de w)} \right]  =
\exp \left\{ -\int_{\R_+ \times\mathds{W}}  (1-e^{-s f (w)}) \nu (\de s, \de w) \right\},
\end{equation}
for any measurable function $f : \mathds{W} \to \R_+$, where $\nu $ is a measure on $\R_+ \times\mathds{W}$ and it is referred to as the L\'evy intensity of the CRM $\mu_c$. The measure $\nu$ is also required to satisfy the following conditions
\[
\nu (\R_+ \times \{w\}) = 0 \quad \forall w \in \mathds{W} , \quad \text{and } \int_{\R_+ \times A} \min \{s,1  \} \nu (\de s , \de w) < \infty
\]
for any bounded $A \in \mathcal{W}$.
The representation \eqref{eq:LK} is of paramount importance to prove all our posterior results, and it clarifies the pivotal role of $\nu$ in the determination of the distributional properties of $\mu_c$.  \cite{Kallenberg2010} provides a very general decomposition for such a measure $\nu$ as follows:
$\nu (\de s , \de w) = \lambda_w (\de s) \Lambda (\de w)$, where $\Lambda$ is  a $\sigma$-finite measure on $(\mathds{W} , \mathcal{W})$ and $\lambda_w $ is  a transition kernel, i.e., $w \to \lambda_w (A)$ is $\mathcal{W}$-measurable for all Borel sets $A \in \mathcal{B}_{\R_+}$ and $A \to \lambda_w (A)$ is a measure on $(\R_+,  \mathcal{B}_{\R_+})$. When $\lambda_w (\de s)\equiv \lambda (\de s)$ does not depend on $w \in \mathds{W}$, we say that the CRM is homogeneous, which is tantamount to saying that the atoms $W_i$'s and the jumps $\tau_i$'s are independent random variables.
In BNP problems, it is common to suppose that $\Lambda (\de w)= \alpha P (\de w)$, where $P$ is a probability measure on $(\mathds{W} , \mathcal{W})$ and $\alpha >0$. Two remarkable examples of CRMs are the $\sigma$-stable process, which can be recovered by choosing $\lambda (\de s) = \sigma s^{-1-\sigma } \de s $, and the gamma process, which corresponds to the choice $\lambda (\de s ) = e^{-s}/s \: \de s$. See also \citep{lijoi2010models} for additional details and connections with the BNP literature.\\

In Section \ref{appendix_multivariate}, we will make use of multivariate CRMs to define a multivariate extension of the Bernoulli process model, called the Bernoulli process model with a condiment. For this reason we now specify what we mean for a multivariate CRM. A vector $\bm{\mu}= (\mu_1, \ldots , \mu_q)$ of completely random measures is said to be a multivariate CRM if the random variables 
\[
    (\mu_1(A_1), \ldots , \mu_q (A_1)), \ldots , (\mu_1 (A_n), \ldots , \mu_q (A_n))
\]
are independent for any choice of bounded and disjoint Borel sets $A_1, \ldots , A_n \in \mathcal{W}$ and for any $n \geq 1$.  
A decomposition similar to the one stated in \Cref{eq:decomposition} holds true for multivariate CRMs as well \citep{Kallenberg2010}. 
In the present paper we focus on multivariate CRMs which are functionals of marked Poisson point processes on $\R_+^q \times \mathds{W}$, i.e., 
\[
\bm{\mu} = \sum_{i \geq 1} \bm{\tau}_i \delta_{W_i},
\]
where $(\bm{\tau}_i)_{i \geq 1}$ are random jumps in $\R_+^q$ and $(W_i)_{i \geq 1}$ is a sequence of random atoms in $\mathds{W}$. Such a multivariate CRM 
has the following L\'evy-Khintchine representation which generalizes \Cref{eq:LK}:
\begin{equation}
\label{eq:multivariateLK}
\begin{split}
&\E [e^{-\int_{\mathds{W}} f_1 (w) \mu_1 (\de w)- \cdots  - \int_{\mathds{W}} f_q (w) \mu_q (\de w)}] \\
& \qquad= 
\exp \left\{- \int_{\mathds{W}}  \int_{\R_{+}^q} (1-e^{-s_1 f_1 (w)- \cdots - s_q f_q (w)}   )  \nu_{(q)} (\de s_1, \ldots, \de s_q, \de w) \right\}
\end{split}
\end{equation}
for arbitrary measurable functions $f_1, \ldots , f_d: \mathds{W} \to \R_+$.
The intensity measure $\nu_{(q)}$ in \eqref{eq:multivariateLK} is required to simultaneously satisfy 
\[
\nu_{(q)} (\R_{+}^q \times \{w\}) = 0 \quad \forall w \in \mathds{W}
\]
and
\[
\int_{\R_+ \times A} \min \{||\bm{s}||,1  \} \nu_{(q)} (\de s_1, \ldots , \de s_q , \de w) < \infty,
\]
for any bounded $A \in \mathcal{W}$, and having denoted by $||\bm{s}||$ the Euclidean norm of the vector $\bm{s}:= (s_1, \ldots , s_q)$.
In the present paper,  we will  work with a \textit{homogeneous} L\'evy intensity measure of the following form $
\nu_{(q)} (\de s_1, \ldots, \de s_q, \de w)=  \lambda_{(q)} (s_1, \ldots , s_q) \de s_1 \cdots \de s_q  P (\de w)$, where 
$P$ is a diffuse probability measure on $(\mathds{W}, \mathcal{W})$ and $\lambda_{(q)}: \R^q_+ \to \R_+$ is measurable.  See, e.g., 
\citep{Kallenberbg2017} for further details.

\section{Posterior analysis for SP priors: proofs and details}  \label{appendix_scaledpriors}

In the present section we derive the marginal, posterior and predictive distributions for the Bernoulli process model under a scaled process prior. Specifically we focus on the following statistical model throughout the section:
\begin{equation}
\label{eq:SP_model}
\begin{split}
Z_n \mid \mu &\simiid {\rm BeP} (\mu_{\Delta_{1,h}}), \quad \text{for}\quad n =1, \ldots , N \\
\mu_{\Delta_{1,h}}  &   \sim {\rm SP} (\nu , h) ,
\end{split}
\end{equation}
where $\mu_{\Delta_{1,h}}$ has been defined at the beginning of Section \ref{sec:SP_prior}.
In Subsection \ref{sub:lemma} we provide some lemmas regarding SP priors, then Subsection \ref{sub:posterior_analysis} is concerned with the Bayesian posterior analysis of the model in \eqref{eq:SP_model}.

\subsection{Preparatory lemmas} \label{sub:lemma}
Some preparatory  lemmas are required before the posterior analysis. The first lemma provides the reader with the conditional distribution of
 $\mu_{\Delta_{1,h}}$ given $\Delta_{1,h}$.
\begin{lemma} \label{lemma:conditional_scaled_random_measure}
    Let $\mu_{\Delta_{1,h}} \sim {\rm SP} (\nu , h) $,  governed by the L\'evy intensity measure $\nu (\de s ,\de w) =\lambda(s)\de s P(\de w)$ on $\R_+ \times \mathds{W}$. 
The conditional distribution of $\mu_{\Delta_{1,h}}$, given $\Delta_{1,h}$, equals the one of a CRM on $(\mathds{W}, \mathcal{W})$ with L\'evy intensity
    \begin{equation*} 
        \Delta_{1,h} \lambda(\Delta_{1,h} s)\ind_{(0,1)} (s) \de s P(\de w).
    \end{equation*}
\end{lemma}
\begin{proof}
Recall the construction of a SP prior, as detailed in Section \ref{sec:SP_prior}. It starts from an underlying  CRM  $\mu=\sum_{i\geq1}\tau_{i}\delta_{W_{i}}$ with intensity  $ \nu$ on $\R_{+} \times \mathds{W}$. Moreover, having denoted by $\Delta_{1}>\Delta_{2}>\ldots$ the decreasingly ordered jumps $\tau_{i}$'s of $\mu$, one considers:
\[
\mu_{\Delta_{1}} = \sum_{i\geq 1}\frac{\Delta_{i+1}}{\Delta_{1}}\delta_{W_{i+1}},
\]
and the SP process is defined by a change of measure of the largest jump $\Delta_1$, replaced with the distribution of $\Delta_{1,h}$. As a consequence it is sufficient to prove that $\mu_{\Delta_{1}} \mid  \Delta_1$ is a CRM with L\'evy intensity
\begin{equation}
\label{eq:Levy_normalized}
\Delta_{1} \lambda(\Delta_{1} s)\ind_{(0,1)} (s) \de s P(\de w).
\end{equation}
In order to prove this remind that  $(\Delta_{i})_{i \geq 2} | \Delta_1$ are the points of a Poisson process with L\'evy intensity $\lambda (s) \ind_{(0, \Delta_1)} (s) \de s $, thanks to the representation by \cite{FergusonKlass_1972}. 
Therefore, the conditional distribution of $\mu_{\Delta_1}$, given $\Delta_1$, may  be found by a simple evaluation of the Laplace functional. To this end, consider a measurable function
$f: \mathds{W} \to \R_+$ and compute
\begin{align*}
&  \E [e^{-\int_\mathds{W} f (w)\mu_{\Delta_{1}} (\de w)} | \Delta_{1}]=
\E \left[ e^{  -\sum_{i \geq 1}   f(W_{i+1})\Delta_{i+1}/\Delta_{1}} |  \Delta_{1}  \right]\\
&\qquad\qquad = \exp \left\{ - \int_\mathds{W} \int_0^{+\infty} (1-e^{- f(w) s/\Delta_1})  \ind_{(0, \Delta_1)} (s) \lambda (s) \de s \; P (\de w) \right\}\\
&\qquad\qquad = \exp \left\{ - \int_\mathds{W} \int_0^{+\infty} (1-e^{- f(w) s})  \ind_{(0, 1)} (s) \lambda (s \Delta_1) \Delta_1 \de s \; P (\de w) \right\}
\end{align*}
which is exactly the Laplace functional of a CRM having L\'evy intensity \eqref{eq:Levy_normalized}.\\
\qed
\end{proof}
We now provide the reader with a sufficient condition to ensure that each $Z_n$ in \eqref{eq:SP_model} is almost surely finite, for any $n \geq 1$.
\begin{lemma} \label{lem:finite_features}
    Consider the model in \Cref{eq:SP_model}. If 
    \begin{equation} \label{eq:cond_finite_feature}
        \E\left[\int_0^1 \largestjump{1}{h} \lambda( s \largestjump{1}{h}) \de s\right] < \infty, 
    \end{equation}
 then each $Z_n$ displays almost surely finitely many features --- i.e.\ $\sum_{i \geq 1} A_{n,i} < \infty$, almost surely, for every $n \geq 1$.
\end{lemma}
\begin{proof}
For a fixed $n \geq 1$, it is sufficient to show that condition \eqref{eq:cond_finite_feature} entails 
    \[
        \E\left[ \sum_{i =1}^\infty A_{n,i} \right] <\infty.
    \]
The expected value in the previous formula may be computed as follows
    \begin{align*}
        \E\left[ \sum_{i =1}^\infty A_{n,i} \right]  &=\E \left[ \E\left[ \sum_{i =1}^\infty A_{n,i} \Big| \Delta_{1,h} \right]  \right]  
        = \E \left[ \E [\mu_{\Delta_{1,h}} (\mathds{W}) | \Delta_{1,h}]  \right]\\
        &=  \E \left[\int_\mathds{W} \int_0^1  s \Delta_{1,h} \lambda(\Delta_{1,h} s) \de s P(\de w) \right]
        =  \E \left[ \int_0^1  s \Delta_{1,h} \lambda(\Delta_{1,h} s) \de s \right] 
    \end{align*}
where we have applied the Campbell theorem \citep{kingman1992poisson} and  Lemma \ref{lemma:conditional_scaled_random_measure} to evaluate the total mass $\mu_{\Delta_{1,h}}(\mathds{W})$ of $\mu_{\Delta_{1,h}}$.
 As a consequence, condition \eqref{eq:cond_finite_feature} is sufficient for the finiteness of the Bernoulli process $Z_n$.\\ 
 \qed
\end{proof}

\subsection{Posterior analysis}  \label{sub:posterior_analysis}

We start with the marginal distribution of the observations $Z_{1:N}$ induced by the model. Our derivation closely follows the proof in \citet{james2017bayesian}. The marginal distribution is the counterpart of the ``exchangeable feature probability function'' (EFPF) for the Indian Buffet Process (IBP; see, e.g., \citet{broderick2013feature}).
\begin{proposition}[Joint marginal distribution] \label{prop:general_marginal} 
For any $N \geq 1$, let $Z_{1:N}$ be a random sample modeled as the BNP-Bernoulli model \eqref{eq:SP_model}, where $\mu_{\Delta_{1,h}} \sim {\rm SP} (\nu , h) $.
The probability that the observations $Z_{1:N}$ display $K_N=k$ distinct features, labelled by $\{ W_1^*, \ldots , W_{K_N}^* \}$, with corresponding frequencies $(M_{N,1}, \ldots , M_{N,K_N})= (m_1, \ldots , m_k)$, equals
\begin{equation*}  
\begin{split}
    & p_k^{(N)} (m_1, \ldots, m_k) =  \int_0^{+\infty} e^{-\sum_{n=1}^N \phi_n(a)}
    \prod_{i=1}^{k}  \int_0^1 
    s^{m_i} (1-s)^{N-m_i} a\lambda (a s ) \de s  \; f_{\Delta_{1,h}} (a) \de a  ,
\end{split}
\end{equation*}
where  $\phi_n (a) = \int_0^{1} s (1- s)^{n-1}  a \lambda(a s)\de s$.
\end{proposition}
\begin{proof}
    From the result showed in Lemma \ref{lemma:conditional_scaled_random_measure}, we know that conditionally on a known value of $\Delta_{1,h}=a$, the random measure $\mu_{\Delta_{1,h}}$ is completely random. Therefore, we can exploit the result in \citet[Proposition 3.1]{james2017bayesian} to characterize the marginal distribution of the feature counts $m_{N,1},\ldots,m_{N,K_N}$. This is given by
    \begin{equation} \label{eq:lik_Z1N}
    \begin{split}
 &       p_k^{(N)}(m_1, \ldots , m_k \mid \Delta_{1,h} = a)\\
 & \qquad \qquad = \exp\left\{-\sum_{n=1}^N \phi_n (a)\right\}  \prod_{i=1}^{k}  \left\{\int_0^1 
    s^{m_i} (1-s)^{N-m_i} a\lambda (a s ) \de s \right\}, 
    \end{split}
    \end{equation}
    with $\phi_n (a) = \int_0^{1} s (1- s)^{n-1}  a \lambda(a s)\de s$. Integrating with respect to $f_{\Delta_{1,h}}$ --- the mixing distribution of $\Delta_{1,h}$ --- yields the desired result.\\
    \qed
\end{proof}
Next, we characterize  the posterior distribution of the random measure $\mu_{\Delta_{1,h}} \sim {\rm SP} (\nu , h)$. The posterior distribution of the law of $\Delta_{1,h}$ is an important ingredient in the study of the predictive properties of the model.  We mention that the posterior characterization of Proposition \ref{prop:general_posterior} is a consequence of \cite[Propositions 2.2]{james2015scaled} and the results developed by  \cite{james2017bayesian}.
\begin{proposition}[Posterior distribution] \label{prop:general_posterior} For any $N \geq 1$, let $Z_{1:N}$ be a random sample modeled as the BNP-Bernoulli model \eqref{eq:SP_model}, where $\mu_{\Delta_{1,h}} \sim {\rm SP} (\nu , h) $.  Suppose that the observations $Z_{1:N}$ display $K_N=k$ distinct features, labelled by $ W_1^*, \ldots , W_{K_N}^* $, with corresponding frequencies $(M_{N,1}, \ldots , M_{N,K_N})= (m_1, \ldots , m_k)$, then the conditional distribution of $\Delta_{1,h}$, given $Z_{1:N}$,  has density function
\begin{equation}  \label{eq:general_posterior_largest_jump}
  g_{\Delta_{1,h}| Z_{1:N}} (a) \propto   \exp\left\{-\sum_{n=1}^N \phi_n (a)\right\} \prod_{i=1}^{k} \left\{ \int_0^1  s^{m_i} (1-s)^{N-m_i} a\lambda (a s ) \de s \right\} f_{\Delta_{1,h}} (a),
\end{equation}
   with $\phi_n (a) = \int_0^{1} s (1- s)^{n-1}  a \lambda(a s)\de s$.
Moreover, the posterior distribution of the random measure $\mu_{\Delta_{1,h}}$, conditionally given $Z_{1:N}$ and $\Delta_{1,h}$, equals
\begin{equation}
        \mu_{\Delta_{1,h}} \mid (\Delta_{1,h}, Z_{1:N})  \overset{d}{=} \mu'_{\Delta_{1,h}} + \sum_{i=1}^{K_N} J_{i} \delta_{W_i^*}, \label{eq:general_posterior_distribution}
\end{equation}
where
\begin{itemize}
    \item[i.] $\mu_{\Delta_{1,h}}'| \Delta_{1,h}  \sim \CRM (\nu_{\Delta_{1,h}}')$ with 
    \begin{align}
       \nu_{\Delta_{1,h}}' (\de s, \de w) =  (1-s)^N \Delta_{1,h} \lambda(s \Delta_{1,h}) \ind_{(0,1)} (s)   \de s \: P(\de w ); \label{eq:general_posterior_levy}
    \end{align}
    \item[ii.]  $J_{ 1:K_N}$ are $K_N$ independent  random jumps and independent of $\mu_{\Delta_{1,h}}'$, with density on $[0,1]$ proportional to
\begin{align}
    f_{J_{i}| \Delta_{1, h}} (s) \propto  (1-s)^{N-m_i} s^{m_i} \Delta_{1,h}\lambda ( \Delta_{1,h} s) . \label{eq:general_posterior_jump}
\end{align}
\end{itemize}
\end{proposition}
\begin{proof}
    Again, leveraging the result showed in Lemma \ref{lemma:conditional_scaled_random_measure}, we know that conditionally on a known value of $\Delta_{1,h}$, the measure $\mu_{\Delta_{1,h}}$ is  completely random. Therefore, we can simply  apply  \citet[Theorem 3.1]{james2017bayesian} to obtain the posterior distribution of $\mu_{\Delta_{1,h}} | (\Delta_{1,h}, Z_{1:N})$ as described in Equation \eqref{eq:general_posterior_distribution}. 
Finally, the posterior distribution of the largest jump $\Delta_{1,h}$ conditionally on the observations $Z_{1:N}$ derived in \Cref{eq:general_posterior_largest_jump} follows by direct application of Bayes' theorem, recognizing  that $f_{\Delta_{1,h}}$ is the prior distribution for $\Delta_{1,h}$, and the distribution in \eqref{eq:lik_Z1N} as the likelihood of the observations $Z_{1:N}| \Delta_{1,h}$.\\
\qed
\end{proof}
Last, we prove the predictive characterization provided in Proposition \ref{prop:general_pred}, which has a pivotal role in our analysis, as it is the conceptual starting point in order to study the predictive behavior of the model, and it again follows form \citep{james2017bayesian}.
\begin{proof}[Proof of Proposition \ref{prop:general_pred}]
We consider $\zeta \stackrel{d}{=} \mu_{\Delta_{1,h}}$, thus we are dealing with the model \eqref{eq:SP_model}.
The posterior  distribution of $\Delta_{1,h}$ in \eqref{eq_mixing} follows from \eqref{eq:general_posterior_largest_jump}, by the argument used in Proposition \ref{prop:general_posterior}.
In order to prove the characterization in Equation \eqref{eq:general_predictive}, we use
once again the fact that conditionally on a known value of $\Delta_{1,h}$, $\mu_{\Delta_{1,h}}$ is a completely random measure (see Lemma \ref{lemma:conditional_scaled_random_measure}). Thus, we can exploit the results in \citep{james2017bayesian} to characterize the predictive distribution of $Z_{N+1}$ given the sample $Z_{1:N}$ and the jump $\Delta_{1,h}$. More specifically the form of the predictive distribution in \eqref{eq:general_predictive} follows by a plain application of \citet[Proposition 3.2]{james2017bayesian}. \\
\qed
\end{proof}

\section{Posterior analysis for SB-SP priors: proofs and details} \label{appendix_scaledbetapriors} 

Here we provide details and proofs of the results in Section \ref{sec:stable}, i.e. a full Bayesian  analysis for the SB-SP prior. More specifically we prove
Theorem \ref{thm:characterization}, then we move to characterize the posterior distribution of $\Delta_{1, h_{c, \beta}}$, marginal, predictive and posterior distributions of the SB-SP model.

\subsection{Proof of \Cref{thm:characterization}}

The posterior density of $\largestjump{1}{h}$, given $Z_{1:N}$, has density proportional to
\begin{equation*}
\prod_{n=1}^N  e^{-\phi_n (a)}   \prod_{i=1}^{K_N} \int_0^1    s^{m_{N,i}} (1-s)^{N-m_{N,i}} a \lambda (a s) \de s \: 
f_{\Delta_{1,h}} (a),
\end{equation*}
where we used the notation $\phi_n (a) = \int_0^{1} s (1- s)^{n-1}  a \lambda(a s)\de s$.
Hence, there exists a normalizing factor $c (m_{N,1}, \ldots , m_{N,k}, N,K_N)$, depending on the sample size $N$, the distinct number of features $K_N$ and the frequency counts, such that
\begin{equation*}
g_{\Delta_{1,h} | Z_{1:N}} (a) = \frac{\prod_{n=1}^N  e^{-\phi_n (a)}  \prod_{i=1}^{K_N} \int_0^1    s^{m_{N,i}} (1-s)^{N-m_{N,i}} a \lambda (a s) \de s \: 
f_{\Delta_{1,h}}(a)}{c (m_{N,1}, \ldots , m_{N,K_N}, N,K_N)},
\end{equation*}
or equivalently we can write
\begin{equation} \label{eq:thm_char_1}
\begin{split}
&g_{\Delta_{1,h} | Z_{1:N}}^{-1} (a) \prod_{n=1}^N  e^{-\phi_n (a)}  \prod_{i=1}^{K_N} \int_0^1    s^{m_{N,i}} (1-s)^{N-m_{N,i}} a \lambda (a s) \de s\: f_{\Delta_{1,h}} (a)\\
& \qquad\qquad\qquad\qquad\qquad\qquad\qquad\qquad\qquad
= c (m_{N,1}, \ldots , m_{N,K_N}, N,K_N).
\end{split}
\end{equation}
If the posterior density $g_{\Delta_{1,h} | Z_{1:N}} (a)$ does not depend on $m_{N,1}, \ldots, m_{N,K_N}$, then the function
\[
    g_{\Delta_{1,h} | Z_{1:N}}^{-1} (a) \prod_{n=1}^N  e^{-\phi_n (a)} \: g(a) = f_1(a,K_N,N)
\]
depends only on $K_N,N$ and $a$, but not  on the frequency counts.
Therefore, \eqref{eq:thm_char_1} boils down to
\begin{equation} \label{eq:thm_char_2}
f_1 (a,K_N,N) \cdot  \prod_{i=1}^{K_N} \int_0^1  s^{m_{N,i}} (1-s)^{N-m_{N,i}} a \lambda (a s) \de s = c (m_{N,1}, \ldots , m_{N,K_N}, N,K_N).
\end{equation}
As a consequence, the  function on the right hand side of \eqref{eq:thm_char_2} is independent of $a$, for any choice of the vector $(m_{N,1}, \ldots , m_{N,K_N}, N,K_N)$.
Now we consider $m_{N,1} = \cdots = m_{N,K_N} =m >0$, and we can say that the function
\begin{equation} \label{eq:thm:eq_w}
\left[ w(a,K_N,N) \int_0^1    s^{m} (1-s)^{N-m} a \lambda (a s) \de s \right]^{K_N}
\end{equation}
does not depend on $a\in \R_+$, where $w(a,K_N,N) = \sqrt[K_N]{f_1(a,K_N,N)}$.  We now select $m= N$, thus the function
\begin{equation} \label{eq:thm_char_3}
w(a,K_N,N) \int_0^1 s^N a \lambda (as) \de s 
\end{equation}
does not depend on $a\in \R_+$.  Note that, since $f_{\Delta_{1,h}}$ and $\lambda$ are functions of class $C^1 (\R_+)$, i.e., derivable with continuous derivative, 
also $w $ is in class $C^1 (\R_+)$ with respect to the variable $a$. Thus, we can take the derivative of \eqref{eq:thm_char_3}, and this is equal to $0$:
\[
\frac{\de }{\de a} w(a,K_N,N) \int_0^a s^N \lambda (s) \de s a^{-N}-N a^{-N-1} w (a,K_N,N)  \int_0^a s^N \lambda (s) \de s + w (a,K_N,N) \lambda (a) =0
\]
which is an ordinary differential equation in $w$, and it can be easily solved by separation of variables, thus obtaining
\[
    w(a,K_N,N) =a^N  \cdot \frac{R}{\int_0^a s^N \lambda (s) \de s}
\]
where $R >0$ is a suitable constant independent of $a$. As a consequence, the function in \eqref{eq:thm:eq_w} equals
\begin{equation*} 
    \left[ \frac{R}{\int_0^1 s^N \lambda (a s) \de s} \cdot \int_0^1    s^{m} (1-s)^{N-m}  \lambda (a s) \de s  \right]^{K_N}
\end{equation*}
and this is independent of $a \in \R_+$. It is possible to choose $m=N-1$ in the previous function, and we can state that
\[
\int_0^1 s^{N-1} \lambda (as) \de s   -\int_0^1 s^N \lambda (as) \de s  = C  \int_0^1  s^N  \lambda (as  )\de s
\]
where $C$ is constant with respect to $a$. If one takes the derivative of the previous equation two times with respect to $a$, then she obtains
\[
\lambda (a) (1-NC)= a \lambda' (a) C ,
\]
which is an ordinary differential equation in $\lambda$ that can be solved by separation of variables. In particular we get the following result
\begin{equation}
    \lambda (a) = \alpha  a^{(1-NC )/C}, \quad\text{for } \alpha >0. \label{eq:lambda_general}
\end{equation}
The exponent of $a$  in \eqref{eq:lambda_general} should satisfy 
\[
    \int_0^{+\infty} \min \{1,a  \} \lambda (a) \de a < +\infty,
\]
from which it is easy to realize that $-2< (1-NC )/C<-1$, hence
\[
    \lambda (a) = \alpha \frac{1}{a^{1+\sigma}}
\]
where $\alpha >0$ and $\sigma \in (0,1)$. The reverse implication of the theorem is trivially true, hence the proof is completed.\\
\qed

\subsection{Detailed derivation of the distribution of $\Delta_{1, h_{c, \beta}}$} \label{app:stable}

We first derive explicitly the distribution of the largest jump given in \Cref{eq:stable_largest_jump}. This follows from direct application of the law of the largest jump,
\begin{align*}
    F_{\Delta_1}(\de a) = \exp\left\{ - \int_a^{\infty} \lambda_\sigma(s)\de s \right\} \lambda_\sigma(a)\de a 
\end{align*} 
when the L\'{e}vy measure is 
\[
    \lambda_\sigma (s)\de s = \sigma s^{-\sigma-1}\ind_{\R_+} (s)\de s.
\]
Having denoted by $f_{\Delta_1}$ the density function of
$F_{\Delta_1}$, we get
\begin{align}
    f_{\Delta_1} (a) & = \lambda_\sigma (a)  e^{-\Lambda (a)}  \ind_{\R_+} (a) =  \sigma  a^{-\sigma-1} \exp \left\{ -\int_a^\infty  \sigma u^{-1-\sigma} \de u \right\} \ind_{\R_+} (a) \nonumber \\
                    & = \sigma  a^{-\sigma-1} e^{-a^{-\sigma}} \ind_{\R_+} (a)\nonumber. 
\end{align}
From direct inspection, we recognize that this is the density function of $\Delta_1 = T^{-1/\sigma}$, where $T$ is a Gamma with parameters $(1,1)$. 
The mixing measure is then obtained by tilting the density $f_{\Delta_1}$ as follows: 
\[
    f_{\Delta_{1, h_{c, \beta}}}(a) \propto f_{\Delta_1}(a)h_{c, \beta}(a) = \sigma a^{-\sigma(c+1)-1}\exp\left\{-\beta a^{-\sigma}\right\}\ind_{\R_+} (a),
\]
i.e.\ letting 
\[
    h_{c, \beta}(a) \propto a^{-\sigma c} \exp\left\{ - (\beta-1) a^{-\sigma} \right\}.
\]
By integration, we get the normalizing constant:
\[
  \int_0^\infty a^{-\sigma(c+1)-1}\exp\left\{-\beta a^{-\sigma}\right\} \de a = \frac{\Gamma(c+1)}{\sigma \beta^{c+1}}.
\]
from which
\begin{align}
    f_{\Delta_{1, h_{c, \beta}}}(a) = \frac{\sigma \beta^{c+1}}{\Gamma(c+1)} a^{-\sigma(c+1)-1}\exp\left\{-\beta a^{-\sigma}\right\}\ind_{\R_+} (a). \label{eq:stable_mixing}
\end{align}

\subsection{Posterior distribution of SB-SP priors}

Here we characterize the posterior distribution of SB-SP priors: the result is not included in the paper, but we think it is useful to have a full picture on SB-SP priors from a Bayesian viewpoint.
\begin{proposition}\label{prop:stable_posterior} 
For $N\geq1$ let $Z_{1:N}$ be a random sample modeled as the BNP-Bernoulli model \eqref{exch_mod}, with $\zeta\sim\text{\rm SB-SP}(\sigma,c,\beta)$. If $Z_{1:N}$ displays $K_{N}=k$ distinct features $\{W_{1}^{\ast},\ldots,W^{\ast}_{K_{N}}\}$, each feature $W_i^*$ appearing exactly $M_{N,i}=m_i$ times in the samples, then the conditional distribution of $\Delta_{1,h_{c,\beta}}$, given $Z_{1:N}$, has a density function of the form
\begin{equation}\label{eq:stable_largest_posteriormain}
g_{\Delta_{1,h_{c,\beta}}\,|\,Z_{1:N}}(a)=\sigma\frac{(\beta+\gamma_{0}^{(N)})^{k+c+1}}{\Gamma(k+c+1)}a^{-k\sigma-(c+1)\sigma-1}\exp\{-a^{-\sigma}(\beta+\gamma_{0}^{(N)})\},
\end{equation}
where $\gamma_{0}^{(n)}=\sigma\sum_{1\leq i\leq n}B(1-\sigma,i)$, with $B(\cdot,\cdot)$ denoting the (standard) Beta function. Moreover, the conditional distribution of $\zeta$, given $(\Delta_{1,h_{c,\beta}},\,Z_{1:N})$, coincides with the distribution of
\begin{equation}\label{eq:stable_posterior_distribution}
        \zeta\,|\, (\Delta_{1,h_{c,\beta}}, Z_{1:N})  \overset{d}{=} \mu^{\prime}_{\Delta_{1,h_{c,\beta}}} + \sum_{i=1}^{K_N} J_{i} \delta_{W_i^{\ast}}, 
\end{equation}
where: 
\begin{itemize}
\item[i)] $\mu^{\prime}_{\Delta_{1,h_{c,\beta}}}$  is a discrete random measure such that $\mu^{\prime}_{\Delta_{1,h_{c,\beta}}}\,|\,\Delta_{1,h_{c,\beta}}\sim\CRM(\nu^{\prime}_{\Delta_{1,h_{c,\beta}}})$, with $ \nu^{\prime}_{\Delta_{1,h_{c,\beta}}}$ being
\begin{equation}\label{eq:levy_stb}
 \nu^{\prime}_{\Delta_{1,h_{c,\beta}}}(\de s,\,\de w)=      \Delta_{1,\Delta_{1,h_{c,\beta}}}^{-\sigma} (1-s)^N \sigma s^{-1-\sigma}  \ind_{(0,1)} (s) \de s P(\de w);
\end{equation}
\item[ii)]
\begin{equation}\label{post_jumps}
J_{i}| \Delta_{1,h_{c,\beta}}\sim\Beta(m_{i}-\sigma,N-m_{i}+1),
\end{equation}
where $\Beta$ denotes the beta distribution.
\end{itemize}
\end{proposition}
\begin{proof}
We apply Proposition \ref{prop:general_posterior}, which describes the general posterior distribution of a SP process.
We first compute the posterior distribution \eqref{eq_mixing} of the largest jump conditionally on observations $Z_{1:N}$. To do so we specify \eqref{eq:general_posterior_largest_jump} in our case, and we first 
 compute the exponent $\phi_n(a)$. In our case the L\'{e}vy density  equals $\lambda_\sigma(s) = \sigma s^{-\sigma-1}$ and the mixing density of $\Delta_{1,h_{c, \beta}}$ is provided in Equation \eqref{eq:stable_mixing}, thus the exponent $\phi_n$ takes the form
\begin{equation} \label{eq:phi_stable}
    \phi_n(a) = \sigma \int_0^1 s(1-s)^{n-1} a^{-\sigma} s^{-\sigma-1}\de s = \sigma a^{-\sigma} B(1-\sigma, n) .
\end{equation}
Recalling the shorthand notation $\gamma_0^{(N)}= \sigma \sum_{1 \leq n \leq N} B(1-\sigma, n)$, the posterior distribution of $\Delta_{1, h_{c, \beta}}$ is then proportional to
\begin{align*}
    a^{k}\exp \left\{ -a^{-\sigma} \stablenews{0}{N} \right\} \prod_{i=1}^{k} \int_0^1  t^{m_{i}} (1-t)^{N-m_{i}} 
    \lambda_\sigma (a t ) \de t \: f_{\Delta_{1, h_{c, \beta}}} (a) \\
        \propto    a^{-\sigma (k+c+1) -1 }\exp \left\{ -  a^{-\sigma} \left[\beta + \stablenews{0}{N} \right]\right\},
\end{align*}
where $f_{\Delta_{1, h_{c, \beta}}}$ has been specified in \eqref{eq:stable_mixing}.
As a consequence we get
    \[
        \largestjump{1}{h_{c, \beta}}^{-\sigma} \mid Z_{1:N} \sim \Gammad\left(k+c+1, \beta+\stablenews{0}{N}\right),
    \]
which corresponds to the posterior density in \eqref{eq:stable_largest_posteriormain}. The characterization of the posterior distribution in \eqref{eq:stable_posterior_distribution} is an easy consequence of Proposition \ref{prop:general_posterior}, by a specialization of this result with the choice $\lambda(s) = \lambda_\sigma (s) = \sigma s^{-\sigma-1}$ for the underlying L\'evy intensity.
\qed
\end{proof}

\subsection{Proof of Proposition \ref{prop:stable_predictive}}

The predictive characterization is a simple consequence of the general characterization in Proposition \ref{prop:general_pred} with the SB-SP specifications $\lambda(s) = \lambda_\sigma (s) = \sigma s^{-\sigma-1}$.
\qed

\subsection{Proof of Proposition \ref{prop:stable_marginal}}
We apply Proposition \ref{prop:general_marginal} to obtain the marginal distribution for the SB-SP prior.
    Conditionally on $\largestjump{1}{h_{c, \beta}}=a$, using the form $\phi_n(a)$ derived in \eqref{eq:phi_stable}, the marginal distribution is given by
    \begin{align*}
         p_k^{(N)}(m_{1},\ldots, m_{k} \mid \largestjump{1}{h_{c, \beta}} = a) &= (\sigma a^{-\sigma})^{k}\exp\left\{- a^{-\sigma} \stablenews{0}{N} \right\}  \prod_{i=1}^{k} \int_0^1 
        s^{m_{i}-\sigma -1 } (1-s)^{N-m_{i}} \de s ,
    \end{align*}
that may be written in terms of the Beta function as follows
    \begin{align*}
         p_k^{(N)}(m_{1},\ldots, m_{k} \mid \largestjump{1}{h_{c, \beta}} = a) &=  (\sigma a^{-\sigma})^{k}\exp\left\{- a^{-\sigma} \stablenews{0}{N}\right\}  \prod_{i=1}^{k} B(m_{i}-\sigma, N-m_{i}+1).
    \end{align*}
    Last, we obtain the marginal distribution in \Cref{eq:stable_marginal} by randomizing with respect to the mixing distribution of the largest jump given in \Cref{eq:stable_mixing}. We need to compute
    \begin{align*}
          p_k^{(N)}(m_{1},\ldots, m_{k}) &= \int_0^\infty  p_k^{(N)}(m_{1},\ldots, m_{k} \mid \largestjump{1}{h_{c, \beta}} = a) f_{\Delta_{1, h_{c, \beta}}}(a) \de a \\
                    &= \frac{\sigma^{k+1} \beta^{c+1}}{\Gamma(c+1)}   \prod_{i=1}^{k} B(m_{i}-\sigma, N-m_{i}+1) \\
                    & \times \int_0^\infty a^{-\sigma(k+c+1)-1} \exp\left\{- a^{-\sigma} \left[\beta+ \stablenews{0}{N}\right] \right\} \de a \\
                    &=  \frac{ \sigma^{k} \beta^{c+1}}{(\beta+\stablenews{0}{N})^{k+c+1}} \frac{\Gamma(k+c+1)}{\Gamma(c+1)}   \prod_{i=1}^{k} B(m_{i}-\sigma, N-m_{i}+1),
    \end{align*}
and the thesis now follows.
\qed

\section{Estimation of the unseen features via SB-SP priors: proofs}  \label{appendix_extrapolation}

Here we detail the proofs of Section \ref{sec:extrapolation}, which is devoted to the unseen-features problem under the SB-SP prior.


\subsection{Proof of Theorem \ref{thm:stable_news}}
We first focus on the proof of \eqref{eq:stable_news}, i.e. the posterior distribution of $U_N^{(M)}$. In order to do this we exploit the predictive characterization provided in Proposition \ref{prop:stable_predictive} to evaluate the probability generating function (PGF) of the random variable $\unseen{N}{M}$ \textit{a posteriori}, conditionally on the sample $Z_{1:N}$. We denote the PGF as $\G_{\unseen{N}{M} } (\, \cdot \, )$. If $t$ belongs to a neighborhood of the origin, then  one has
\begin{equation} \label{eq:thm1.1}
\G_{\unseen{N}{M}}  (t)  =  \E \left[  t^{\unseen{N}{M}}  \mid Z_{1:N}  \right]
= \E \left[ \E \left[  t^{\unseen{N}{M}}  \mid Z_{1:N}  , \Delta_{1,h_{c,\beta}} \right] \mid Z_{1:N} \right] 
\end{equation}
where we have applied the tower property of the conditional expectation. We now  observe that, conditionally on
$Z_{1:N}$ and $\Delta_{1,h_{c,\beta}}$, the random variable $\unseen{N}{M}$ may be represented as 
\[
\unseen{N}{M} |( Z_{1:N}  , \Delta_{1,h_{c,\beta}}  )  \stackrel{{\rm d}}{=}
\sum_{i \geq 1} \ind\left( \sum_{m=1}^M A_{N+m,i}' >0 \right),
\]
where we used the representation given in Proposition \ref{prop:stable_predictive}. Here, independently across $i$, $A_{N+m,i}'$ is a Bernoulli random variable with parameter $\rho_i'$, conditionally  on the random measure
$\mu_{\Delta_{1,h_{c,\beta}}}'= \sum_{i \geq 1} \rho_i' \delta_{W_i'}$ with L\'evy intensity
$\sigma \Delta_{1,h_{c,\beta}}^{-\sigma} (1-s)^N s^{-1-\sigma} \ind_{(0,1)} (s) \de s P(\de w)$.
We now focus on the evaluation of the expected value in \Cref{eq:thm1.1}:
\begin{align*}
    &\E \left[  t^{\unseen{N}{M}}  \mid Z_{1} ,\ldots , Z_N  , \Delta_{1,h_{c,\beta}} \right] \\ 
      & \qquad= \E \left[  \E \left[ \prod_{i \geq 1} \left( (t-1)\ind\left\{ \sum_{m=1}^M A_{N+m,i}' >0\right\} + 1\right)  \Big| \mu_{\Delta_{1, h_{c, \beta}}}' \right]\right] \\
    & \qquad = \E \left[ \prod_{i \geq 1}  \left[ (t-1) \P \left(\sum_{m=1}^M A_{N+m,i}' >0 \mid \mu_{\Delta_{1, h_{c, \beta}}}'\right)+1   \right]  \right] \\
    & \qquad =  \E \left[   \prod_{i \geq 1}  \left[ (t-1) \left\{1- \prod_{m=1}^M \P ( A_{N+m,i}' =0  \mid \mu_{\Delta_{1, h_{c, \beta}}}' ) \right\}+1   \right] \right] ,
\end{align*}
where we applied the independence of the Bernoulli random variables $A_{N+m,i}'$s, conditionally on 
$\mu_{\Delta_{1, h_{c, \beta}}}'$.
We now recall that  $\mu_{\Delta_{1, h_{c, \beta}}}'$ is a CRM with a known L\'evy measure and  that the
$A_{N+m,i}'$s are Bernoulli with parameter $\rho_i'$ to obtain
\begin{align*}
    & \E \left[  t^{\unseen{N}{M}}  \mid Z_{1} ,\ldots , Z_N  , \Delta_{1,h_{c,\beta}}\right] \\
     & \qquad =
    \E \left[ \prod_{i \geq 1}   ( (t-1) (1-(1- \rho_i')^M) +1   )   \right] \\
    & \qquad =  \E \left[ \exp  \left\{  \sum_{i \geq 1}  \log \left[ (t-1) (1-(1- \rho_i')^M) +1   \right] \right\}  \right] \\
    & \qquad = \exp \left\{ -  (1-t) \int_0^1 (1-(1-s)^M)  (1-s)^N  \Delta_{1, h_{c, \beta}}^{-\sigma}  \sigma
    s^{-1-\sigma} \de s   \right\}.\\
    &\qquad =  \exp \left\{ - (1-t)  \Delta_{1, h_{c , \beta}}^{-\sigma} \stablenews{N}{M} \right\},
\end{align*}
where we used the identity
\[
    \int_0^1 \left[1-(1-s)^M \right]  (1-s)^N   s^{-1-\sigma} \de s = \sum_{m=1}^M B(1-\sigma, N+m).
\]
We replace this expression in \Cref{eq:thm1.1} to obtain
\begin{align}
    \label{eq:thm1.2}
    \G_{\unseen{N}{M}}  (t) = \E [  \exp \{ - (1-t)  \Delta_{1,h_{c,\beta}}^{-\sigma} \stablenews{N}{M} \} \mid Z_{1:N}].
\end{align}
The results now follows by integrating with respect to the posterior distribution of $\Delta_{1,h_{c,\beta}}^{-\sigma}$, given in  Equation \eqref{eq:stable_largest_posteriormain}:
\begin{align*}
    \G_{\unseen{N}{M}}  (t) 
    &= \frac{(\beta+\stablenews{0}{N})^{K_N+c+1}}{\Gamma (K_N+c+1)}  \int_0^\infty \exp\left\{-(1-t) \stablenews{N}{M} x\right \} 
    x^{K_N+c} e^{-(\beta+\stablenews{0}{N} )x} \de x\\
    & = \frac{(\beta+\stablenews{0}{N})^{K_N+c+1}}{\Gamma (K_N+c+1)} \frac{\Gamma (K_N+c+1)}{(\beta+\stablenews{0}{N} +(1-t) \stablenews{N}{M})^{K_N+c+1}} \\
    & = \left(\frac{\beta+\stablenews{0}{N}}{\beta+\stablenews{0}{N+M} -t \stablenews{N}{M}} \right)^{K_N+c+1}
     =  \left( \frac{1-p_N^{(M)}}{1 -tp_N^{(M)}} \right)^{K_N+c+1},
\end{align*}
 for any $|t| < 1/ p_N^{(M)}$, where $p_N^{(M)} := \stablenews{N}{M}/(\beta+\stablenews{0}{N+M})\leq 1$. This is the probability generating function of a negative binomial distribution where 
 $K_N+c+1$ is the number of failures, and $p_N^{(M)}$ is the success probability
 in each experiment.\\

We now apply similar arguments to derive the posterior distribution of $U_N^{(M,r)}$, provided in \eqref{eq:stable_news_freq}. Again, we calculate the probability generating function of $\unseenfreq{N}{M}{r}$ \textit{a posteriori}, denoted here as $\G_{\unseenfreq{N}{M}{r}} (\, \cdot \, )$. If $t$ belongs to a neighborhood of the origin, then  one has
\begin{equation} \label{eq:thm2.1}
\begin{split}
& \G_{\unseenfreq{N}{M}{r}}  (t) =  \E \left[  t^{\unseenfreq{N}{M}{r}}  \mid Z_{1:N} \right]
= \E \left[ \E \left[  t^{\unseenfreq{N}{M}{r}}  \mid Z_{1:N}  , \largestjump{1}{h_{c,\beta}} \right] \mid Z_{1:N} \right] .
\end{split}
\end{equation}
It is now easy to see that,  conditionally on
$Z_{1} ,\ldots , Z_N  , \largestjump{1}{h_{c,\beta}} $, the random variable $\unseenfreq{N}{M}{r}$ may be written as
\[
\unseenfreq{N}{M}{r} | Z_{1} ,\ldots , Z_N  , \largestjump{1}{h_{c,\beta}} \stackrel{{\rm d}}{=} \sum_{i \geq 1} \ind\left\{ \sum_{m=1}^M A_{N+m,i}' =r\right\}
\]
by applying Proposition \ref{prop:stable_predictive}. With the same notation used in the first part of the proof, we recall that the $A_{N+m,i}'$s are independent Bernoulli variables with parameters $\rho_i'$, conditionally  on the CRM
$\mu_{\largestjump{1}{h_{c,\beta}}}'= \sum_{i \geq 1} \rho_i' \delta_{W_i'}$ with L\'evy intensity
$\sigma \largestjump{1}{h_{c,\beta}}^{-\sigma} (1-s)^N s^{-1-\sigma} \ind_{(0,1)} (s) \de s P(\de w)$. Across $i$, the random variables
\[
    S_{M,i} := \sum_{m=1}^M A_{N+m,i}' 
\]
are independent, each one distributed as a binomial with parameters $M$ and success probability $\rho_i'$. We then evaluate the expected value appearing in \eqref{eq:thm2.1} as follows:
\begin{align*}
    \E \left[  t^{\unseenfreq{N}{M}{r}}  \mid Z_{1:N}  , \largestjump{1}{h_{c,\beta}} \right]  
    &=  \E \left[ \E \left[ \prod_{i \geq 1} \left( (t-1)\ind\left\{ \sum_{m=1}^M A_{N+m,i}' =r\right\} + 1\right)  \Big| \mu_{\Delta_{1, h_{c, \beta}}}' \right]   \right]\\
    &  = \E \left[ \prod_{i \geq 1}  \left[ (t-1) \P (S_{M,i} =r \mid \mu_{\Delta_{1, h_{c, \beta}}}')+1   \right]  \right] \\
    &   =  \E \left[   \prod_{i \geq 1}  \left[ (t-1) \binom{M}{r} (\rho_i')^r (1-\rho_i')^{M-r}  +1\right]\right].
\end{align*}
Since  $\mu_{\Delta_{1, h_{c, \beta}}}'$ is a CRM with a known L\'evy measure, we can evaluate the previous expected value:
\begin{align*}
&  \E \left[  t^{\unseenfreq{N}{M}{r}}  \mid Z_{1} ,\ldots , Z_N  , \largestjump{1}{h_{c,\beta}}  \right]  \\
 & =
   \E \left[ \exp  \left\{  \sum_{i \geq 1}  \log \left( (t-1) \binom{M}{r} (\rho_i')^r (1-\rho_i')^{M-r} +1 \right) \right\}  \right] \\
&  = \exp \left\{ -  (1-t)  \binom{M}{r} \int_0^1  s^{r-\sigma-1} (1-s)^{M+N-r}   \de s \sigma \Delta_{1, h_{c, \beta}}^{-\sigma} \right\}\\
&  =  \exp \left\{ - (1-t) \Delta_{1, h_{c, \beta}}^{-\sigma} \sigma \binom{M}{r}   B(r-\sigma , M+N -r+1) \right\} \\
&=  \exp \left\{ - (1-t) \Delta_{1, h_{c, \beta}}^{-\sigma} \stablenewsfreq{N}{M}{r} \right\},
\end{align*}
where we used the notation introduced in the statement of the theorem, i.e. $\stablenewsfreq{N}{M}{r} = \sigma \binom{M}{r}   B(r-\sigma , M+N -r+1)$.
Then the probability generating function in \Cref{eq:thm2.1} is obtained by integrating with respect to the posterior distribution of the largest jump provided in \eqref{eq:stable_largest_posteriormain}:
\begin{align*}
    \G_{\unseenfreq{N}{M}{r}}  (t) & = \int_0^\infty  \exp \left\{ - (1-t) x \stablenewsfreq{N}{M}{r}\right\} \cdot
    \frac{(\beta+\stablenews{0}{N})^{K_N+c+1}}{\Gamma (K_N+c+1)} x^{K_N+c} e^{-(\beta+\stablenews{0}{N})x} \de x\\
    & = \frac{\Gamma (K_N+c+1)}{(\beta+\stablenews{0}{N}+(1-t) \stablenewsfreq{N}{M}{r})^{K_N+c+1}} \cdot
    \frac{(\beta+\stablenews{0}{N})^{K_N+c+1}}{\Gamma (K_N+c+1)} \\
    & = \left( \frac{ \beta+\stablenews{0}{N}}{\beta+\stablenews{0}{N}+\stablenewsfreq{N}{M}{r} -t \stablenewsfreq{N}{M}{r}} \right)^{K+c+1} =  \left( \frac{1-p_N^{(M,r)}}{1 -tp_N^{(M,r)}} \right)^{K_N+c+1}
\end{align*}
 for any $|t| < 1/ p_N^{(M,r)}$, where we have set 
 \[
    p_N^{(M,r)}:= \frac{\stablenewsfreq{N}{M}{r}}{\beta+\stablenewsfreq{N}{M}{r}+\stablenews{0}{N}}. 
\]
 Then we conclude 
 that the posterior distribution of $\unseenfreq{N}{M}{r}$ is  a negative binomial distribution where 
 $K_N+c+1$ is the number of failures, and $p_N^{(M,r)}$ is the success probability
 in each experiment.\\
\qed


\subsection{Proof of Theorem \ref{thm:stable_conv}}
In order to prove this result, we first exploit the L\'{e}vy continuity theorem, to obtain a convergence in distribution, and later strengthen this result to show that the convergence holds true also in the almost-sure sense.
For the convergence in distribution, thanks to \Cref{thm:stable_news}, the characteristic function of $\unseen{N}{M}/M^\sigma\mid Z_{1:N} $ is given by
\begin{align*}
    \Phi_{\unseen{N}{M}/M^\sigma} (t) & = \left( \frac{1-p_{N}^{(M)}}{1-p_N^{(M)} e^{i t /M^\sigma}}  \right)^{K_N+c+1}
    \end{align*}
    where $t \in \R$, $K_N$ is the number of distinct features in $Z_{1:N}$ and $p_N^{(M)} = \stablenews{N}{M}/(\stablenews{0}{N}+\stablenews{N}{M} +\beta )$. The quantity above can be rewritten as
    \begin{align*}
        \Phi_{\unseen{N}{M}/M^\sigma} (t) & = \left(  \frac{\beta +\stablenews{0}{N}}{\beta +\stablenews{0}{N}+\stablenews{N}{M}  - \stablenews{N}{M}e^{i t/M^\sigma}}  \right)^{K_N+c+1}.
\end{align*}
We can exploit \citet[Lemma 1]{masoero2019more} to determine the asymptotic expansion $\stablenews{N}{M}= M^\sigma  \Gamma (1-\sigma)  (1+O (M^{-\sigma}))  $  as $M \to +\infty$, having used the big-$O$ notation. Thus, using the asymptotic expansion of the exponential function, one has
\begin{align*}
\Phi_{\unseen{N}{M}/M^\sigma} (t) & =  \left(  \frac{\beta+\stablenews{0}{N}}{\beta+\stablenews{0}{N}+\stablenews{N}{M}  - \stablenews{N}{M}(1+it M^{-\sigma} + O(M^{-2\sigma}))}  \right)^{K_N+c+1} \\
& =  \left(  \frac{\beta+\stablenews{0}{N}}{\beta+\stablenews{0}{N}- \stablenews{N}{M} i t M^{-\sigma}  +O(M^{-\sigma})}  \right)^{K_N+c+1}   \\
& =  \left(  \frac{\beta+\stablenews{0}{N}}{\beta+\stablenews{0}{N} -it \Gamma (1-\sigma) +O(M^{-\sigma})}  \right)^{K_N+c+1}   \\
\end{align*}
which converges to the characteristic function of a gamma random variable with parameters $(K_N+c+1, (\stablenews{0}{N}+\beta)/\Gamma(1-\sigma))$ as $M \to +\infty$. This proves that 
\[
    \unseen{N}{M}/M^\sigma \mid Z_{1:N} \overset{\rm{ d }}{\to} W_N, \quad \text{where } W_N \sim \Gammad\left(K_N+c+1, \frac{\beta+\stablenews{0}{N}}{\Gamma(1-\sigma)}\right).
\]
In order to prove convergence in the almost sure sense, we exploit the corresponding results proved for the stable beta-Bernoulli process in \citet[Theorem 2]{masoero2019more} for the statistic
$\unseen{N}{M}$. We first notice that if we condition on the value of the largest jump $\largestjump{1}{h_{c, \beta}}$, then the SB-SP-Bernoulli is a completely random measure whose asymptotic behavior is analogous to the stable beta-Bernoulli process. Thus, specializing the almost sure convergence results given in \citet[Theorem 2]{masoero2019more}, a posteriori, we have
\begin{equation}
    \P \left( \lim_{M \to +\infty } \frac{\unseen{N}{M}}{M^\sigma}  = a^{-\sigma} \Gamma (1-\sigma) \Big|  Z_{1:N}, \largestjump{1}{h_{c, \beta}} = a\right) =1 . \label{eq:thm_as_1}
\end{equation}
The probability limit for the model in which the largest jump is random is obtained by observing that
\begin{align*} 
\begin{split}
&\P \left( \lim_{M \to +\infty } \frac{\unseen{N}{M}}{M^\sigma}  = \largestjump{1}{h_{c, \beta}}^{-\sigma} \Gamma (1-\sigma) \Big|  Z_{1:N} \right) \\
&\quad\quad\quad= \E \left[ \P \left( \lim_{M \to +\infty } \frac{\unseen{N}{M}}{M^\sigma}  = \largestjump{1}{h_{c, \beta}}^{-\sigma} \Gamma (1-\sigma) \Big|  Z_{1:N}, \largestjump{1}{h_{c, \beta}}\right) 
\Big| Z_{1:N} \right] \stackrel{\eqref{eq:thm_as_1}}{=} 1,
\end{split}
\end{align*}
in other words $\unseen{N}{M}/M^\sigma $ converges almost surely to the random variable $\largestjump{1}{h_{c, \beta}}^{-\sigma} \Gamma (1-\sigma)$, with respect to the conditional probability $\P$ given $Z_{1:N}$. Note also that the posterior distribution of $\largestjump{1}{h_{c, \beta}}^{-\sigma} \Gamma (1-\sigma)$ is a Gamma with parameters 
\[
  \left(K_N+c+1, \frac{\beta+\stablenews{0}{N}}{\Gamma (1-\sigma)}\right),
\]
thus the a.s. convergence in \eqref{eq:conv_as} now follows.\\

We proceed along the same lines as to show the validity of \eqref{eq:conv_as_freq}. First, we show the convergence in distribution of $\unseenfreq{N}{M}{r}$ using the characteristic function, and then we show that the result also holds in an almost sure sense.
From Theorem \ref{thm:stable_news}, the characteristic function of $\unseenfreq{N}{M}{r}/M^\sigma \mid Z_{1:N} $ is given by
\begin{equation*}
    \Phi_{\unseenfreq{N}{M}{r}/M^\sigma} (t) = \left( \frac{1-p_{N}^{(M,r)}}{1-p_N^{(M,r)} e^{i t /M^\sigma}}  \right)^{K_N+c+1}
\end{equation*}
where $t \in \R$, and $p_N^{(M)} = \stablenewsfreq{N}{M}{r}/(\stablenews{0}{N}+\stablenewsfreq{N}{M}{r} +\beta)$, and $\stablenewsfreq{N}{M}{r}$ was defined in the statement of Theorem \ref{thm:stable_news}. The expression above is equivalent to
\begin{align*}
    \Phi_{\unseenfreq{N}{M}{r}/M^\sigma} (t) & = \left(  \frac{\beta+\stablenews{0}{N}}{\beta+\stablenews{0}{N} +\stablenewsfreq{N}{M}{r}(1-e^{it/M^\sigma}) } \right)^{K_N+c+1}.
\end{align*}
Thanks to the well-known asymptotic relation for the ratio of gamma functions, it is easy to see that
\[
    \stablenewsfreq{N}{M}{r} = \frac{\sigma}{r!} \Gamma(r-\sigma) (r-\sigma) \frac{\Gamma(M+1)}{\Gamma(M+1-r)} \frac{\Gamma(N+M+2-r)}{\Gamma(N+M+2-\sigma)}  = \frac{\sigma}{r!} \Gamma (r-\sigma) M^{\sigma} (1+O (M^{-1}))
\]
as $M \to +\infty$. Hence, the characteristic function under study boils down to
\begin{align*}
\Phi_{\unseenfreq{N}{M}{r}/M^\sigma} (t)  & =  \left(  \frac{\beta +\stablenews{0}{N}}{\beta+\stablenews{0}{N} +\sigma \Gamma (r-\sigma) (r!)^{-1} M^{-\sigma}
(1+O(M^{-1}) ) (1-e^{it/ M^{\sigma}})} \right)^{K_N+c+1} \\
& =\left(  \frac{\beta +\stablenews{0}{N}}{\beta +\stablenews{0}{N} -\sigma \Gamma (r-\sigma) (r!)^{-1}it +O(M^{-\sigma})} \right)^{K_N+c+1}
\end{align*}
which converges, as $M \to +\infty$, to the characteristic function of a gamma random variable with parameters as in the thesis. The almost sure statement of \eqref{eq:conv_as_freq} goes along similar lines, indeed one can exploit the convergence theorems proved by \cite{masoero2019more} to state that
\begin{equation*}
\P \left( \lim_{M \to +\infty } \frac{\unseenfreq{N}{M}{r}}{M^\sigma}  =\frac{\sigma  (1-\sigma)_{(r-1)}}{r!} \largestjump{1}{h_{c, \beta}}^{-\sigma} \Gamma (1-\sigma) \Big|  Z_{1:N}, \largestjump{1}{h_{c, \beta}}\right) =1 .
\end{equation*}
Exactly as before, one can conclude that
\[
\P \left( \lim_{M \to +\infty } \frac{\unseenfreq{N}{M}{r}}{M^\sigma}  =\frac{\sigma  (1-\sigma)_{(r-1)}}{r!} \largestjump{1}{h_{c, \beta}}^{-\sigma} \Gamma (1-\sigma) \Big|  Z_{1:N}\right) =1,
\]
where the posterior distribution of the limiting random variable is a gamma with the same parameters as in the statement of the theorem  (Equation \eqref{eq:conv_as_freq}).\\
\qed

\section{Multivariate extension}  \label{appendix_multivariate}

In the present section we discuss the multivariate version of the  Bernoulli process, which we call the Bernoulli process with a condiment
or the simple multinomial process, using the terminology of \cite{james2017bayesian}. 
We first revise the model of \cite{james2017bayesian} and the associated prior, called stable-Beta-Dirichlet process, then we move to introduce a new scaled prior for the model. In both the cases, we determine closed-form results to face prediction of new features with condiments. These models are extremely important in  genomics to account for the presence of variants at certain genomic loci with a specific characteristic (or condiment). See, e.g., \cite{lee2016bayesian2}.

\subsection{Bernoulli process with a condiment}

The IBP process with a condiment has been introduced by \cite{james2017bayesian} and we remind the definition here.  For $q=1, 2, \ldots$, we define the vector of probabilities $\bm{p} = (p_1, \ldots , p_q)$ taking values in the following set
\[
S_q = \{   \bm{s}:= (s_1, \ldots, s_q) : \; s_j >0 \text{ as } j =1, \ldots , q, \; |\bm{s}|:= \sum_{j=1}^q s_j < 1 \}
\]
where for a generic vector $\bm{s}$, $|\bm{s}|= \sum_{j=1}^q s_j$ denotes the $L^1$ norm of the vector. For a fixed vector $\bm{p} \in S_q$, we
 also define the \textit{simple multinomial} distribution $\mathsf{M} (1, \bm{p})$. A vector $\bm{A}= (A_1, \ldots , A_q) \in \{0, 1\}^q$ is said to have the 
  \textit{simple multinomial} distribution with parameter vector $\bm{p}$ iff  it has the following probability mass function
\[
\P (\bm{A}= \bm{a}) = \P (A_1= a_1, \ldots , A_q = a_q) =  \left\{ \begin{array}{ll}
\prod_{j=1}^q  p_j^{a_j} \cdot (1-|\bm{p}|)^{1-|\bm{a}|} &  \text{if } |\bm{a}| \leq 1\\
0 & \text{if }   |\bm{a}| >1
\end{array}  \right.
\]
and we will write $\bm{A} \sim \mathsf{M} (1, \bm{p})$. In other words $\bm{A}$ concentrates on the vectors of $\{ 0,1\}^q$ for which at most one element is equal to 1 and all the other entries are zero.\\

The Bernoulli process with a condiment assumes that each observation $\bm{Z}$ is a multivariate $\{0,1\}^q$-valued stochastic process
\[
\bm{Z} (w ) = \sum_{i \geq 1} \bm{A}_i \delta_{w_i} (w)
\]
where $(w_i)_{i \geq 1}$ are features in $\mathds{W}$ and $(\bm{A}_i)_{i \geq 1}$ are independent simple multinomial random variables with parameter vector
$\bm{p}_i = (p_{i,1}, \ldots , p_{i,q})$ as $i =1, 2, \ldots$. Here $|\bm{p}_i|$ represents the probability that an individual displays feature $w_i$, while 
$p_{i,j}$ is the probability that the individual exhibits feature $w_i$ with condiment $j \in \{  1, \ldots , q\}$. Thus, $\bm{Z}$ is termed a \textit{simple mutlinomial process} with parameter $\bm{\zeta}= \sum_{i \geq 1} \bm{p}_i \delta_{w_i}$, and it is denoted by $\mathsf{MP} (\bm{\zeta})$. In order to carry out BNP inference, we need to specify a distribution for the discrete measure $\bm{\zeta}$.
Thus, we obtain a multivariate version of the model \eqref{exch_mod}:
\begin{equation}
\label{eq:model_multivariate}
\begin{split}
\bm{Z}_n | \bm{\zeta} &\simiid \mathsf{MP} (\bm{\zeta}) \quad n =1, \ldots , N\\
\bm{\zeta}  & \sim \mathscr{Z}
\end{split}
\end{equation}
where $\mathscr{Z}$ denotes the distribution of the discrete random measure $\bm\zeta$. 

\subsection{Priors based on multivariate CRMs}  \label{sec:multi_CRM}

In this section we consider a class of priors $\mathscr{Z}$ in \eqref{eq:model_multivariate} defined by \cite{james2017bayesian} and based on a multivariate extension of CRMs (see \cite{daleyII}). In particular consider a multivariate CRM on $\mathds{W}$:
\[
\bm{\mu}  = \sum_{i \geq 1} \bm{\rho}_i \delta_{W_i}
\]
where $\bm{\rho}_i= (\rho_{i,1}, \ldots , \rho_{i,q})$ is a vector of $[0,1]$-valued random jumps with the property $\sum_{i \geq 1} |\bm\rho_i| < +\infty$, the $W_i$'s are i.i.d. $\mathds{W}$-valued random locations independent of the  $\bm{\rho}_i$'s. Under this nonparametric prior each observation $\bm{Z}_n$ in \eqref{eq:model_multivariate} admits the representation $\bm{Z}_n | \bm{\mu}= \sum_{i \geq 1} \bm{A}_{n,i} \delta_{W_i}$, where 
$\bm{A}_{n,i}= (A_{n,i,1}, \ldots , A_{n,i,q})| \bm{\mu} \simind \mathsf{M}(1, \bm{\rho}_i)$.
Note that the random measure $\bm{\mu}$ equals the vector of random measures $(\mu_1, \ldots , \mu_q)$, where
\[
\mu_j = \sum_{i \geq 1} \rho_{i,j} \delta_{W_i}, \quad j=1, \ldots , q.
\]
As a simple CRM of Section \ref{appendix_crm}, the multivariate extension of a CRM is characterized by its L\'evy-Khintchine representation:
\[
\begin{split}
&\E [e^{-\int_{\mathds{W}} f_1 (w) \mu_1 (\de w)- \cdots  - \int_{\mathds{W}} f_q (w) \mu_q (\de w)}] \\
& \qquad= 
\exp \left\{- \int_{\mathds{W}}  \int_{\R_{+}^q} (1-e^{-s_1 f_1 (w)- \cdots - s_q f_q (w)}   )   \lambda_{(q)} (s_1, \ldots , s_q) \de s_1 \cdots \de s_q  P (\de w) \right\}
\end{split}
\]
for arbitrary measurable functions $f_1, \ldots , f_d : \mathds{W} \to \R_+$,  where $P$ is a probability measure on $\mathds{W}$. The multivariate L\'evy intensity  $\lambda_{(q)}$ is assumed to satisfy  the integral condition
\[
\int_{\R_+^q}  \min \{1,||\bm{s}|| \}  \lambda_{(q)} (s_1, \ldots ,s_q) \de s_1 \cdots \de s_q < +\infty
\]
where $||\bm s||$ is the Euclidean norm of the vector $\bm{s}$.
When $\lambda_{(q)} (s_1, \ldots ,s_q) $ concentrates on $S_q$, the law of $\bm{\mu}$ may be employed as a distribution for the parameter $\bm{\zeta}$ of the simple multinomial process in \eqref{eq:model_multivariate}. A possible choice indicated by \cite{james2017bayesian} is to select a stable-Beta-Dirichlet process, which is a generalization of the Beta-Dirichlet process
\citep{kim2012bayesian} with power law behavior. We say that a multivariate CRM $\bm{\mu}= (\mu_1, \ldots, \mu_q)$ is a stable-Beta-Dirichlet process with parameters $(\alpha , \kappa+\alpha; \bm{\gamma}; \vartheta)$, where $\bm{\gamma} = (\gamma_1, \ldots , \gamma_q)$, if it is characterized by the following L\'evy intensity specification
\begin{equation}
\label{eq:SBD_process}
\lambda_{(q)} (\bm{s}) = \frac{\vartheta \Gamma (|\bm{\gamma}|)}{\prod_{j=1}^q  \Gamma (\gamma_j)}  |\bm{s}|^{-\alpha-|\bm{\gamma}|}  (1-\bm{s})^{\kappa+\alpha-1}\prod_{j=1}^q s_j^{\gamma_j-1}\ind_{[0,1]}(|\bm{s}| )
, \; \bm{s} \in S_q
\end{equation}
where  $0 \leq \alpha < 1, \kappa > -\alpha, \vartheta >0 $ and $\gamma_j >0$ for any $j=1, \ldots , q$. 
We write $\bm{\mu} \sim {\rm mSBD} (\alpha, \kappa+\alpha;\bm{\gamma}; \vartheta)$ to denote the distribution of the stable-Beta-Dirichlet process.
As emphasized by \cite{james2017bayesian}, it can be easily checked, by means of the Laplace functional, that $\sum_{j=1}^q \mu_j$
is a stable-Beta process of \cite{teh2009indian}, i.e. a simple CRM on $\mathds{W}$ with L\'evy intensity on $[0,1]\times \mathds{W}$ equal to $\vartheta s^{-\alpha-1}(1-s)^{\kappa+\alpha-1} \de s P (\de w)$.

\subsubsection{Estimation of the unseen features with a condiment}

In order to face predictive inference with the model \eqref{eq:model_multivariate} under the prior specification $\bm{\zeta}  \sim {\rm mSBD} (\alpha, \kappa+\alpha;\bm{\gamma}; \vartheta)$, we need to characterize the predictive distribution of $\bm{Z}_{N+1}| \bm{Z}_{1:N}$ for the model \eqref{eq:model_multivariate}. To this end it is worth recalling the definition of the finite-dimensional Beta-Dirichlet distribution by \cite{kim2012bayesian}.
A random vector $\bm{P}:=(P_1, \ldots , P_q)$ on $S_q$ is said to follow a Beta-Dirichlet distribution with positive parameters $\alpha, \kappa$ and $\bm{\gamma}= (\gamma_1, \ldots , \gamma_q)$ if the probability density function of the random vector $(P_1, \ldots , P_q)$ has density proportional to
\begin{equation}
\label{eq:betaDir_finiteD}
|\bm{s}|^{\alpha-|\bm{\gamma}|}\cdot (1-|\bm{s}|)^{\kappa-1} \prod_{j=1}^q s_j^{\gamma_j-1} \cdot \ind_{S_q}(\bm{s})
\end{equation}
and we write $(P_1, \ldots , P_q)\sim \mathscr{B}\mathscr{D}(\alpha, \kappa; \bm{\gamma})$.
This distribution can be characterized as follows: $|\bm{P}|$ has a Beta distribution with parameters $(\alpha, \kappa)$ and the normalized vector
$(P_1/|\bm{P}|, \ldots , P_q /|\bm{P}|)$ follows a Dirichlet distribution with parameters $(\gamma_1, \ldots, \gamma_q)$.\\

We first characterize the distribution of $\bm{Z}_{N+1}| \bm{Z}_{1:N}$ under the prior specification  $\bm{\zeta}  \sim {\rm mSBD} (\alpha, \kappa+\alpha;\bm{\gamma}; \vartheta)$ in \eqref{eq:model_multivariate}. The following result  is immediate  from the theory developed by \cite{james2017bayesian}.
\begin{theorem} \label{thm:multi_predictive}
For any $N \geq 1$, let $\bm{Z}_{1:N}$ be a random sample modeled as the BNP multinomial process model \eqref{eq:model_multivariate}, with $\bm{\zeta} \sim 
{\rm mSBD} (\alpha, \kappa+\alpha;\bm{\gamma}; \vartheta)$. If $\bm{Z}_{1:N}$ displays $K_N=k$ distinct features, labeled by $W_1^*, \ldots  , W_{K_N}^*$, with condiment-specific frequencies $(M_{N,1,j}, \ldots , M_{N,K_N,j})= (m_{1,j}, \ldots , m_{k,j})$, for any $j=1, \ldots , q$, then the
 conditional distribution of $\bm{Z}_{N+1}$, given $\bm Z_{1:N}$, coincides with the distribution of 
\begin{equation}
\label{eq:multi_predictive}
\bm{Z}_{N+1} | \bm{Z}_{1:N} \stackrel{d}{=} \bm{Z}_{N+1}' + \sum_{i=1}^{K_N} \bm{A}_{N+1,i}  \delta_{W_i^*}
\end{equation}
where:
\begin{itemize}
\item[i)] $\bm{Z}_{N+1}' $ is such that $\bm{Z}_{N+1}' = \sum_{i \geq 1}\bm{A}_{N+1,i}' \delta_{W_i'} \sim \mathsf{MP} (\bm{\mu}')$ and 
$\bm{\mu}'\sim  {\rm mSBD} (\alpha, \kappa+M+\alpha;\bm{\gamma}; \vartheta)$; 
\item[ii)] $\bm{A}_{N+1, 1:K_N}$ is a collection of independent simple multinomial random variables with respective parameters $\bm{J}_{1:K_N}$, such that each $\bm{J}_i = (J_1, \ldots , J_q)$ has a Beta-Dirichlet distribution, i.e., $\bm{J}_i \simind \mathscr{B}\mathscr{D}(m_i-\alpha, N-m_i +\kappa+\alpha; \bm{\gamma}+ \bm{m}_i)$, where we put $\bm{m}_i := (m_{i,1}, \ldots , m_{i,q})$ and $m_i = \sum_{j=1}^q m_{i,j}= |\bm{m}_i|$ for any $i=1, \ldots , K_N$.
\end{itemize}
\end{theorem}
Note that in Theorem \ref{thm:multi_predictive} $M_{N,i,j}$ is the random number of times feature $W_i^*$ has been observed out of $\bm{Z}_{1:N}$ with condiment $j \in \{ 1, \ldots, q\}$, while
$m_i= \sum_{j=1}^q m_{i,j}$ is the number of times feature $W_i^*$ has been observed out of the sample.\\

For any $N \geq 1$, let $\bm{Z}_{1:N}$ be an observable sample modeled as the multinomial model in \eqref{eq:model_multivariate}, with 
$\bm{\zeta} \sim {\rm mSBD} (\alpha, \kappa+\alpha;\bm{\gamma}; \vartheta)$. Moreover, under the same model, for $M \geq 1$ let $\bm{Z}_{N+1:N+M}= (\bm{Z}_{N+1}, \ldots , \bm{Z}_{N+M})$ be an additional and unobserved sample. We now define the number of hitherto unobserved feature with condiment $\ell \in \{ 1, \ldots , q\}$  that will be recorded out of $\bm{Z}_{N+1:N+M}$ as
\begin{equation}
\label{eq:new_condiment_def}
U_{N,\ell}^{(M)} := \sum_{i\geq 1} \ind \left( \sum_{m=1}^M  A_{m,i,\ell} >0 \right) \cdot  \ind \left( \sum_{n=1}^N A_{n,i,\ell}=0 \right).
\end{equation}
Posterior inference for such a quantity could have potential interest in genomics to account for the presence of a variant with certain biological characteristics (condiment).
The next theorem provides the posterior distribution of $U_{N,\ell}^{(M)}$.
\begin{theorem}  \label{thm:multi_condiment_prediction_NON-scaled}
For any $N \geq 1$, let $\bm{Z}_{1:N}$ be a random sample modeled as the BNP simple multinomial process model \eqref{eq:model_multivariate}, with $\bm{\zeta} \sim 
{\rm mSBD} (\alpha, \kappa+\alpha;\bm{\gamma}; \vartheta)$. Suppose that $\bm{Z}_{1:N}$ displays $K_N=k$ distinct features, labeled by $W_1^*, \ldots  , W_{K_N}^*$, with condiment-specific frequencies $(M_{N,1,j}, \ldots , M_{N,K_N,j})= (m_{1,j}, \ldots , m_{k,j})$, for any $j=1, \ldots , q$. Then, 
the posterior distribution of $U_{N,\ell}^{(M)}$, given $\bm Z_{1:N}$, coincides with the distribution of
\begin{equation}
\label{eq:posterior_condiment_j_NON_scaled}
U_{N,\ell}^{(M)} | \bm Z_{1:N} \sim  {\rm Poisson}  \left( \vartheta  \sum_{m=1}^M (-1)^{m+1} \binom{M}{m}
B (m-\alpha, N+\alpha+\kappa)  \frac{(\gamma_\ell)_m}{(|\bm{\gamma}|)_m}\right)
\end{equation}
\end{theorem}
\begin{proof}
The proof is based on the posterior characterization provided in Theorem \ref{thm:multi_predictive} and the evaluation of the probability generating function of the random variable $\unseen{N,\ell}{M}$, conditionally on the sample $\bm{Z}_{1:N}$. The probability generating function is denoted as usual by 
$\G_{U_{N,\ell}^{(M)} } (\, \cdot \, )$. Thanks to the characterization \eqref{eq:multi_predictive}, conditionally on $\bm{Z}_{1:N}$, the random variable
$U_{N,\ell}^{(M)}$ may be written as
\[
U_{N,\ell}^{(M)}  | \bm{Z}_{1:N} \stackrel{d}{=}\sum_{i \geq 1} \ind \left( \sum_{m=1}^M  A_{m+N,i,\ell}' >0 \right).
\]
Fix $t$ in a neighborhood of the origin, then one has 
\begin{equation}  \label{eq:G_UNjM}
\G_{U_{N,\ell}^{(M)}}  (t)  =  \E \left[  t^{U_{N,\ell}^{(M)}}  \mid \bm{Z}_{1:N}  \right].
\end{equation}
 Here, independently across $i$, $A_{N+m,i,\ell}'$ is a Bernoulli random variable with parameter $\rho_{i,\ell}'$, conditionally  on the random measure
$\bm{\mu}'= \sum_{i \geq 1} \bm{\rho}_i' \delta_{W_i'}$ with L\'evy intensity $\lambda_{(q)}' (\bm{s}) \de s_1 \cdots \de s_q P(\de w)$ such that 
\begin{equation}
\label{eq:SBD_process_posterior}
\lambda_{(q)}' (\bm{s}) = \frac{\vartheta \Gamma (|\bm{\gamma}|)}{\prod_{j=1}^q  \Gamma (\gamma_j)}  |\bm{s}|^{-\alpha-|\bm{\gamma}|}  (1-\bm{s})^{N+\kappa+\alpha-1}\prod_{j=1}^q s_j^{\gamma_j-1}\ind_{[0,1]}(|\bm{s}| )
, \; \bm{s} \in S_q .
\end{equation}
Thus, the expected value in \eqref{eq:G_UNjM} boils down to
\begin{align*}
\G_{U_{N,\ell}^{(M)}}  (t)  & = \E \left[ t^{\sum_{i \geq 1} \ind \left( \sum_{m=1}^M  A_{m+N,i,\ell}' >0 \right)  }  \right] = 
\E \left[  \prod_{i \geq 1}  \E \left[ t^{\ind \left( \sum_{m=1}^M  A_{m+N,i,\ell}' >0 \right)}  \mid \bm{\mu}'\right] \right]\\
& = \E \left[ \prod_{i \geq 1} \left(t+(1-t) \prod_{m=1}^M \P (A_{m+N,i,\ell}' =0| \bm{\mu}')\right) \right]\\
& = \E \left[ \prod_{i \geq 1} (t+(1-t)(1-\rho_{i,\ell}')^M) \right].\\
\end{align*}
where we used the fact that each $A_{m+N,i,\ell}'$ is a Bernoulli random variable with parameter $\rho_{i,\ell}'$, conditionally on the random measure
$\bm{\mu}'$, and in addition these random variables are conditionally independent.  We now exploit the Laplace functional of the multivariate CRM 
$\bm{\mu}'$ to obtain
\begin{align}
\G_{U_{N,\ell}^{(M)}}  (t)  & = \E \left[ \exp \left\{ \sum_{i \geq 1} \log (t+ (1-t)  (1-\rho_{i,\ell}')^M) \right\} \right] \nonumber\\
&= \exp \left\{- (1-t)\int_{S_q}  [1-(1-s_\ell)^M]   \lambda_{(q)}' (\bm{s}) \de s_1 \cdots \de s_q\right\}  \nonumber\\
& = \exp \left\{ (1-t)\sum_{m=1}^M (-1)^m \binom{M}{m}\int_{S_q} s_\ell^{m}   \lambda_{(q)}' (\bm{s}) \de s_1 \cdots \de s_q\right\}  \label{eq:integrals_multi}
\end{align}
where $\lambda_{(q)}' $ has been specified in \eqref{eq:SBD_process_posterior} and we exploited the following formula
\begin{equation}
[1-(1-s_\ell)^M]  = 1- \sum_{m=0}^M (-1)^m \binom{M}{m} s_\ell^m  = -\sum_{m=1}^M (-1)^m \binom{M}{m} s_\ell^m .
\label{eq:formula_bino}
\end{equation}
The integrals over $S_q$ in \eqref{eq:integrals_multi} may be  easily evaluated
(see, e.g., \cite[Formula 4.635.2]{tables}) to get
\[
\int_{S_q} s_\ell^{m}   \lambda_{(q)}' (\bm{s}) \de s_1 \cdots \de s_q =  \vartheta  \frac{(\gamma_\ell)_m}{(|\bm{\gamma}|)_{m}}  \cdot  B (m-\alpha, N+\alpha+\kappa) .
\]
By substituting the previous expression in \eqref{eq:integrals_multi}, we obtain 
\begin{align*}
\G_{U_{N,\ell}^{(M)}}  (t)  & =  \exp \left\{ (t-1)\sum_{m=1}^M (-1)^{m+1} \binom{M}{m} \vartheta  \frac{(\gamma_\ell)_m}{(|\bm{\gamma}|)_{m}}  \cdot  B (m-\alpha, N+\alpha+\kappa)   \right\} 
\end{align*}
which is exactly the probability generating function of a Poisson random variable with parameter
\[
\sum_{m=1}^M (-1)^{m+1} \binom{M}{m} \vartheta  \frac{(\gamma_\ell)_m}{(|\bm{\gamma}|)_{m}}  \cdot  B (m-\alpha, N+\alpha+\kappa) .
\]
\qed
\end{proof}
As  a consequence of Theorem \ref{thm:multi_condiment_prediction_NON-scaled}, one can define a BNP estimator of $U_{N,\ell}^{(M)}$ with respect to a squared  loss function as follows:
\begin{equation}
\label{eq:est_multi_NON-scaled}
\hat{U}_{N,\ell}^{(M)} =  \vartheta  \sum_{m=1}^M (-1)^{m+1} \binom{M}{m}
B (m-\alpha, N+\alpha+\kappa)  \frac{(\gamma_\ell)_m}{(|\bm{\gamma}|)_m}.
\end{equation} 
We point out that for computational convenience one may write
\begin{equation} \label{eq:est_rewritten_multi}
\hat{U}_{N,\ell}^{(M)} =  \vartheta B(1-\alpha, N+\alpha+\kappa) \E_{(X,Y)} \left[ \frac{1-(1-XY)^M}{Y} \right]
\end{equation}
where the expected value is taken with respect to the two independent random variables with the following beta distributions
\[
X \sim {\rm Beta} (\gamma_\ell, |\bm{\gamma}|-\gamma_\ell), \quad Y \sim {\rm Beta} (1-\alpha, N+\alpha+\kappa).
\]
The equality \eqref{eq:est_rewritten_multi} may be easily proved by observing that
\[
\E_X [X^m]=  \frac{(\gamma_\ell)_m}{(|\bm{\gamma}|)_m} \quad \text{and} \quad B(m-\alpha, N+\alpha+\kappa) =  \E_Y [Y^{m-1}] B (1-\alpha, N+\alpha+\kappa).
\]

\subsection{Scaled stable-Beta-Dirichlet prior for multinomial processes}
\label{sec:multi_scaled}

From Theorems \ref{thm:multi_predictive}-\ref{thm:multi_condiment_prediction_NON-scaled}, it is apparent that, under the stable-Beta-Dirichlet process, the conditional distribution of a statistic involving hitherto unobserved features, depends on the initial sample $\bm{Z}_{1:N}$ only trough the sample size $N$ and not on other sample statistics. This behavior resembles what happens for the Bernoulli process model described in the main paper when the prior $\zeta$ in \eqref{exch_mod} is a CRM. We then introduce a multivariate analogue of the stable-Beta scaled prior, that will be termed \textit{scaled stable-Beta-Dirichlet process} with the goal to enrich the predictive structure.  We introduce a discrete random measure depending on the random jump $\Delta_{1,h_{c, \beta}}$, that has been defined in the main paper as a  polynomial-exponential tilting of the density function \eqref{eq:stable_largest_jump}, whose density equals
\begin{equation}
\label{eq:density_Dhcbeta}
f_{\Delta_{1, h_{c, \beta}}} (a) =  \frac{\sigma \beta^{c+1}}{\Gamma(c+1)} a^{-\sigma(c+1)-1}\exp\left\{-\beta a^{-\sigma}\right\}\ind_{\R_+} (a)
\end{equation}
as shown in \eqref{eq:stable_mixing}. The scaled stable-Beta-Dirichlet random measure is an almost surely discrete random measure that can be represented as
\[
\bm\mu_{\Delta_{1, h_{c, \beta}}}  = \sum_{i \geq 1} \bm{\rho}_i  \delta_{W_i}, \quad \bm{\rho}_i = (\rho_{i,1},\ldots, \rho_{i,q})
\]
and consisting of $q$ components
\[
\mu_{\Delta_{1, h_{c, \beta}}, j} = \sum_{i \geq 1}\rho_{i,j} \delta_{W_i} \quad \text{as } j= 1, \ldots , q.
\]
Conditionally on the jump $\Delta_{1, h_{c, \beta}}$, the multivariate random measure $\bm\mu_{\Delta_{1, h_{c, \beta}}}$ is completely
random with L\'evy intensity $\lambda_{(q),\Delta_{1, h_{c, \beta}}} (\bm{s}) \de s_1 \cdots \de s_q P (\de p)$ with the specification
\begin{equation}
\label{eq:scaled_SBP}
\lambda_{(q),\Delta_{1, h_{c, \beta}}} (\bm{s}) = \frac{ \Gamma (|\bm{\gamma}|)}{\prod_{j=1}^q  \Gamma (\gamma_j)} 
\sigma \Delta_{1, h_{c, \beta}}^{-\sigma} 
 |\bm{s}|^{-\sigma-|\bm{\gamma}|} \prod_{j=1}^q s_j^{\gamma_j-1}\ind_{[0,1]}(|\bm{s}| ) 
, \; \bm{s}  \in S_q
\end{equation}
where  $0 <\sigma < 1$ and $\gamma_j >0$ for any $j=1, \ldots , q$. 
We write $\bm{\mu}_{\Delta_{1,h}} \sim \text{\rm S-mSBD} (\sigma, \bm{\gamma}; h_{c, \beta})$. A remarkable property of this model is that
$\sum_{j=1}^q \mu_{\Delta_{1, h_{c, \beta}}, j}$ is distributed as the stable-Beta scaled process prior, i.e., $|\bm{\mu}_{\Delta_{1, h_{c, \beta}}}|
\sim {\rm \text{SB-SP}} (\sigma,c, \beta)$. Such a property may be easily proved by means of the Laplace functionals. Note that one could potentially  introduce an additional mass parameter in the model, but this is irrelevant to carry out posterior inference in the stable case.

\subsubsection{Posterior Analysis}  \label{sec:multi_posterior_analysis}

We now provide posterior, predictive and marginal characterizations for the multivariate model \eqref{eq:model_multivariate} under the scaled stable-Beta-Dirichlet process prior specification for $\mathscr{Z}$. The results we present here may be proved by exploiting \cite[Section 5]{james2017bayesian}, conditionally on $\Delta_{1, h_{c, \beta}}$ and then by marginalizing over the mixing distribution \eqref{eq:density_Dhcbeta}.  We omit the details.

\begin{theorem} \label{thm:multi_posterior_scaled}
For any $N \geq 1$, let $\bm{Z}_{1:N}$ be a random sample modeled as the BNP simple multinomial process model \eqref{eq:model_multivariate}, with $\bm{\zeta} \sim 
\text{\rm S-mSBD} (\sigma, \bm{\gamma}; h_{c, \beta}) $. If $\bm{Z}_{1:N}$ displays $K_N=k$ distinct features, labeled by $W_1^*, \ldots  , W_{K_N}^*$, with condiment-specific frequencies $(M_{N,1,j}, \ldots , M_{N,K_N,j})= (m_{1,j}, \ldots , m_{k,j})$, for any $j=1, \ldots , q$, then the
 conditional distribution of $\Delta_{1,h_{c, \beta}}$ given $\bm{Z}_{1:N}$, coincides with the distribution of
\begin{equation}
\label{eq:multi_jump}
\Delta_{1, h_{c, \beta}}^{-\sigma} \sim {\rm Gamma} (K_N+c+1, \beta+ \gamma_0^{(N)})
\end{equation}
where $\gamma_0^{(n)} = \sigma\sum_{1 \leq i \leq n} B(1-\sigma, i)$.  Moreover, the conditional distribution of $\bm{\zeta}$, given $\bm{Z}_{1:N}, \Delta_{1, h_{c, \beta}}$, coincides with the distribution of 
\begin{equation}
\label{eq_multi_posterior_scaled}
\bm{\zeta} | (\bm{Z}_{1:N}, \Delta_{1, h_{c, \beta}}) \stackrel{d}{=} \bm{\mu}_{\Delta_{1, h_{c, \beta}}}' +\sum_{i=1}^{K_N} \bm{J}_{i} \delta_{W_i^*}
\end{equation}
where:
\begin{itemize}
\item[i)] $\bm{\mu}_{\Delta_{1, h_{c, \beta}}}'$ is a discrete multivariate random measure with L\'evy intensity 
\begin{equation}
\label{eq:levy_multi}
\begin{split}
&\nu_{\Delta_{1, h_{c, \beta}}}' (\de s_1, \ldots , \de s_q , \de w)= \frac{\Gamma (|\bm{\gamma}|)}{\prod_{j=1}^q  \Gamma (\gamma_j)} \\
& \qquad \times
  |\bm{s}|^{-\sigma-|\bm{\gamma}|}  (1-\bm{s})^{N}\prod_{j=1}^q s_j^{\gamma_j-1}\ind_{[0,1]}(|\bm{s}| ) \sigma  \Delta_{1, h_{c, \beta}}^{-\sigma} \de s_1 \cdots \de s_q \, P(\de w);
\end{split}
\end{equation} 
\item[ii)]  $\bm{J}_{1:K_N}$ is a vector of independent random jumps such that each $\bm{J}_{ i} = 
(J_{ 1}, \ldots , J_{ q})$ has a Beta-Dirichlet distribution, i.e., 
\begin{equation}
\label{eq:multi_discrete_J_scaled}
\bm{J}_{i} | \Delta_{1, h_{c, \beta}}  \sim \mathscr{B}\mathscr{D}(m_i-\sigma, N-m_i+1 ; \bm{\gamma}+ \bm{m}_i)
\end{equation}
 where we put $\bm{m}_i := (m_{i,1}, \ldots , m_{i,q})$ and $m_i = \sum_{j=1}^q m_{i,j}= |\bm{m}_i|$ for any $i=1, \ldots , K_N$.
\end{itemize}
\end{theorem}

\begin{theorem} \label{thm:multi_predictive_scaled}
For any $N \geq 1$, let $\bm{Z}_{1:N}$ be a random sample modeled as the BNP simple multinomial process model \eqref{eq:model_multivariate}, with $\bm{\zeta} \sim 
\text{\rm S-mSBD} (\sigma, \bm{\gamma}; h_{c, \beta}) $. If $\bm{Z}_{1:N}$ displays $K_N=k$ distinct features, labeled by $W_1^*, \ldots  , W_{K_N}^*$, with condiment-specific frequencies $(M_{N,1,j}, \ldots , M_{N,K_N,j})= (m_{1,j}, \ldots , m_{k,j})$, for any $j=1, \ldots , q$, then the
 conditional distribution of $\Delta_{1,h_{c, \beta}}$ given $\bm{Z}_{1:N}$, coincides with \eqref{eq:multi_jump}. Moreover, the 
 conditional distribution of $\bm{Z}_{N+1}$, given $\bm{Z}_{1:N}, \Delta_{1, h_{c, \beta}}$, coincides with the distribution of 
\begin{equation}
\label{eq_multi_pred_scaled}
\bm{Z}_{N+1} | (\bm{Z}_{1:N}, \Delta_{1, h_{c, \beta}}) \stackrel{d}{=} \bm{Z}_{N+1}' +\sum_{i=1}^{K_N} \bm{A}_{N+1,i} \delta_{W_i^*}
\end{equation}
where:
\begin{itemize}
\item[i)] $\bm{Z}_{N+1}'$ is such that $\bm{Z}_{N+1}' | \Delta_{1, h_{c, \beta}}= \sum_{i \geq 1}\bm{A}_{N+1,i}' \delta_{W_i'} \sim \mathsf{MP} 
(\bm{\mu}_{\Delta_{1, h_{c, \beta}}}')$  and $\bm{\mu}_{\Delta_{1, h_{c, \beta}}}'| \Delta_{1, h_{c, \beta}}$ is the completely random measure having the L\'evy intensity \eqref{eq:levy_multi};
\item[ii)] $\bm{A}_{N+1, 1:K_N}$ is a collection of independent simple multinomial random variables with parameters
$\bm{J}_{1:K_N}$, each one distributed according to Equation \eqref{eq:multi_discrete_J_scaled}.
\end{itemize}
\end{theorem}

\begin{theorem}
\label{thm:multi_marginal_scaled}
For any $N \geq 1$, let $\bm{Z}_{1:N}$ be a random sample modeled as the BNP simple multinomial process model \eqref{eq:model_multivariate}, with $\bm{\zeta} \sim 
\text{\rm S-mSBD} (\sigma, \bm{\gamma}; h_{c, \beta}) $. The probability that  $\bm{Z}_{1:N}$ displays a particular feature allocation of $K_N=k$ distinct features with condiment-specific frequencies $(M_{N,1,j}, \ldots , M_{N,K_N,j})= (m_{1,j}, \ldots , m_{k,j})$, for any $j=1, \ldots , q$, equals
\begin{equation}
\label{eq:multi_marginal_scsled}
\begin{split}
p_k^{(N)} (\bm{m}_1, \ldots , \bm{m}_k)  &= \prod_{i=1}^k \left\{ B (m_i-\sigma, N-m_i+1)  \frac{\prod_{j=1}^q (\gamma_j)_{m_{i,j}}}{(|\bm{\gamma}|)_{m_i}}\right\}\\
& \qquad\qquad \times \frac{\Gamma (k+c+1)}{\Gamma (c+1)} \cdot \frac{\sigma^k \beta^{c+1}}{(\beta+\gamma_{0}^{(N)})^{k+c+1}} .
\end{split}
\end{equation}
\end{theorem}

\subsubsection{Estimation of the unseen features with a condiment}
\label{appendix_multi_estimation_scaled}

For any $N \geq 1$, let $\bm{Z}_{1:N}$ be an observable sample modeled as the simple multinomial model in \eqref{eq:model_multivariate}, with 
$\bm{\zeta} \sim \text{\rm S-mSBD} (\sigma, \bm{\gamma}; h_{c, \beta})$. Moreover, under the same model, for $M \geq 1$ let $\bm{Z}_{N+1:N+M}= (\bm{Z}_{N+1}, \ldots , \bm{Z}_{N+M})$ be an additional and unobserved sample. Under this model, we now determine the posterior distribution of the sample statistic  $U_{N,\ell}^{(M)}$ in 
\eqref{eq:new_condiment_def}, counting the number of hitherto unobserved feature with condiment $\ell \in \{ 1, \ldots , q\}$  that will be recorded out of the additional sample.
\begin{theorem}  \label{thm:multi_condiment_prediction-scaled}
For any $N \geq 1$, let $\bm{Z}_{1:N}$ be a random sample modeled as the BNP simple multinomial process model \eqref{eq:model_multivariate}, with $\bm{\zeta} \sim \text{\rm S-mSBD} (\sigma, \bm{\gamma}; h_{c, \beta})$. Suppose that $\bm{Z}_{1:N}$ displays $K_N=k$ distinct features with condiment-specific frequencies $(M_{N,1,j}, \ldots , M_{N,K_N,j})= (m_{1,j}, \ldots , m_{k,j})$, for any $j=1, \ldots , q$. Then, 
the posterior distribution of $U_{N,\ell}^{(M)}$, given $\bm Z_{1:N}$, coincides with the distribution of
\begin{equation}
\label{eq:posterior_condiment_j_scaled}
U_{N,\ell}^{(M)} | \bm Z_{1:N} \sim  {\rm NegativeBinomial}  \left( K_N+c+1, \frac{\psi_{N,\ell}^{(M)}}{\psi_{N,\ell}^{(M)} + \gamma_0^{(N)} +\beta}\right)
\end{equation}
where we defined
\[
\psi_{N,\ell}^{(M)}  := \sigma  \sum_{m=1}^M \binom{M}{m} (-1)^{m+1} \frac{(\gamma_\ell)_{m}}{(|\bm{\gamma}|)_{m}}  B (m-\sigma, N+1).
\]
\end{theorem}
\begin{proof}
The proof is based on the posterior characterization in Theorem \ref{thm:multi_posterior_scaled} and on Theorem \ref{thm:multi_predictive_scaled}.
As in the proof of Theorem \ref{thm:multi_condiment_prediction_NON-scaled} we evaluate the probability generating function of the random variable $\unseen{N,\ell}{M}$, conditionally on the sample $\bm{Z}_{1:N}$. The probability generating function is denoted as usual by 
$\G_{U_{N,j}^{(M)} } (\, \cdot \, )$. Thanks to the characterization \eqref{eq_multi_pred_scaled}, conditionally on $\bm{Z}_{1:N}, \Delta_{1, h_{c, \beta}}$, the random variable
$U_{N,\ell}^{(M)}$ may be written as
\[
U_{N,\ell}^{(M)}  | (\bm{Z}_{1:N}, \Delta_{1, h_{c, \beta}}) \stackrel{d}{=}\sum_{i \geq 1} \ind \left( \sum_{m=1}^M  A_{m+N,i,\ell}' >0 \right).
\]
Fix $t$ in a neighborhood of the origin, then one has 
\begin{equation}  \label{eq:G_UNjM_scaled}
\G_{U_{N,\ell}^{(M)}}  (t)  =  \E \left[  t^{U_{N,\ell}^{(M)}}  \mid \bm{Z}_{1:N} \right]  
= \E \left[ \E  \left[  t^{U_{N,\ell}^{(M)}}  \mid \bm{Z}_{1:N}  , \Delta_{1, h_{c, \beta}} \right] \mid  \bm{Z}_{1:N}  \right]  
\end{equation}
by an application of the tower property. We  now focus on the evaluation of the inner expected value in \eqref{eq:G_UNjM_scaled}:
\begin{align*}
& \E  \left[  t^{U_{N,\ell}^{(M)}}  \mid \bm{Z}_{1:N}  , \Delta_{1, h_{c, \beta}} \right]  = \E \left[   t^{ \sum_{m=1}^M  A_{m+N,i,\ell}' }  \right]\\
 & \qquad = \E \left[ \prod_{i \geq 1} \E [1 \cdot \P (\sum_{m=1}^M A_{m+N,i,\ell}' =0)+ t \cdot \P (\sum_{m=1}^M  A_{m+N,i,\ell}'  >0)] \right].
\end{align*}
From Theorem \ref{thm:multi_predictive_scaled}, the $ A_{m+N,i,\ell}'$s are independent random variables as $m=1, \ldots ,M$, and  each one $A_{N+m,i,\ell}'$ is a Bernoulli with parameter $\rho_{i,\ell}'$, conditionally  on the random measure
$\bm{\mu}_{\Delta_{1, h_{c, \beta}}}'= \sum_{i \geq 1} \bm{\rho}_i' \delta_{W_i'}$ with L\'evy intensity \eqref{eq:levy_multi}.
As a consequence we obtain
\begin{align*}
 \E  \left[  t^{U_{N,\ell}^{(M)}}  \mid \bm{Z}_{1:N}  , \Delta_{1, h_{c, \beta}} \right] & = \E \left[ \prod_{i \geq 1} \left[ 
t + (1-t) (1-\rho_{i,\ell}')^M \right]\right].
\end{align*}
Proceeding along the same lines as in the proof of Theorem \ref{thm:multi_condiment_prediction_NON-scaled} we have that
\begin{align*}
 &\E  \left[  t^{U_{N,\ell}^{(M)}}  \mid \bm{Z}_{1:N}  , \Delta_{1, h_{c, \beta}} \right]  =  \E \left[ \exp \left\{ \sum_{i \geq 1}
 \log  (t+(1-t)(1-\rho_{i,\ell})^M) \right\} \right]\\
 & \qquad =  \exp \left\{ - (1-t)\int_{\mathds{W}} \int_{S_q}   [1- (1-s_\ell)^M] \nu_{\Delta_{1, h_{c, \beta}}}' (\de s_1, \ldots , \de s_q , \de w) \right\} .
\end{align*}
Now define 
\[
\lambda_{(q),\Delta_{1, h_{c, \beta}}}' (\bm{s}):= \frac{\Gamma (|\bm{\gamma}|)}{\prod_{j=1}^q  \Gamma (\gamma_j)} 
  |\bm{s}|^{-\sigma-|\bm{\gamma}|}  (1-\bm{s})^{N}\prod_{j=1}^q s_j^{\gamma_j-1}\ind_{[0,1]}(|\bm{s}| ) \sigma  \Delta_{1, h_{c, \beta}}^{-\sigma} 
\]
thus, the conditional expected value under study may be written as
\begin{align}
 &\E  \left[  t^{U_{N,\ell}^{(M)}}  \mid \bm{Z}_{1:N}  , \Delta_{1, h_{c, \beta}} \right] 
  =  \exp \left\{ - (1-t) \int_{S_q}   [1- (1-s_\ell )^M] \lambda_{(q),\Delta_{1, h_{c, \beta}}}' (\bm{s}) \de s_1, \ldots , \de s_q  \right\}  \nonumber \\
  & \qquad = \exp \left\{ - (1-t)   \sum_{m=1}^M   \binom{M}{m} (-1)^{m+1}\int_{S_q} s_\ell^m \lambda_{(q),\Delta_{1, h_{c, \beta}}}' (\bm{s}) \de s_1, \ldots , \de s_q \right\} \label{eq:integral_scaled}
\end{align} 
where we applied \eqref{eq:formula_bino}. The integral over $S_q$ appearing in \eqref{eq:integral_scaled} may be evaluated resorting to 
\cite[Formula 4.635.2]{tables}, therefore
\[
\begin{split}
&\int_{S_q} s_\ell^m \lambda_{(q),\Delta_{1, h_{c, \beta}}}'  (\bm{s}) \de s_1, \ldots , \de s_q  \\
& \qquad \qquad= 
\sigma \Delta_{1, h_{c, \beta}}^{-\sigma} \frac{\Gamma (|\bm{\gamma}|)}{\prod_{j=1}^q \Gamma (\gamma_j)}\int_{S_q}  s_\ell^m (1-|\bm{s}|)^{N}
|\bm{s}|^{-\sigma-|\bm{\gamma}|}\prod_{j =1}^q s_j^{\gamma_j-1}   \de s_1 \cdots \de s_q\\
&\qquad\qquad=  \sigma\Delta_{1, h_{c, \beta}}^{-\sigma}   \frac{(\gamma_\ell)_m}{(|\bm{\gamma}|)_m}  B (m-\sigma, N+1).
\end{split}
\]
Thus, by substituting the previous expression in \eqref{eq:integral_scaled} one obtains
\begin{equation}
\label{eq:multi_cond_scaled_U}
\E  \left[  t^{U_{N,\ell}^{(M)}}  \mid \bm{Z}_{1:N}  , \Delta_{1, h_{c, \beta}} \right] 
  = \exp \left\{ - (1-t)\Delta_{1, h_{c, \beta}}^{-\sigma} \psi_{N, \ell}^{(M)}   \right\}
\end{equation}
where we recall that $\psi_{N, \ell}^{(M)}$ has been defined as follows
\[
\psi_{N, \ell}^{(M)} = \sigma\sum_{m=1}^M   \binom{M}{m} (-1)^{m+1}    \frac{(\gamma_\ell)_m}{(|\bm{\gamma}|)_m}  B (m-\sigma, N+1).
\]
As a consequence,  the probability generating function in \eqref{eq:G_UNjM_scaled} equals
\begin{equation*}
\G_{U_{N,\ell}^{(M)}}  (t)  \stackrel{\eqref{eq:multi_cond_scaled_U}}{=}  \E \left[ \exp \left\{ - (1-t)\Delta_{1, h_{c, \beta}}^{-\sigma} \psi_{N, \ell}^{(M)}   \right\}\mid \bm{Z}_{1:N}\right].
\end{equation*}
The conclusion follows by a marginalization w.r.t. the posterior distribution of $\Delta_{1, h_{c, \beta}}^{-\sigma}$ which is a gamma random variable (see \eqref{eq:multi_jump}):
\begin{align*}
\G_{U_{N,\ell}^{(M)}}  (t)  & = \int_0^\infty e^{-(1-t) \psi_{N, \ell}^{(M)}x}  \cdot \frac{(\beta+\gamma_0^{(N)})^{K_N+c+1}}{\Gamma (K_N+c+1)}
x^{K_N+c} e^{-x (\gamma_0^{(N)}+\beta)} \de x\\
 & = \frac{(\beta+\gamma_0^{(N)})^{K_N+c+1}}{\Gamma (K_N+c+1)} \int_0^\infty e^{-[ (\gamma_0^{(N)}+\beta)+(1-t) \psi_{N, \ell}^{(M)}]x} 
x^{K_N+c+1-1} \de x\\
 & = \frac{(\beta+\gamma_0^{(N)})^{K_N+c+1}}{[ (\gamma_0^{(N)}+\beta)+(1-t) \psi_{N, \ell}^{(M)}]^{K_N+c+1}}\\
 & = \left( \frac{\beta+\gamma_0^{(N)}}{ \gamma_0^{(N)}+\beta+\psi_{N, \ell}^{(M)}-t \psi_{N, \ell}^{(M)}} \right)^{K_N+c+1}
\end{align*}
which is the probability generating function of a negative binomial distribution as in the statement.\\
\qed
\end{proof}
As a consequence of Theorem \ref{thm:multi_condiment_prediction-scaled}, the BNP estimator of $U_{N, \ell}^{(M)}$ under a squared loss function equals
\begin{equation}
\label{eq:mutil_NB_estimator}
\hat{U}_{N,\ell}^{(M)} =   (K_N+c+1)\frac{\psi_{N, \ell}^{(M)}}{\gamma_0^{(N)}+\beta} .
\end{equation}
For computational purposes, we finally note that the parameter $\psi_{N, \ell}^{(M)}$ in the posterior representations may be computed as
\[
\psi_{N, \ell}^{(M)}=  B (1-\sigma, N+1)  \E_{(X,Y)}\left[ \frac{1-(1-XY)^M}{Y} \right]
\]
where the expected value is made w.r.t. two independent random variables $X$ and $Y$ having beta distributions as follows
\[
X \sim {\rm Beta } (\gamma_\ell, |\bm{\gamma}|-\gamma_\ell) \quad  \text{and} \quad  Y \sim {\rm Beta} (1-\sigma, N+1).
\]

\section{Synthetic experiments from the model} \label{sec:app_synthetic_model}
We now analyze empirically the properties of the SB-SP-Bernoulli model used in Section \ref{sec:theory1}. We will use the acronym SSB for brevity in the captions. I.e., we consider the hierarchical model detailed in \eqref{exch_mod}, with $\mu\sim\text{\rm SB-SP}(\sigma,c,\beta)$. The predictive characterization detailed in Proposition \ref{prop:stable_predictive}, together with Equation \eqref{eq:stable_news}, provides an algorithm to sample $N$ observations from the model: given $\beta>0, \sigma \in (0,1), c>0$,
\begin{itemize}
    \item at every step $n = 1, \ldots, N$, conditionally on the previous $n-1$ samples $Z_{1:n-1}$ showing $K_{n-1}$ distinct features, each feature $k = 1,\ldots,m_{K_{n-1}}$ with frequency $m_k$, sample 
    \begin{itemize}
        \item a random number of new features observed:
            \[
                U_{n-1}^{(1)} \mid Z_{1:n-1} \sim \NegBin \left(K_{n-1}+c+1, \frac{\stablenews{n-1}{1}}{\beta + \stablenews{0}{n}}\right);
            \]
        \item for previously observed feature $i = 1,\ldots,K_{n-1}$:
            \[
                A_{n,i} \mid Z_{1:n-1} \sim \Bern\left(m_i - \sigma, N - m_i + 1\right);
            \]
    \end{itemize}
\end{itemize}
In particular, for the first step, $K_0 = 0$.

\subsection{Predictive behavior of the number of new features from the prior}

First, we investigate the predictive behavior of the model as we vary the hyperparameters of the process --- $\beta, \sigma, c$. Because our interest is in understanding the coverage properties of the posterior predictive distribution induced by the model, we report, together with the predictive mean, also posterior predictive credible intervals. In this first set of simulations reported in \Cref{sec:app_synthetic_sigma,sec:app_synthetic_c,sec:app_synthetic_beta}, we assume the hyperparameter $\sigma, c, \beta$ to be known.

\subsubsection{The role of $\sigma$} \label{sec:app_synthetic_sigma}

We start by analyzing the role of $\sigma$ in \Cref{fig:app_synth_sigma}. As suggested by the asymptotic behavior analyzed in Theorem \ref{thm:stable_conv}, $\sigma$ directly controls the asymptotic rate of growth of the number of distinct features: as $\sigma$ increases, the expected number of variants increases, approaching a linear behavior as $\sigma \to 1$. We notice that this behavior is reminiscent of the tail parameter of the stable beta-Bernoulli process \citep{teh2009indian,broderick2012beta}.

\begin{figure}
    \centering
      \centering \includegraphics[width=\textwidth,height=\textheight,keepaspectratio]{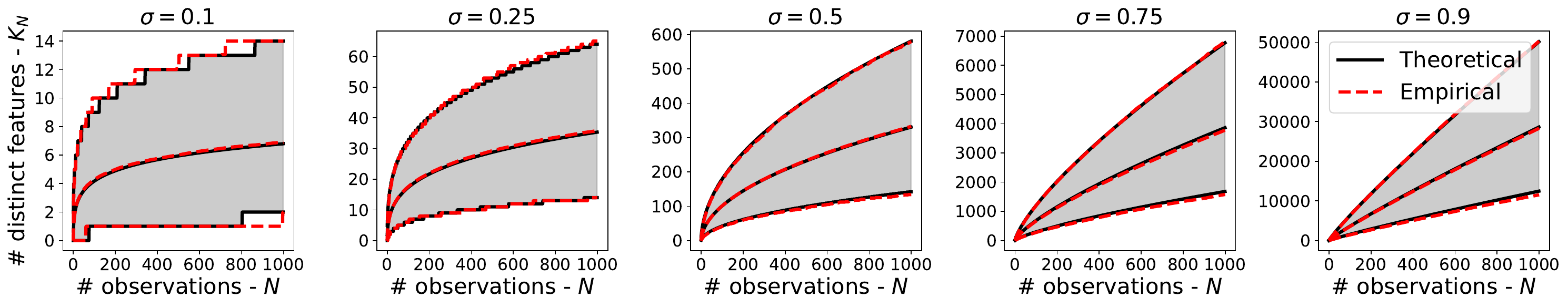}
    \caption{\footnotesize{$90\%$ centered credible interval for the number of distinct features $K_N$ ($y$-axis) as a function of the sample size $N$ ($x$-axis). We fix $\beta = 1$, $c = 5$, and vary $\sigma$ across subplots. For the $5\%, 50\%, 95\%$ quantiles, we compare the theoretical value (solid black lines) to empirical result  (dashed red lines), obtained by drawing $N_{MC} = 1000$ different datasets with the same parameter specification.}}
    \label{fig:app_synth_sigma}
\end{figure}

\subsubsection{The role of c} \label{sec:app_synthetic_c}

We now move to the analysis of the polynomial tilting parameter, $c$. As suggested by the predictive distribution given in \Cref{eq:stable_news}, $c$ acts as a ``prior'' number of features. That is, in the prior, the expected number of features to be observed from $N$ samples is a Negative Binomial random variable with parameters $c+1, \stablenews{0}{N}/(\beta+\stablenews{0}{N})$, i.e.\ with expectation given by 
\[
    \E[U_0^N] = (c+1) \left( \frac{\stablenews{0}{N}}{\beta} \right).
\]
Again, larger values of $c$ induce a higher rate of growth in the number of features, as showed in Figure \ref{fig:app_synth_tilting}.
\begin{figure}
    \centering
      \centering \includegraphics[width=\textwidth,height=\textheight,keepaspectratio]{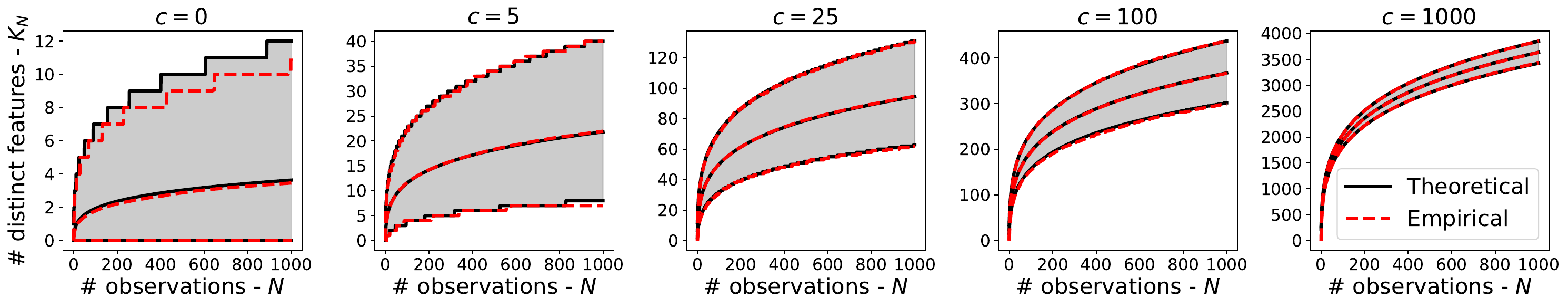}
    \caption{\footnotesize{$90\%$ centered credible interval for the number of distinct features $K_N$ ($y$-axis) as a function of the sample size $N$ ($x$-axis). We repeat the same experiments as in \Cref{fig:app_synth_sigma}, but now fix $\beta = 1$, $\sigma = 0.2$, and vary $c$ across subplots.}}
    \label{fig:app_synth_tilting}
\end{figure}

\subsubsection{The role of $\beta$} \label{sec:app_synthetic_beta}

Last, we analyze the role of the exponential tilting parameter, $\beta$. Inspecting again the predictive distribution \Cref{eq:stable_news}, $\beta$ affects the number of new variants thorugh the success probability of the negative binomial --- for fixed $c,\sigma,N,M,Z_{1:N}$, the expected number of new variants $U_N^{(M)}\mid Z_{1:N}$ depends inversely on the parameter $\beta$. We verify this empirically in \Cref{fig:app_synth_beta}.
\begin{figure}
      \centering \includegraphics[width=\textwidth,height=\textheight,keepaspectratio]{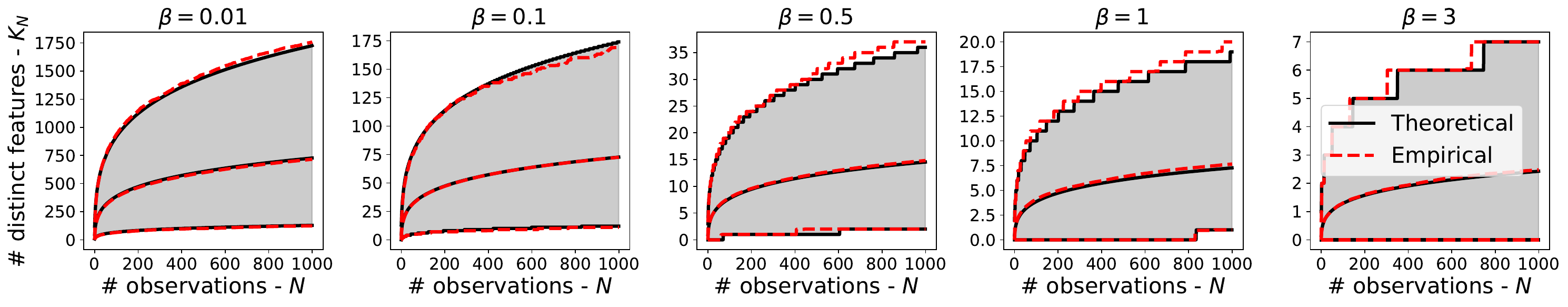}
    \caption{\footnotesize{$90\%$ centered credible interval for the number of distinct features $K_N$ ($y$-axis) as a function of the sample size $N$ ($x$-axis). We repeat the same exercise as in Figures \ref{fig:app_synth_sigma}-\ref{fig:app_synth_tilting} but now $c = 1$, $\sigma = 0.2$, and vary $\beta$ across subplots.}}
    \label{fig:app_synth_beta}
\end{figure}
\subsection{Predictive behavior of the number of new features from the posterior}

Next, we perform a slightly different exercise from the one described above. We still assume the parameters to be known, and we investigate how the posterior predictive behavior varies as we change the number of training samples $N$ with respect to a total sampling ``capacity'' $L$. Intuitively, for a fixed value of this ``sampling capacity'', $N+M=L$, the expected number of observed features from the model should be the independent of the choice of $N, M$. However, we expect the distribution (e.g., the posterior variance), to concentrate as $N$ increases relative to $M$. To perform this experiment, we do as follow: we fix $\beta, c, \sigma$ and, for each $\ell = 1,\ldots, 2000$, we let $K_\ell = U_0^{(\ell)}$. Next, for $N \in \{50,100,500,1000\}$, we compute $U_N^{(M)} \mid Z_{1:N}$, where we condition on the number of observed variants as given by the curve $\{K_\ell\}_{\ell = 1,\ldots,2000}$. As displayed in  \Cref{fig:app_synth_width} and \Cref{fig:app_synth_width_2}, the width of the credible intervals shrinks with increasing training sizes $N$.
\begin{figure}
      \centering \includegraphics[width=\textwidth,height=\textheight,keepaspectratio]{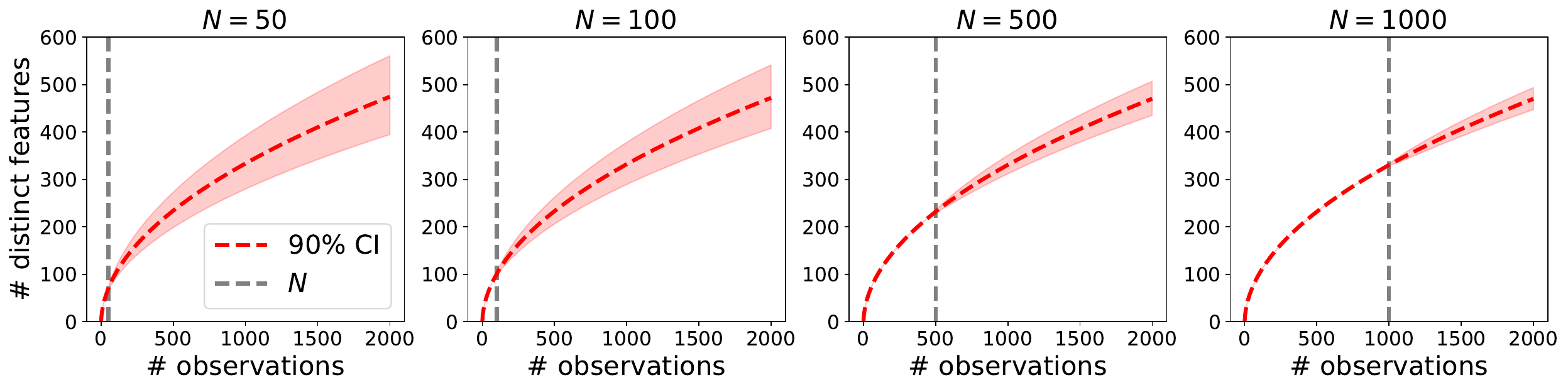}
    \caption{\footnotesize{$90\%$ centered credible interval for the expected number of distinct features $\E[U_N^{(M)} ~\mid~ Z_{1:N}]$ ($y$-axis) as a function of the sample size $N$ ($x$-axis). We fix $\beta = 1$, $c = 5$, $\sigma = 0.5$, and total sequencing capacity $L = 2000$. In different subplots, we show $\E[U_N^{(M)} \mid Z_{1:N}]$ for different values of $N$. Here, the first $N$ samples display exactly $K_N = \E[U_0^{(N)}]$ distinct features.}}
    \label{fig:app_synth_width}
\end{figure}

\begin{figure}
      \centering \includegraphics[width=\textwidth,height=\textheight,keepaspectratio]{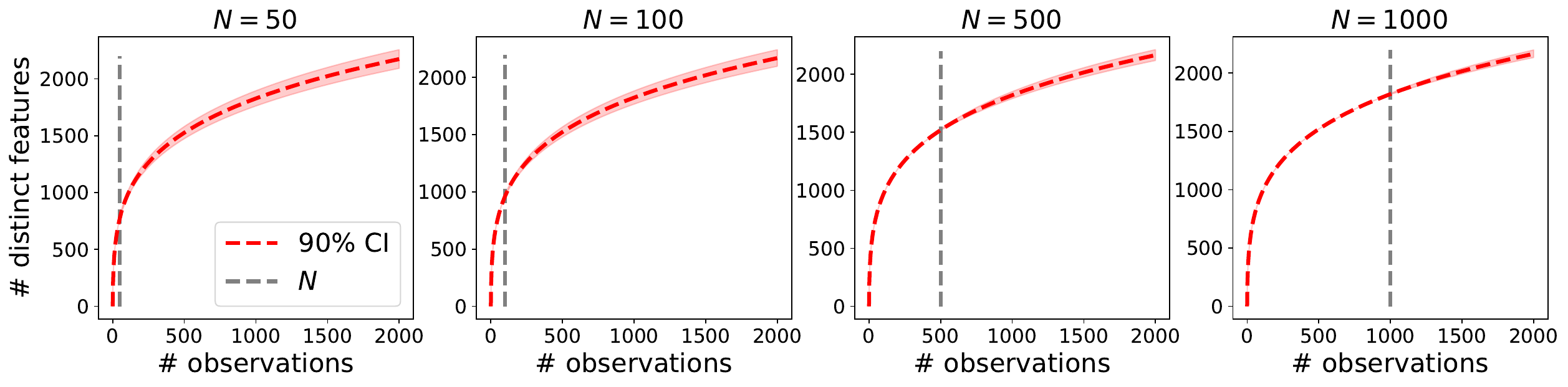}
    \caption{\footnotesize{$90\%$ centered credible interval for the expected number of distinct features $\E[U_N^{(M)} ~\mid~ Z_{1:N}]$ ($y$-axis) as a function of the sample size $N$ ($x$-axis). We repeat the same exercise as in \Cref{fig:app_synth_width} but now fix $\beta = 2$, $c = 1000$, $\sigma = 0.2$.}}
    \label{fig:app_synth_width_2}
\end{figure}

\subsection{Estimation of the parameters}

Next, we move to the more interesting scenario in which the parameters are unknown and need to be inferred from the data. The natural way to estimate the unknown parameters is to maximize a likelihood criterion, such as the marginal distribution of the feature counts $m_1,\ldots,m_K$, given in Equation \eqref{eq:stable_marginal}. We found this method to work well both on real data, as displayed in Section \ref{sec:exp}, and on synthetic data. We here report some results in \Cref{fig:app_synth_learned_1,fig:app_synth_learned_2}. In general, and not surprisingly, the precision of our estimates increases with larger sample sizes.
\begin{figure}
      \centering \includegraphics[width=\textwidth,height=\textheight,keepaspectratio]{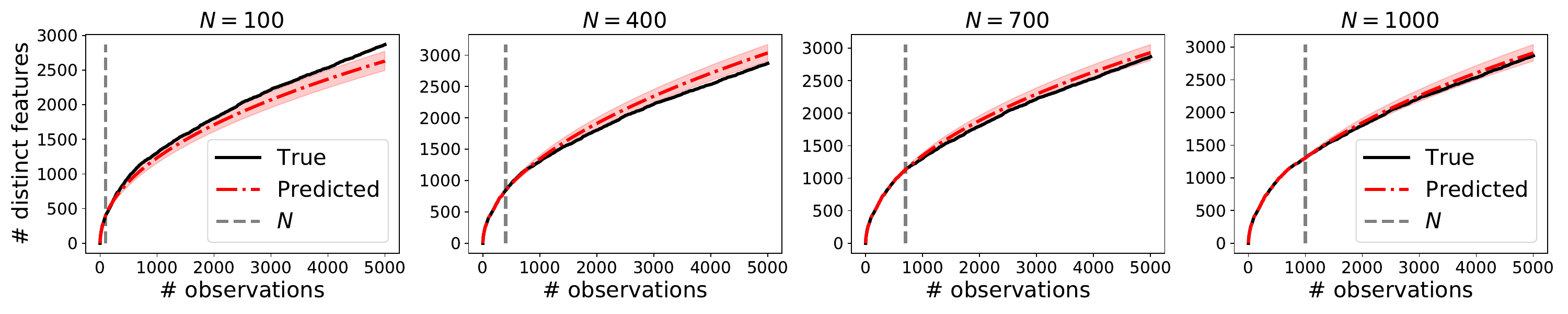}
    \caption{\footnotesize{$99\%$ credible interval centered around the posterior predictive mean (dashed red line) of the number of distinct features $U_N^{(M)} \mid Z_{1:N}$ ($y$-axis) as a function of the sample size $N$ ($x$-axis). We fix $\beta = 1$, $c = 20$, $\sigma = 0.5$, and learn the parameters for different training size $N \in \{100,400,700,1000\}$ across subplots for a total sequencing capacity $L = 5000$.}}
    \label{fig:app_synth_learned_1}
\end{figure}

\begin{figure}
      \centering \includegraphics[width=\textwidth,height=\textheight,keepaspectratio]{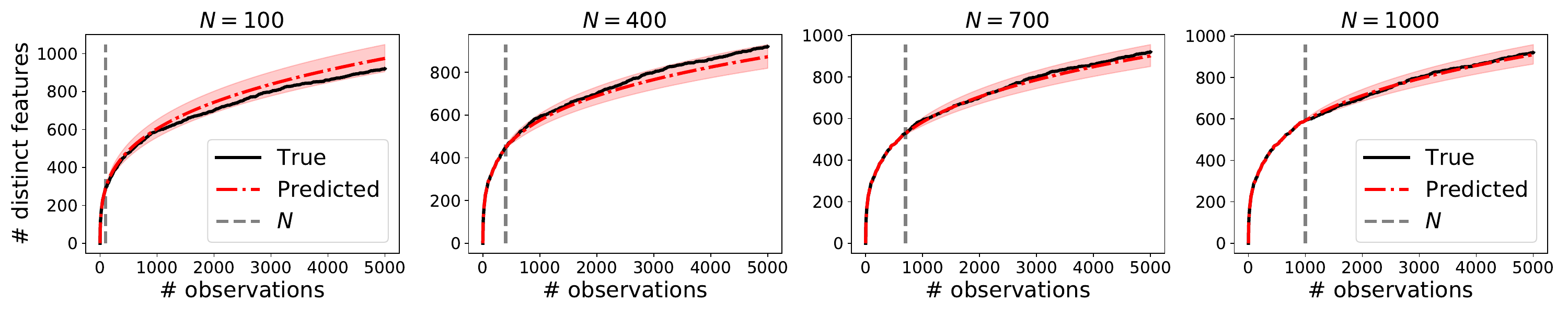}
    \caption{\footnotesize{$99\%$ credible interval centered around the posterior predictive mean (dashed red line) of the number of distinct features $U_N^{(M)} \mid Z_{1:N}$ ($y$-axis) as a function of the sample size $N$ ($x$-axis). We repeat the same experiment as in \Cref{fig:app_synth_learned_1}, but now for $\beta = 1$, $c=100$, $\sigma = 0.25$.}}
    \label{fig:app_synth_learned_2}
\end{figure}

In our synthetic experiments, as expected, the values maximizing the marginal likelihood converge to the underlying true values of the data generating process as the sample size $N \to \infty$. By performing a visual investigation, we find that indeed the negativd marginal likelihood is a convex function in each argument, with a unique, well-defined minimum (see \Cref{fig:app_synth_learned_3,fig:app_synth_learned_4,fig:app_synth_learned_5}).

\begin{figure}
      \centering \includegraphics[width=\textwidth,height=\textheight,keepaspectratio]{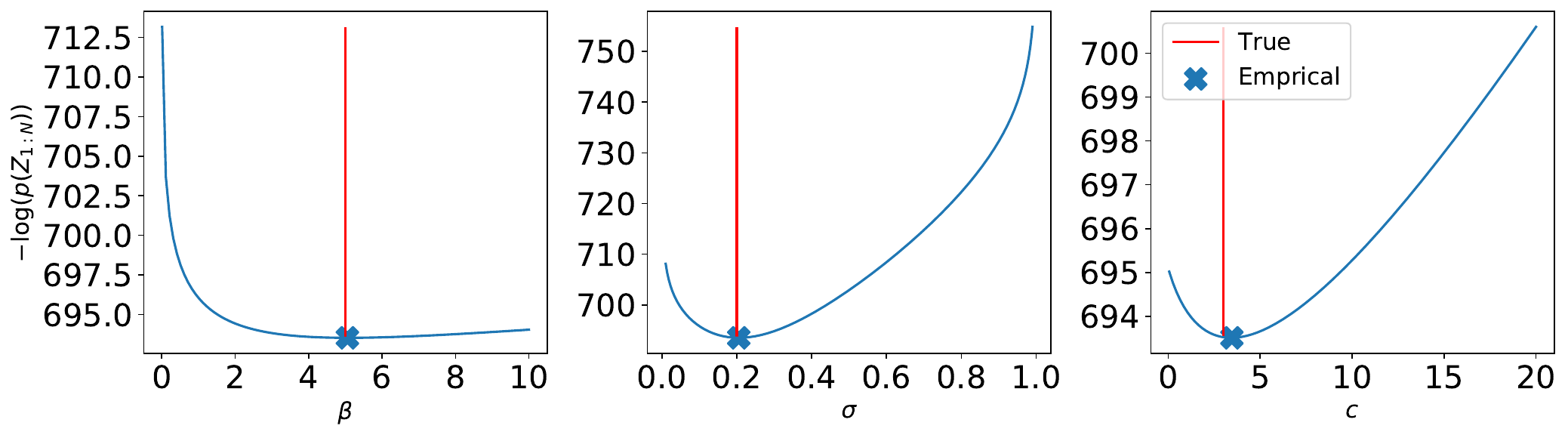}
    \caption{\footnotesize{We draw a synthetic dataset of size $N = 10{'}000$ from a SSB with parameters $ \beta = 5, \sigma = 0.1, c = 3$. In the left subplot, we plot the value of the negative marginal likelihood (vertical axis) as we vary the value of $\beta$ (horizontal axis), keeping $\sigma=0.1$ and $c=3$ fixed at the true value. We repeat the same procedure, now varying $\sigma$ and keeping $\beta=5$, $c=3$ fixed at their true value in the central subplot. Last, in the right subplot, we inspect the marginal likelihood as we vary the value of $c$, keeping $\beta=  5$, $\sigma = 0.1$. We then minimize numerically the negative log-likelihood, and report in each subplot with a blue cross the numerical value of the corresponding hyperparameter (horizontal axis) together with the corresponding marginal likelihood value (vertical axis).}}
    \label{fig:app_synth_learned_3}
\end{figure}

\begin{figure}
      \centering \includegraphics[width=\textwidth,height=\textheight,keepaspectratio]{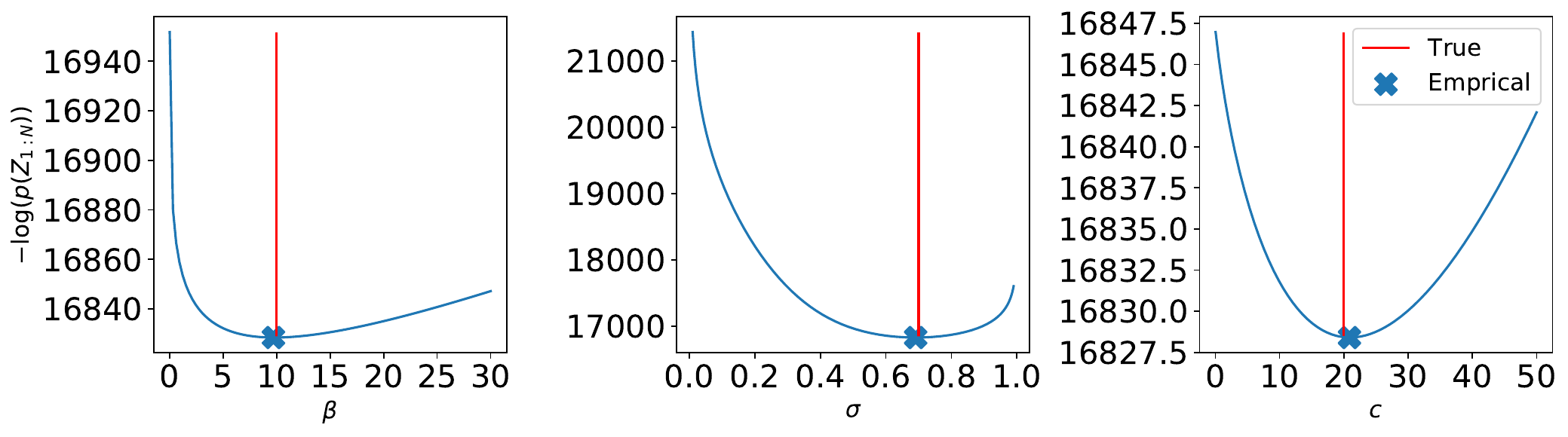}
    \caption{\footnotesize{We repeat the same exercise of \Cref{fig:app_synth_learned_3} for $N=1{'}000$, $\beta = 10$, $\sigma = 0.7$, $c = 20$.}}
    \label{fig:app_synth_learned_4}
\end{figure}

\begin{figure}
      \centering \includegraphics[width=\textwidth,height=\textheight,keepaspectratio]{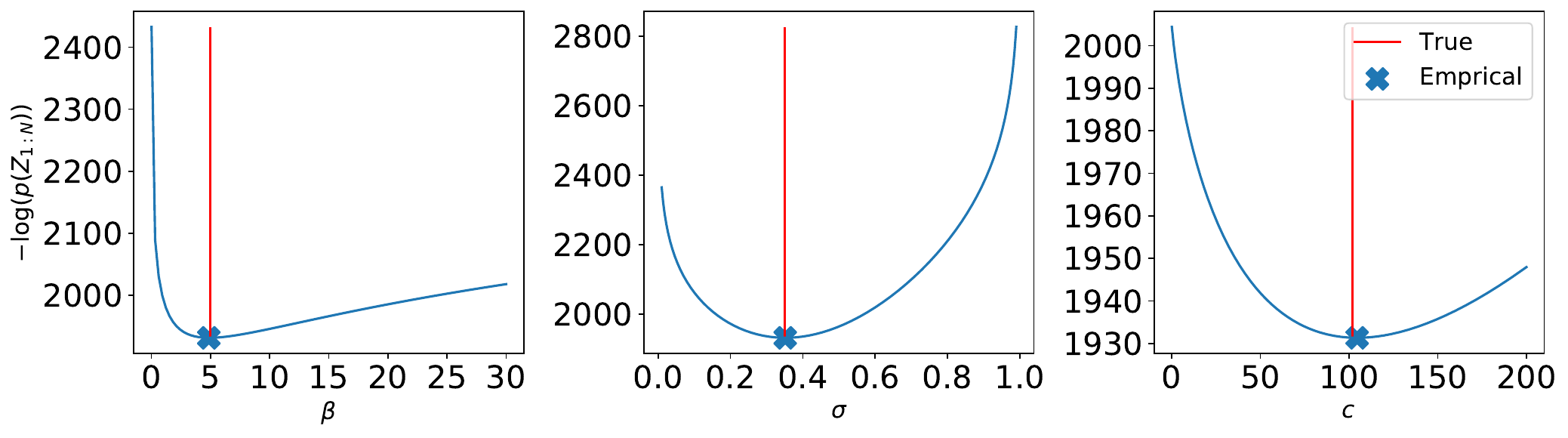}
    \caption{\footnotesize{We repeat the same exercise as in \Cref{fig:app_synth_learned_3,fig:app_synth_learned_4} for $N=100$, $\beta =5 $, $\sigma = 0.35$ and $c = 102$.}}
    \label{fig:app_synth_learned_5}
\end{figure}

When most of the features are very rare (e.g., they appear once or twice in the sample), we found that an alternative empirical Bayes approach, akin to the one adopted in \citet{masoero2019more}, worked better, as further discussed in \Cref{sec:app_synthetic_zipf}.

\section{Synthetic experiments from Zipf distributions} \label{sec:app_synthetic_zipf}
To compare the predictive performance of the SB-SP-Bernoulli process proposed in \Cref{sec:stable} to existing competing methods, we also consider synthetic data from Zipf-distributed frequencies (see \Cref{fig:app_zipf_freqs}). That is to say, we imagine that there exists a countable number of features in the population, and that, for some $\xi>0$, feature $k$ is observed independently of any other feature with probability $\pi_k = (k+1)^{-\xi}$. An observation $X_\ell$ then a binary vector, in which, conditionally on the frequencies $\pi = (\pi_1,\pi_2,\ldots)$ the $k$-th coordinate is a Bernoulli random variable:
\begin{align}
   X_{\ell,k} \mid \pi \overset{i.i.d.}{\sim} \Bern(\pi_k), \quad \pi = \{(k+1)^{-\xi}\}_{k\ge 1}. \label{eq:zipf_model}
\end{align}
We perform our simulations as follows: we fix a total sequencing capacity of $L = 2000$, and draw $L$ i.i.d.\ samples from the model, following the recipe given in \Cref{eq:zipf_model}. For simulation purposes, we only consider the first $K = 10^5$ features to have non-zero probability, i.e.\, $\pi_k = 0,$ for all $k > K$. 
\begin{figure}
      \centering \includegraphics[width=\textwidth,height=\textheight,keepaspectratio]{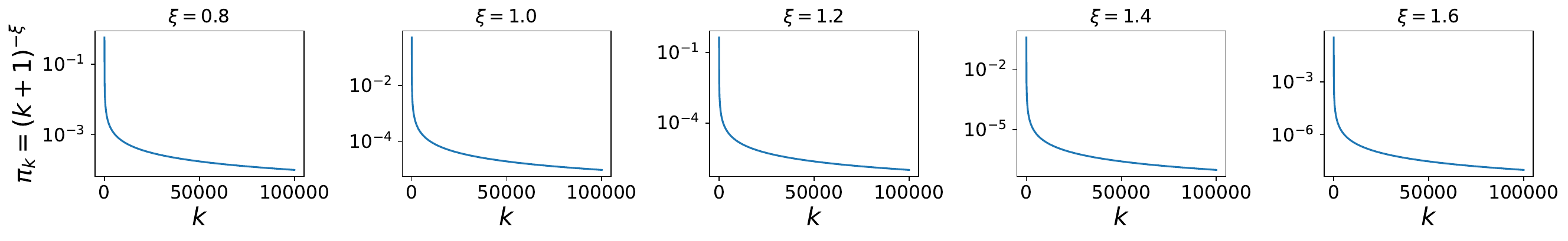}
    \caption{\footnotesize{Frequencies distribution for different choices of the parameter $\xi$.}}
    \label{fig:app_zipf_freqs}
\end{figure}
We compare the estimates of our proposed SB-SP-Bernoulli model (\Cref{sec:stable}), to the stable beta-Bernoulli process [3BP], the linear program of \citet{zou2016quantifying}, the first four orders of the Jackknife estimator originally proposed in \citet{burnham1978estimation} and recently employed in the genomics context by \citet{gravel2014predicting}, and the Good-Toulmin estimator, recently used in \citet{chakraborty2019somatic}, with the two alternative smoothing choices described in \citet{orlitsky2016optimal}. Estimates for Bayesian methods are obtained by using the posterior predictive mean for the number of new variants conditionally on the observed sample, with hyper-parameters learned by numerically maximizing the marginal distribution (EFPF) of the features counts, as described in \Cref{sec:learning_EFPF}.

As expected, we find the nonparametric Bayesian estimators to do particularly well for larger values of the exponent $\xi$ --- that is when most features are exceedingly rare. The SB-SP-Bernoulli and the SB-SP-Bernoulli-parameter beta-Bernoulli processes performed comparably on these datasets, both in terms of estimation accuracy and uncertainty quantification, as displayed in \Cref{fig:app_zipf_BNP_preds}.
\begin{figure}
      \centering \includegraphics[width=\textwidth,height=\textheight,keepaspectratio]{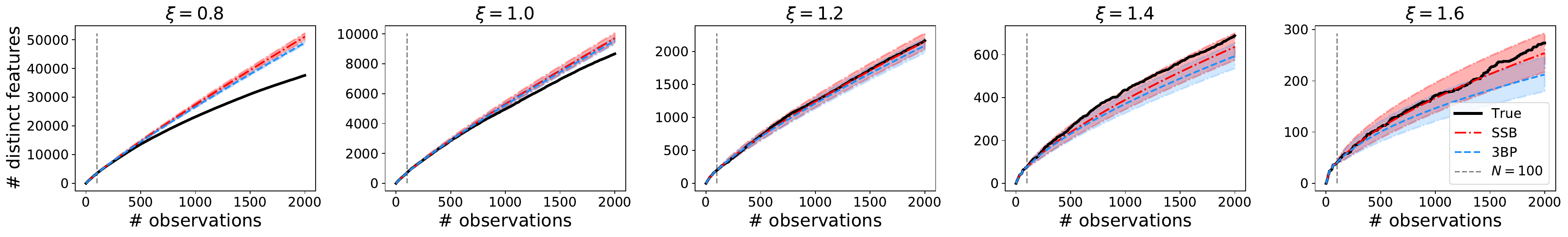}
    \caption{\footnotesize{Estimates for the number of new features for the SB-SP-Bernoulli (red) and the stable beta-Bernoulli (blue) processes as the exponent $\xi$ varies across subplots. Shaded regions cover a 95\% credible interval around the predictive mean. The solid black line represents the true counts. Here, the training is done using the first $N=100$ observations, and extrapolating up to the remaining $M=1900$ observations.}}
    \label{fig:app_zipf_BNP_preds}
\end{figure}

To better asses the predictive quality of the different methods, we ran extensive simulation experiments; for each value of $\xi \in \{0.8, 1, 1.2, 1.4, 1.6\}$, we generated $S = 100$ datasets of size $L=2000$, and for each value of $N \in \{ 10, 50, 100, 200 \}$ we trained each method, and extrapolated to predict the number of new variants to be observed up to $M = L-N \in \{1990, 1950, 1900, 1800\}$ remaining samples. We report as measure of accuracy the percentage accuracy incurred by each estimation method $v_{N,a}^{(M)}$, defined in \Cref{eq:percentage_accuracy}, at the largest extrapolation level $M = L-N$, across different values of $N$ and all $S = 100$ simulation studies. Results are reported via boxplots in \Cref{fig:app_zipf_all_methods_1,fig:app_zipf_all_methods_2,fig:app_zipf_all_methods_3}. While all methods improve their performance with larger sample sizes, we find that the BNP estimators (SSP, 3BP) provide relatively more accurate results for smaller samples sizes (e.g., $N=10, N=50$ in \Cref{fig:app_zipf_all_methods_1,fig:app_zipf_all_methods_2}). The performance of the BNP methods exceed those of competing methods for larger values of the exponent ($\xi \in \{ 1.2, 1.4, 1.6 \}$), while higher order Jackknife and linear programs tend to do better for smaller values of the exponent ($\xi \in \{ 0.8, 1\}$).
\begin{figure}
      \centering \includegraphics[width=\textwidth,height=\textheight,keepaspectratio]{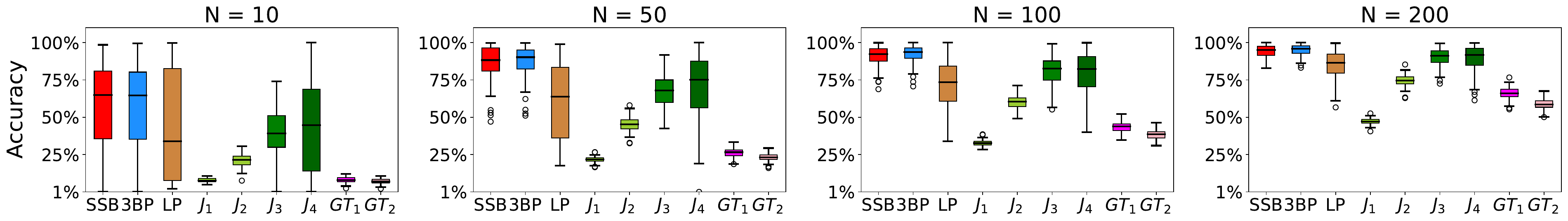}
    \caption{\footnotesize{Accuracy of the competing methods (SB-SP-Bernoulli [SSP], stable beta-Bernoulli [3BP], Jackknife [J], linear program [LP], Good-Toulmin [GT]) on simulated data from a Zipf model (\Cref{eq:zipf_model}) with parameter $\xi = 1.2$.  For $L = 2000$, we report  $v_{N,a}^{(M)}$ as $N$ increases, for $M=L-N$. For each $N$, results across $S = 100$ datasets are reported in the boxplots.}}
    \label{fig:app_zipf_all_methods_1}
\end{figure}
\begin{figure}
      \centering \includegraphics[width=\textwidth,height=\textheight,keepaspectratio]{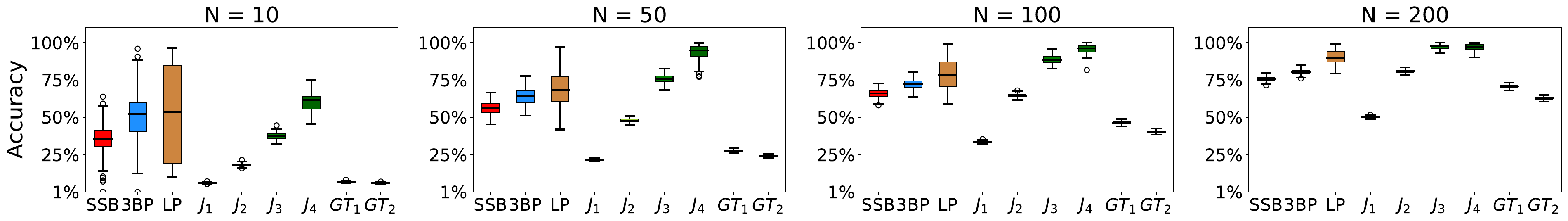}
    \caption{\footnotesize{Accuracy of the competing methods (SB-SP-Bernoulli [SSP], stable beta-Bernoulli [3BP], Jackknife [J], linear program [LP], Good-Toulmin [GT]) on simulated data from a Zipf model (\Cref{eq:zipf_model}) with parameter $\xi = 0.8$.  For $L = 2000$, we report  $v_{N,a}^{(M)}$ as $N$ increases, for $M=L-N$. For each $N$, results across $S = 100$ datasets are reported in the boxplots.}}
    \label{fig:app_zipf_all_methods_2}
\end{figure}
\begin{figure}
      \centering \includegraphics[width=\textwidth,height=\textheight,keepaspectratio]{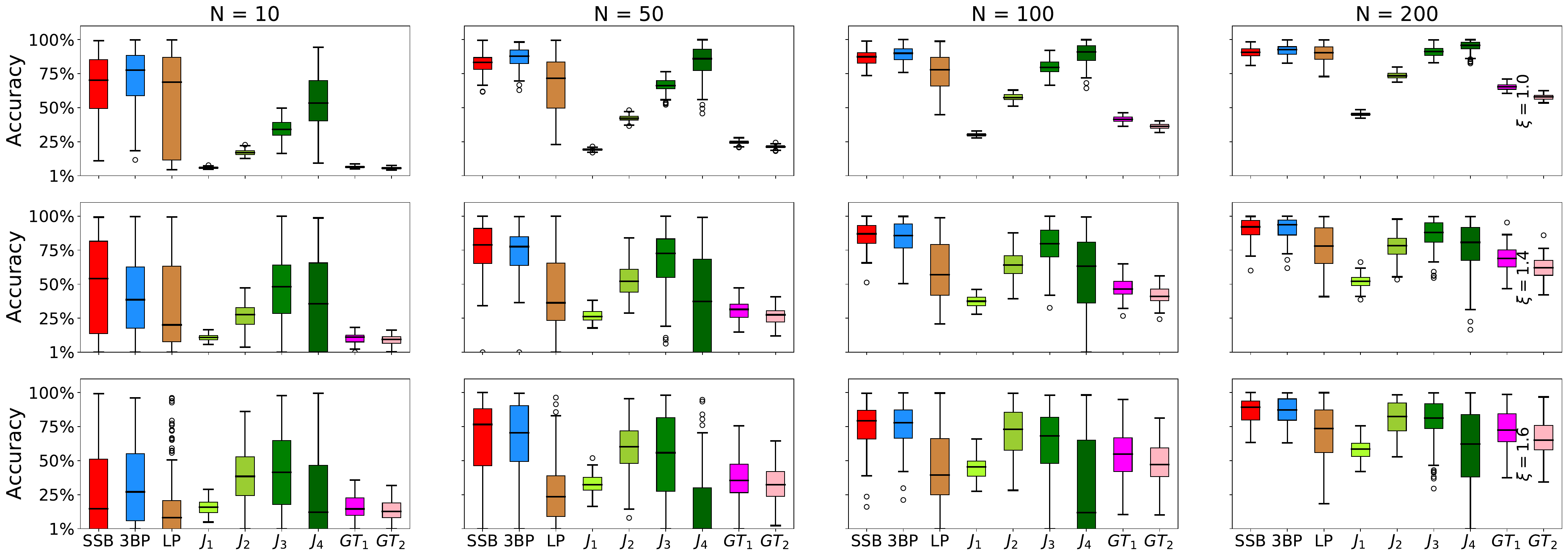}
    \caption{\footnotesize{Accuracy of the competing methods (SB-SP-Bernoulli [SSP], stable beta-Bernoulli [3BP], Jackknife [J], linear program [LP], Good-Toulmin [GT]) on simulated data from a Zipf model (\Cref{eq:zipf_model}) with parameter $\xi \in \{1, 1.4, 1.6\}$ (top row, center row, bottom row).  For $L = 2000$, we report  $v_{N,a}^{(M)}$ as $N$ increases, for $M=L-N$. For each $N$, results across $S = 100$ datasets are reported in the boxplots.}}
    \label{fig:app_zipf_all_methods_3}
\end{figure}

\section{Additional experiments on the gnomAD dataset} \label{sec:gnomAD}

\subsection{Experimental setup}
\label{sec:exp_setup_gnomAD}

In order to run our experiments, we use data from the gnomAD (genome aggregation dataset) discovery project \citep{karczewski2019variation}, the largest and most comprehensive publicly available human genome dataset. We follow the same experimental setup adopted in \citet{masoero2019more}. We briefly summarize this setup in this section. The gnomAD dataset contains 125{'}748 exomes sequences (i.e.\, protein-coding regions of the genome), from 8 main populations. Sample size varies widely across sub populations, e.g.\ the ``Other'' subgroup counts about 3{'}000 observations, while ``South Easy Asian'' contains almost 16{'}000 individuals (see \citet{karczewski2019variation} for additional details).

For privacy reasons not all individual sequences are accessible. Hence, in order to run our analysis we generate synthetic data which closely resembles the true data as follows. For every subpopulation with $N$ individuals and every position $j=1,\ldots,K$ in the exome, we have access to the total number of individuals $N_j$ showing variation at position $j$. We compute the empirical frequency of variation at site $j$, $\hat{\theta}_j := N_j/N$ for all $j=1,\ldots,K$. Our data is then generated by sampling independent Bernoulli random vectors $X_1,\ldots,X_N$, with $X_n = [x_{n,1},\ldots,x_{n,K}]$. The entries in the vector are independent Bernoulli random variables, $x_{n,j} \sim \Bern(\hat{\theta}_j)$. 

\subsection{Results from the gnomAD data}

For each of eight subpopulations in the data, we performed the following experiment. Let  $\hat{\theta}=[\hat{\theta}_1,\ldots,\hat{\theta}_{K_{\max}}] \subseteq[0,1]$ denote the ``genetic signature'' of the population, with $\hat{\theta}_k=N_k/N$, with $N_{tot}$ the total number of individuals in the population and $N_k$ the number of individuals in the population displaying such variant, $1\le N_k \le N_{tot}$. Then, for each population, we generate $S=50$ datasets by drawing $N_{tot}$ i.i.d.\ binary random vectors of length $K_{\max}$ as described above, with biases given by $\hat{\theta}$. We then retain for each dataset $N\in\{50,100\}$ observations for training, and try to predict the number of new variants that are going to be observed if we were to sample additional $M = N_{tot} - N$ observations.

In a nutshell, also on this data, the findings are similar to the results obtained on the MSK-IMPACT cancer data. In particular, we find that when the sample size $N$ is small, the proposed SB-SP Bernoulli model leads to predictions that are often comparable or more accurate than competing methods. 

First, we report the accuracy metric $\predaccuracy{N}{M}$ for eight subpopulations in gnomAD, Afroamerican (Amr.), South East Asian (SE. As.), Other East Asian (Ot. E. As.), Finnish (Fin.), South European (S. Eu.), Swedish (Swe.), South Asian (S. As.) and the remaining Other. In \Cref{fig:app_gnomAD_accuracy_50} we show results (over $S=50$ Monte-Carlo re-draws of the data from the estimated frequencies $\hat{\theta}$) of retaining $N = 50$ datapoints for training, and extrapolating to the largest available sample size $M$. In \Cref{fig:app_gnomAD_accuracy_100} we report results for the same metric, with training performed by retaining $N=100$ datapoints.

\begin{figure}
    \centering
      \centering \includegraphics[width=\textwidth,height=\textheight,keepaspectratio]{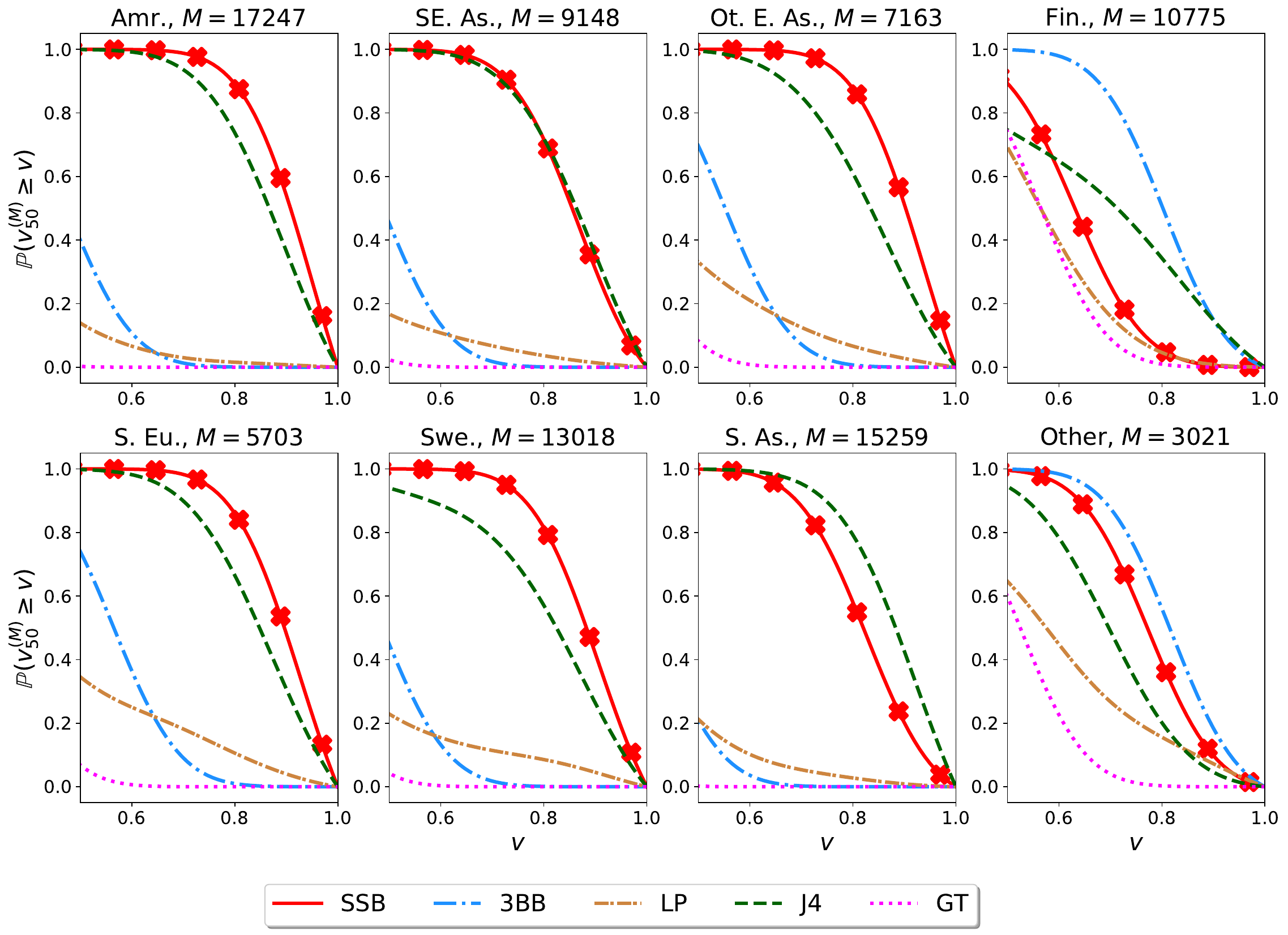}
    \caption{\footnotesize{Accuracy metric $\predaccuracy{50}{M}$ for eight subpopulations in the gnomAD dataset. For each subpopulation we retain $N=50$ observations for training, and extrapolate to the largest possible value $M$. Results are over $S=50$ Monte-Carlo draws of the data, as described in \Cref{sec:exp_setup_gnomAD}}}
    \label{fig:app_gnomAD_accuracy_50}
\end{figure}

\begin{figure}
    \centering
      \centering \includegraphics[width=\textwidth,height=\textheight,keepaspectratio]{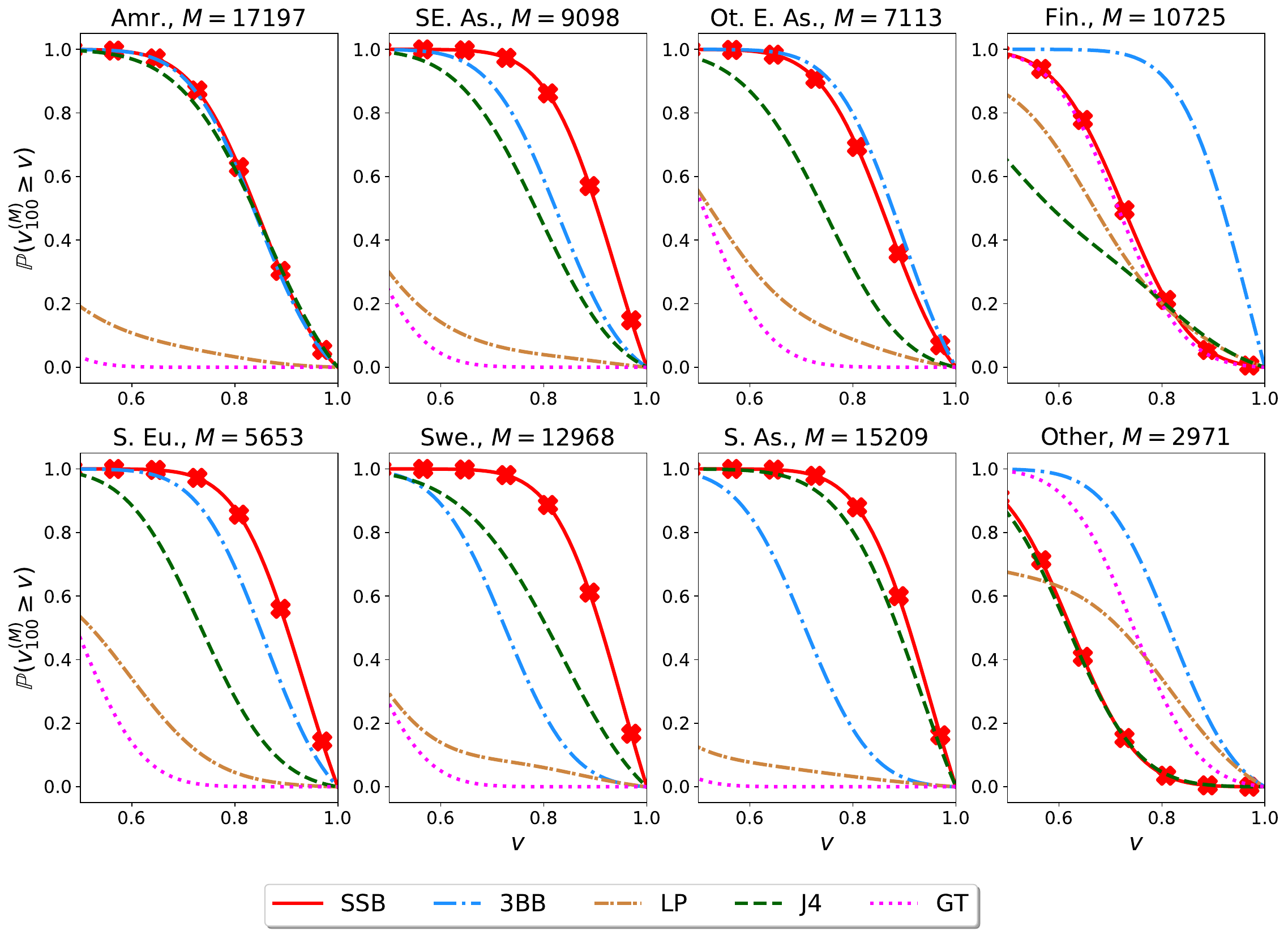}
    \caption{\footnotesize{Same setup as in \Cref{fig:app_gnomAD_accuracy_50}, now for $N=100$.}}
    \label{fig:app_gnomAD_accuracy_100}
\end{figure}

Next, we provider boxplots that report the (aggregated) accuracy of the metric $\predaccuracy{N}{M}$ across all the eight populations, and all the $S=50$ Monte-Carlo draws (so that each boxplot reports the accuracy of a total of $50 \times 8 = 400$ accuracy values), for $N=50$ (\Cref{fig:app_gnomAD_boxplot_50}) as well as $N=100$ (\Cref{fig:app_gnomAD_boxplot_100}).

\begin{figure}
    \centering
      \centering \includegraphics[width=\textwidth,height=\textheight,keepaspectratio]{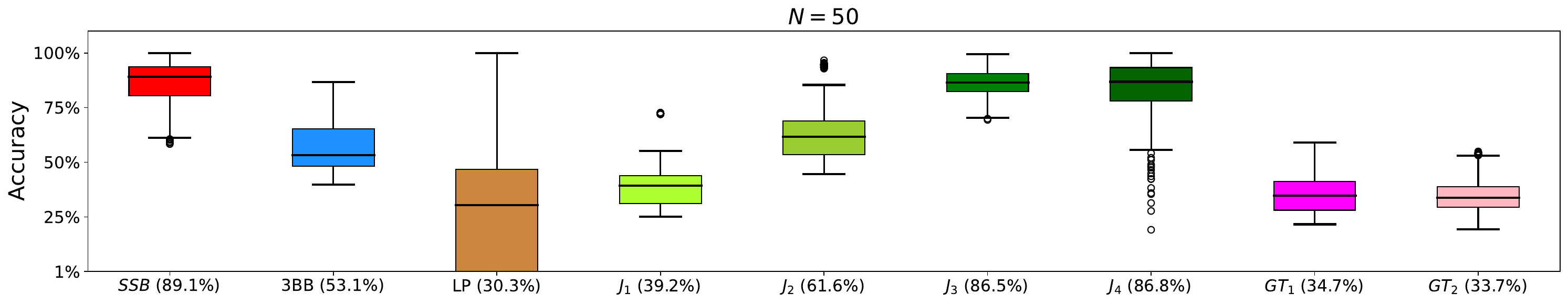}
    \caption{\footnotesize{Accuracy of the compared methods, now over all the eight subpopulations and over $50$ Monte Carlo draws for each population. $N=50$, and $M$ is set to be the largest possible extrapolation size for each subpopulation.}}
    \label{fig:app_gnomAD_boxplot_50}
\end{figure}

\begin{figure}
    \centering
      \centering \includegraphics[width=\textwidth,height=\textheight,keepaspectratio]{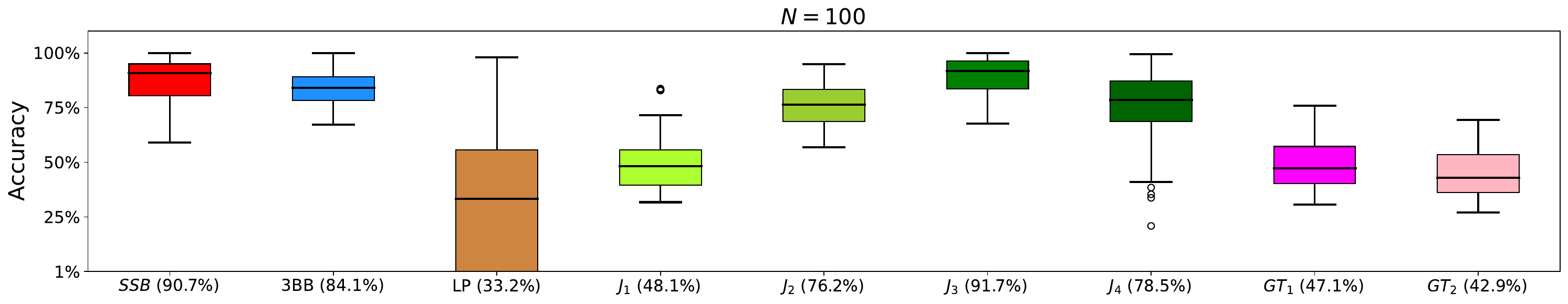}
    \caption{\footnotesize{Same setup as in \Cref{fig:app_gnomAD_boxplot_50}, now for $N=100$.}}
    \label{fig:app_gnomAD_boxplot_100}
\end{figure}

\clearpage

\subsection{Additional boxplots}

Since in \Cref{fig:app_gnomAD_boxplot_50} and \Cref{fig:app_gnomAD_boxplot_100} we are aggregating result in which $N$ is consistent for all populations, but $M$ differs, we also report boxplots of each subpopulation individually.

\begin{figure}
    \centering
      \centering \includegraphics[width=\textwidth,height=\textheight,keepaspectratio]{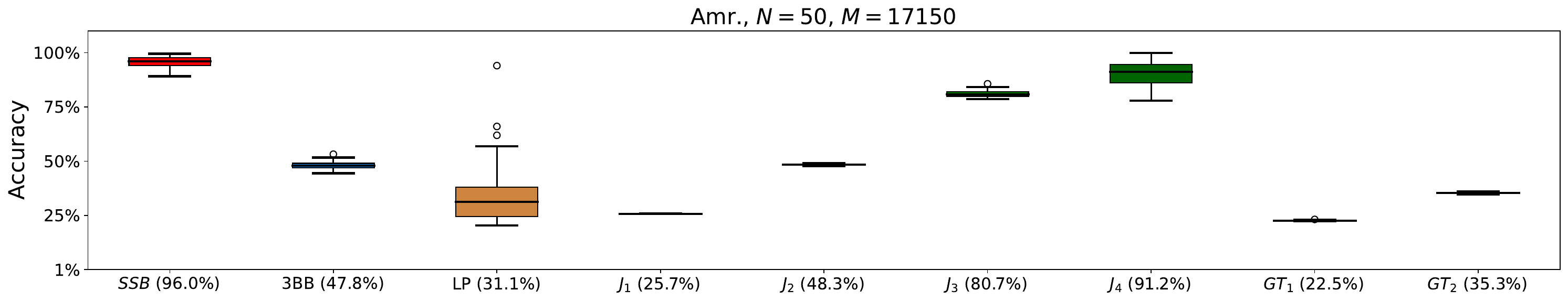}
    \caption{\footnotesize{Same setup as in \Cref{fig:app_gnomAD_boxplot_50}, but only for the American subpopulation.}}
\end{figure}

\begin{figure}
    \centering
      \centering \includegraphics[width=\textwidth,height=\textheight,keepaspectratio]{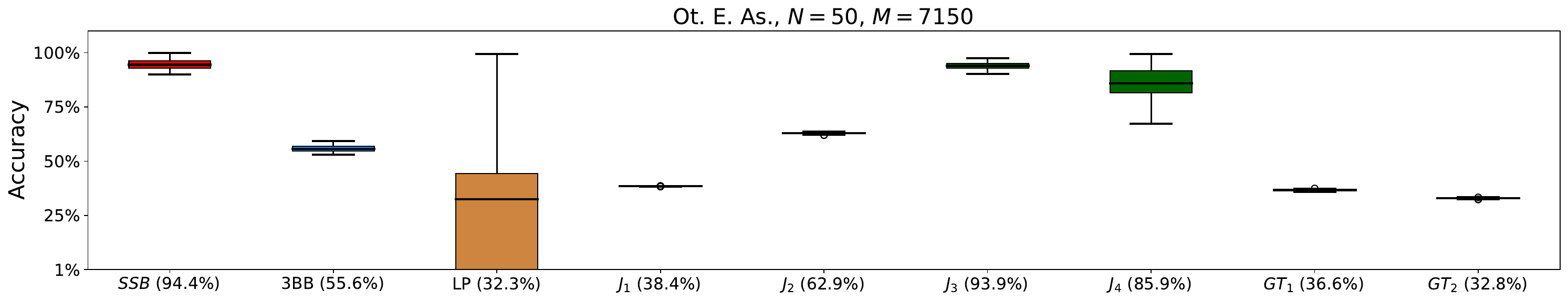}
    \caption{\footnotesize{Same setup as in \Cref{fig:app_gnomAD_boxplot_50}, but only for the Other East Asian subpopulation.}}
\end{figure}

\begin{figure}
    \centering
      \centering \includegraphics[width=\textwidth,height=\textheight,keepaspectratio]{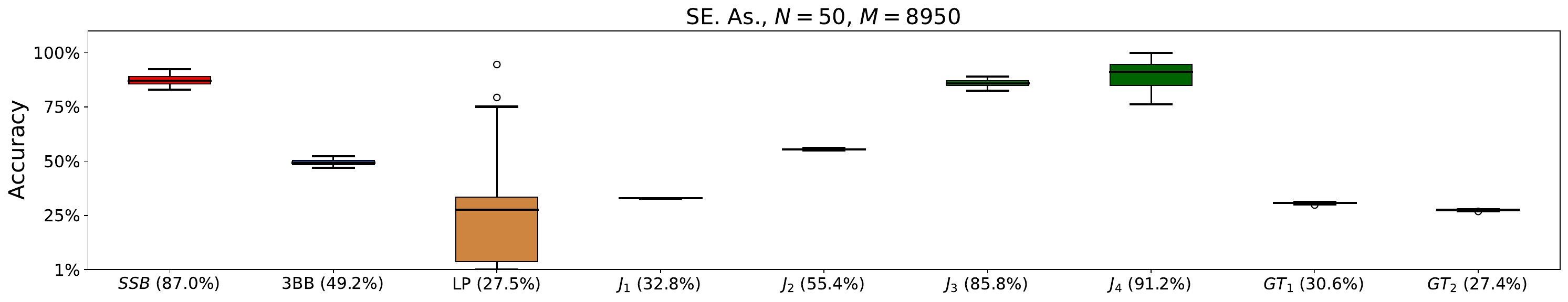}
    \caption{\footnotesize{Same setup as in \Cref{fig:app_gnomAD_boxplot_50}, but only for the East Asian subpopulation.}}
\end{figure}

\begin{figure}
    \centering
      \centering \includegraphics[width=\textwidth,height=\textheight,keepaspectratio]{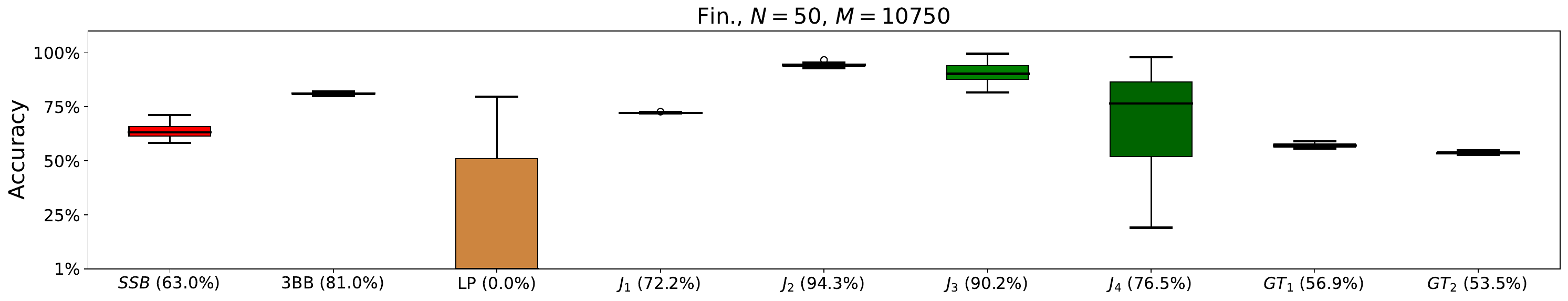}
    \caption{\footnotesize{Same setup as in \Cref{fig:app_gnomAD_boxplot_50}, but only for the Finnish subpopulation.}}
\end{figure}

\begin{figure}
    \centering
      \centering \includegraphics[width=\textwidth,height=\textheight,keepaspectratio]{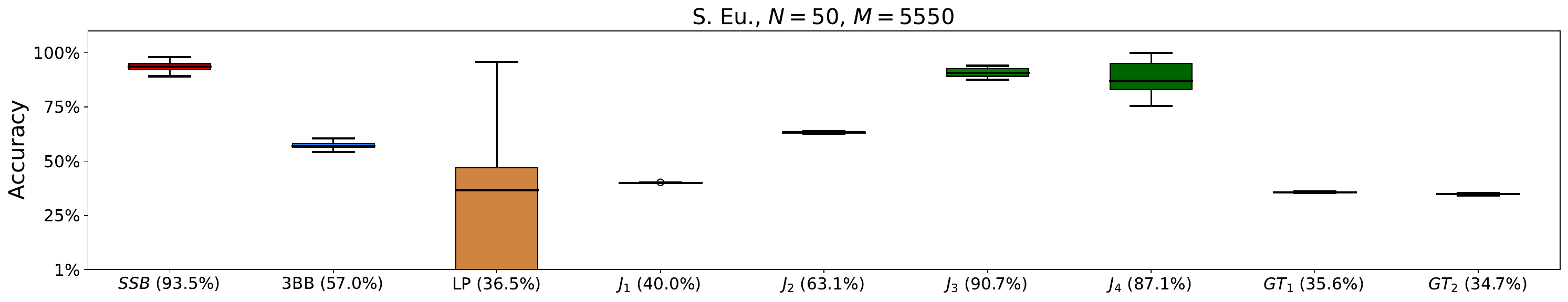}
    \caption{\footnotesize{Same setup as in \Cref{fig:app_gnomAD_boxplot_50}, but only for the Southern European subpopulation.}}
\end{figure}

\begin{figure}
    \centering
      \centering \includegraphics[width=\textwidth,height=\textheight,keepaspectratio]{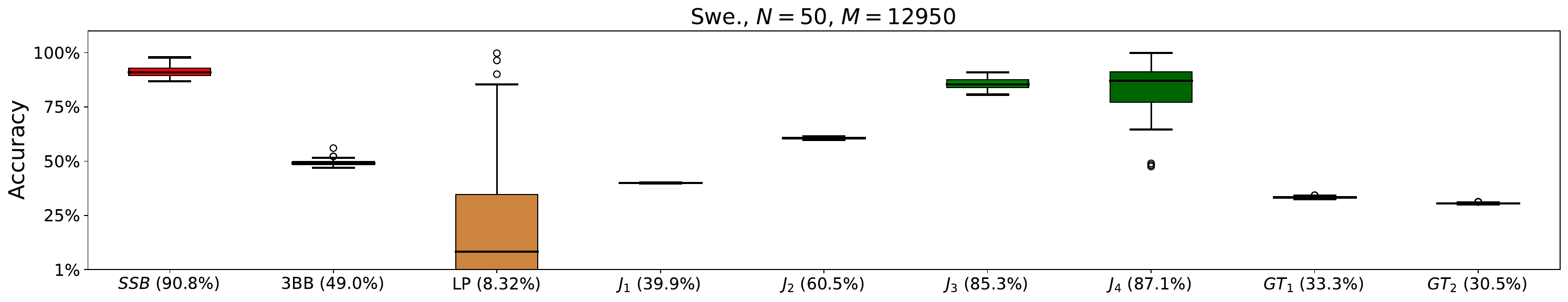}
    \caption{\footnotesize{Same setup as in \Cref{fig:app_gnomAD_boxplot_50}, but only for the Swedish subpopulation.}}
\end{figure}

\begin{figure}
    \centering
      \centering \includegraphics[width=\textwidth,height=\textheight,keepaspectratio]{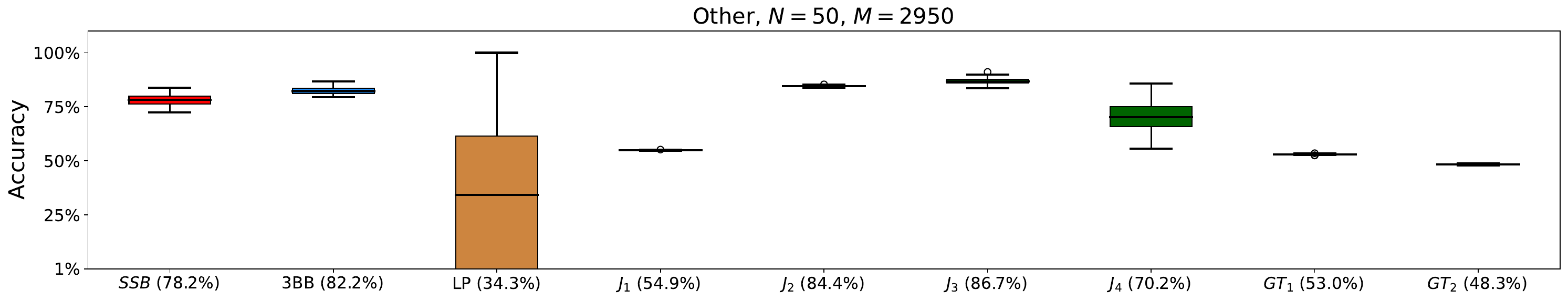}
    \caption{\footnotesize{Same setup as in \Cref{fig:app_gnomAD_boxplot_50}, but only for the ``Other'' subpopulation.}}
\end{figure}

\begin{figure}
    \centering
      \centering \includegraphics[width=\textwidth,height=\textheight,keepaspectratio]{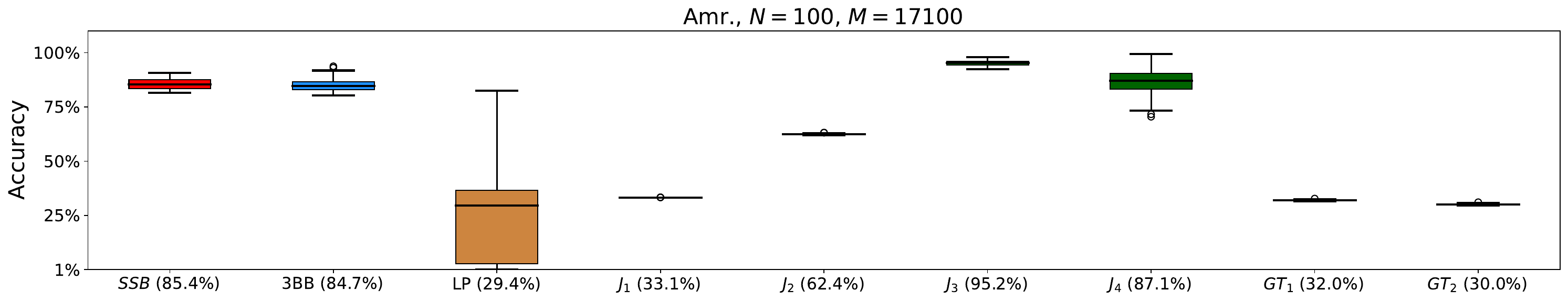}
    \caption{\footnotesize{Same setup as in \Cref{fig:app_gnomAD_boxplot_100}, but only for the American subpopulation.}}
\end{figure}

\begin{figure}
    \centering
      \centering \includegraphics[width=\textwidth,height=\textheight,keepaspectratio]{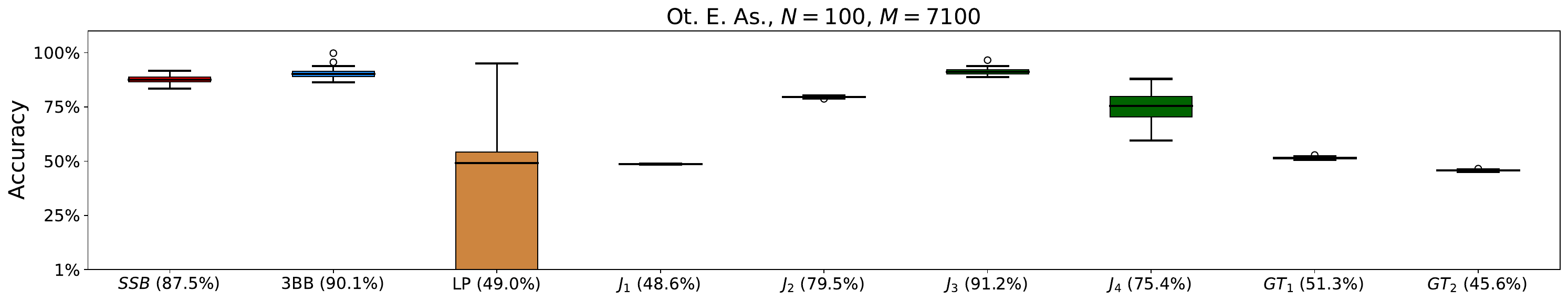}
    \caption{\footnotesize{Same setup as in \Cref{fig:app_gnomAD_boxplot_100}, but only for the Other East Asian subpopulation.}}
\end{figure}

\begin{figure}
    \centering
      \centering \includegraphics[width=\textwidth,height=\textheight,keepaspectratio]{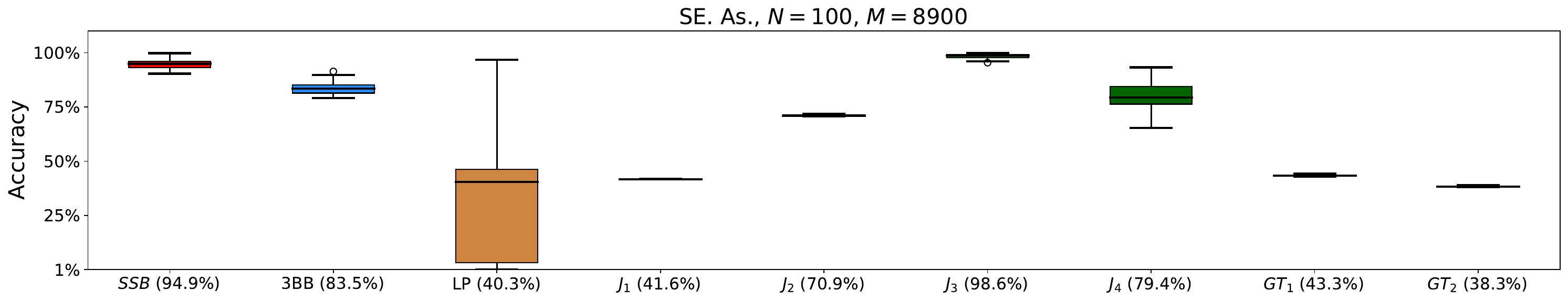}
    \caption{\footnotesize{Same setup as in \Cref{fig:app_gnomAD_boxplot_100}, but only for the East Asian subpopulation.}}
\end{figure}

\begin{figure}
    \centering
      \centering \includegraphics[width=\textwidth,height=\textheight,keepaspectratio]{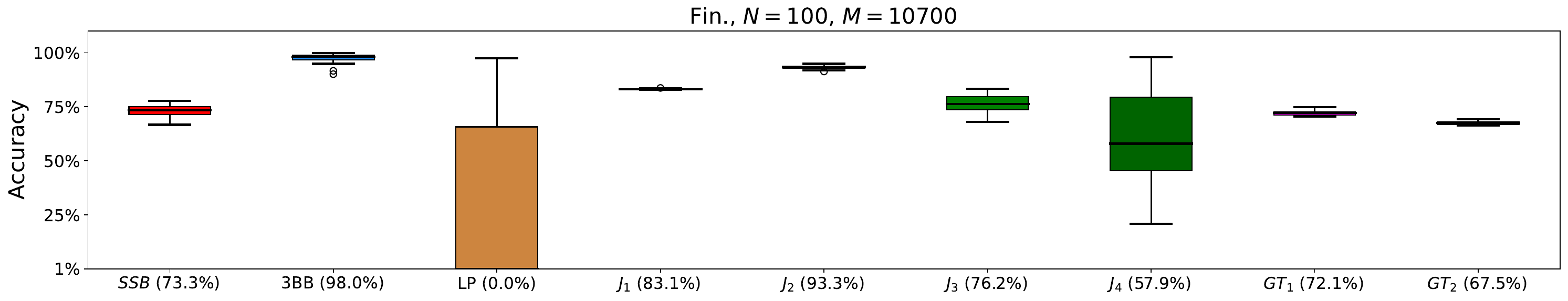}
    \caption{\footnotesize{Same setup as in \Cref{fig:app_gnomAD_boxplot_100}, but only for the Finnish subpopulation.}}
\end{figure}

\begin{figure}
    \centering
      \centering \includegraphics[width=\textwidth,height=\textheight,keepaspectratio]{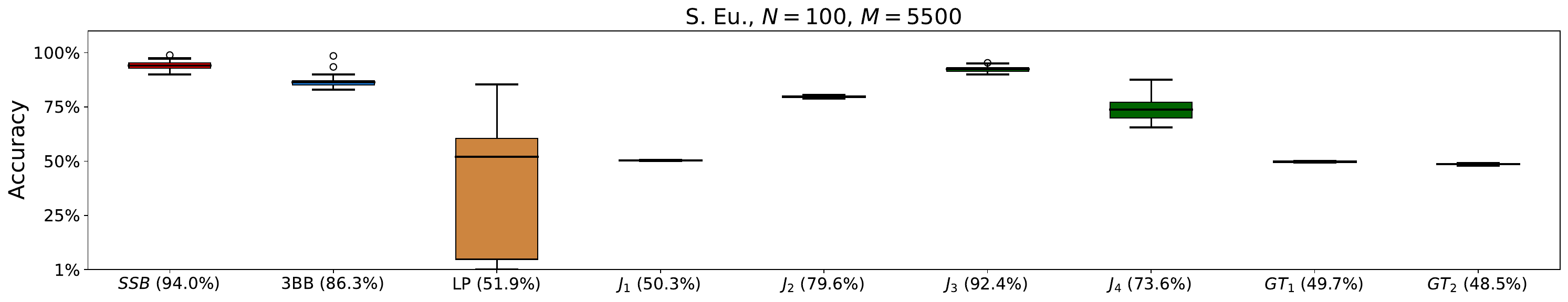}
    \caption{\footnotesize{Same setup as in \Cref{fig:app_gnomAD_boxplot_100}, but only for the Southern European subpopulation.}}
\end{figure}

\begin{figure}
    \centering
      \centering \includegraphics[width=\textwidth,height=\textheight,keepaspectratio]{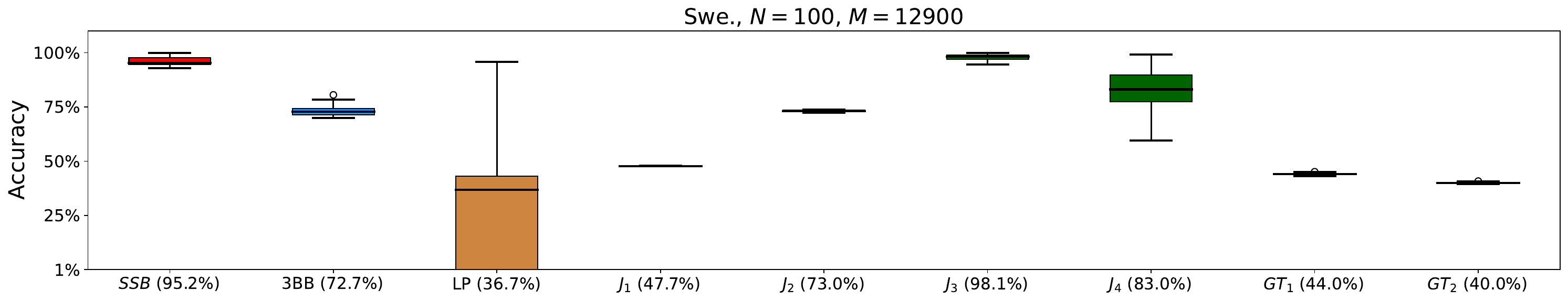}
    \caption{\footnotesize{Same setup as in \Cref{fig:app_gnomAD_boxplot_100}, but only for the Swedish subpopulation.}}
\end{figure}

\begin{figure}
    \centering
      \centering \includegraphics[width=\textwidth,height=\textheight,keepaspectratio]{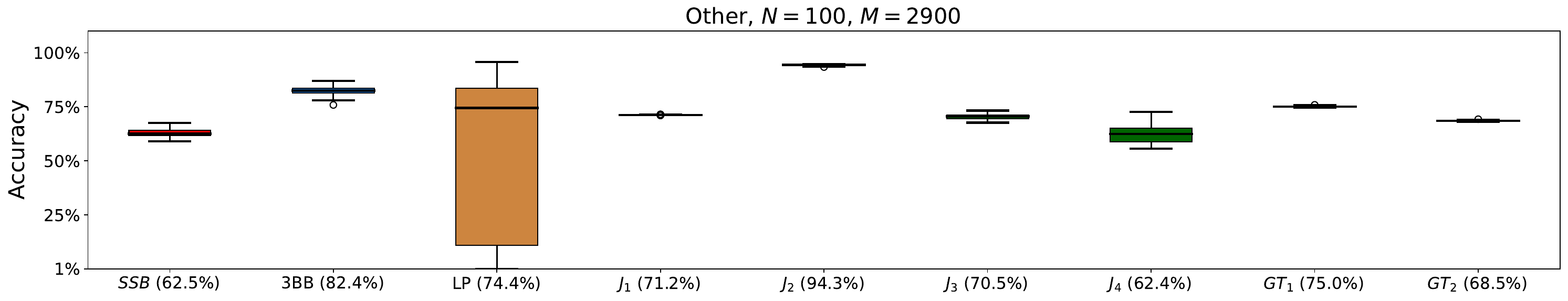}
    \caption{\footnotesize{Same setup as in \Cref{fig:app_gnomAD_boxplot_100}, but only for the ``Other'' subpopulation.}}
\end{figure}

\clearpage

\bibliographystyle{abbrvnat}
\bibliography{references}

\end{document}